%% file: Recovering_quantum_gates.tex
\documentclass[aps,pra,twocolumn,10pt,superscriptaddress]{revtex4-1}
\pdfoutput=1

\usepackage[T1]{fontenc} 
\usepackage[english]{babel} 
\usepackage[utf8]{inputenc}

 \usepackage{times}

\usepackage[table,usenames,dvipsnames]{xcolor}
\usepackage[
	hyperindex, 
	breaklinks, 
	colorlinks=true,
	hypertexnames=false
						]{hyperref}
\PassOptionsToPackage{linktocpage}{hyperref} 
\usepackage{url}

\hypersetup{
  colorlinks   = true, 
  urlcolor     = magenta, 
  linkcolor    = blue, 
  citecolor   = green 
}
\usepackage{bookmark} 

\usepackage{graphicx}
\usepackage[caption=false, justification=justified]{subfig}

\usepackage{amsthm} 
\usepackage{amssymb}
\usepackage{amsmath,mathtools}
\usepackage{dsfont}
\usepackage[vcentermath]{youngtab}
\usepackage{comment}
\usepackage{enumitem}

\usepackage{tikz}

\newtheorem{theorem}{Theorem}
\newtheorem{corollary}[theorem]{Corollary}
\newtheorem{lemma}[theorem]{Lemma}
\newtheorem{proposition}[theorem]{Proposition}
\newtheorem{definition}[theorem]{Definition}

\newcommand{\myleft}{\mathopen{}\mathclose\bgroup\left}
\newcommand{\myright}{\aftergroup\egroup\right}
\newcommand{\e}{\ensuremath\mathrm{e}}

\newcommand{\rmd}{\ensuremath\mathrm{d}}
\DeclareMathOperator{\Tr}{Tr}

\DeclareMathOperator{\Id}{Id}
\DeclareMathOperator*{\argmin}{arg\,min}

\newcommand{\fro}{\mathrm{F}}

\DeclareMathOperator{\Comm}{Comm}
\DeclareMathOperator{\End}{L}
\DeclareMathOperator{\Lin}{L}

\DeclareMathOperator{\Pos}{Pos}

\DeclareMathOperator{\CP}{CP}

\DeclareMathOperator{\U}{U}
\DeclareMathOperator{\Cl}{Cl}

\renewcommand{\fro}{2}
\newcommand{\utpn}{\text{\rm u,tp,0}}

\DeclareMathOperator\tr{Tr}
\DeclareMathOperator\Cliff{Cl}
\DeclareMathOperator\GL{GL}

\newcommand{\CC}{\mathbb{C}}
\newcommand{\RR}{\mathbb{R}}

\newcommand{\ZZ}{\mathbb{Z}}
\newcommand{\NN}{\mathbb{N}}
\newcommand{\FF}{\mathbb{F}}
\newcommand{\1}{\mathds{1}}
\newcommand{\EE}{\mathbb{E}}
\newcommand{\PP}{\mathbb{P}}

\newcommand{\mc}[1]{\mathcal{#1}}

\newcommand{\landauO}{\mc{O}}
\newcommand{\DM}{\mc{S}}
\newcommand{\meas}{\mc{M}}

\newcommand{\argdot}{{\,\cdot\,}}

\newcommand{\ad}{^\dagger}

\newcommand{\norm}[1]{\left\Vert #1 \right\Vert} 
\newcommand{\normB}[1]{\Bigl\Vert #1 \Bigr\Vert} 
\newcommand{\snorm}[1]{\norm{#1}_\infty} 

\newcommand{\tnorm}[1]{\norm{#1}_{1}} 

\newcommand{\TrNorm}[1]{\norm{#1}_{1}} 

\newcommand{\fnorm}[1]{\norm{#1}_\fro} 

\newcommand{\fnormB}[1]{\normB{#1}_\fro}

\newcommand{\lTwoNorm}[1]{\norm{#1}_{\ell_2}} 

\newcommand{\ket}[1]{\left.\left|{#1}\right.\right\rangle}

\newcommand{\bra}[1]{\left.\left\langle{#1}\right.\right|}

\newcommand{\ketbra}[2]{\ket{#1} \!\! \bra{#2}}

\newcommand{\sandwich}[3]
  {\left\langle  #1 \right| #2 \left| #3 \right\rangle}

\renewcommand{\Pr}{\operatorname{\PP}} 

\newcommand{\eps}{\mathrm{eps}}
\newcommand{\utp}{\text{\rm u,tp}}
\newcommand{\Favg}{F_\text{avg}}

\newcommand{\AGF}{AGF}
\newcommand{\AGFs}{AGFs}
\newcommand{\recerror}{\ensuremath\varepsilon_{\text{\rm rec}}}
\newcommand{\recerrorTr}{\ensuremath\varepsilon}
\newcommand{\recerrorF}{\ensuremath\varepsilon_{\text{\rm F}}}

\makeatletter 
 \hypersetup{pdftitle = {Recovering quantum gates from few average gate fidelities},
	     pdfauthor = {Ingo Roth, Richard Kueng, Shelby Kimmel, Yi-Kai Liu, David Gross, Jens Eisert, Martin Kliesch},
	     pdfsubject = {Compressed sensing, quantum physics},
 	     pdfkeywords = {Compressive sensing, compressed sensing,  
			    quantum, process, channel, tomography,  
			    convex relaxation, 
			    low-rank matrix recovery, 
			    randomized benchmarking, 
			    null space property, 
			    minimum, conic singular value, 
			    recovery guarantee,
			    Weingarten, function,
			    integration, unitary group,
			    Clifford group,
			    sampling complexity,
			    upper, lower, bound,
			    average, gate, fidelity
			    }
	    }
\makeatother

\newcommand{\fu}{Dahlem Center for Complex Quantum Systems, Freie Universit\"{a}t Berlin, Germany}
\newcommand{\ug}{Institute of Theoretical Physics and Astrophysics, National Quantum Information Centre, University of Gda\'{n}sk, Poland}
\newcommand{\hhu}{
	Institute for Theoretical Physics,
	Heinrich Heine University D{\"u}sseldorf, 
	Germany
}
\newcommand{\ca}{Institute for Quantum Information and Matter, California Institute of Technology, Pasadena, USA}
\newcommand{\ma}{National Institute of Standards and Technology, Gaithersburg, USA}
\newcommand{\maa}{Joint Center for Quantum Information and Computer Science (QuICS), University of Maryland, College Park, USA}
\newcommand{\mb}{Department of Computer Science, Middlebury College, USA}
\newcommand{\uc}{Institute for Theoretical Physics, University of Cologne, Germany}

\newcommand{\detailsandproof}{Section~\ref{sec:details}}
\newcommand{\refdetailsandproof}{Section}
\newcommand{\suppfootnote}{}

\begin{document}

\title{Recovering quantum gates from few average gate fidelities}

\author{I.\ Roth}
\email[Corresponding author: ]{i.roth@fu-berlin.de}
\affiliation{\fu}

\author{R.\ Kueng}
\affiliation{\ca}

\author{S.\ Kimmel}
\affiliation{\mb}

\author{Y.-K.\ Liu}
\affiliation{\ma}
\affiliation{\maa}

\author{D.\ Gross}
\affiliation{\uc}

\author{J.\ Eisert}
\affiliation{\fu}

\author{M.\ Kliesch}
\affiliation{\ug}
\affiliation{\hhu}

\begin{abstract}
	Characterising quantum processes is a key task in the development of quantum technologies, especially at the noisy intermediate scale of today's devices. One method for characterising processes is randomised benchmarking, which is robust against state preparation and measurement (SPAM) errors, and
	can be used to benchmark Clifford gates. A complementing approach asks for full tomographic knowledge.
	Compressed sensing techniques achieve full tomography of quantum channels essentially at 
	optimal resource efficiency.
	So far, guarantees for compressed sensing protocols rely on unstructured random measurements and can not be applied to the data acquired from randomised benchmarking experiments. It has been an open question whether or not the favourable features of both worlds can be combined. 
	In this work, we give a positive answer to this question.
	For the important case of characterising multi-qubit 
	unitary gates, we provide a rigorously guaranteed and practical reconstruction method that works with an essentially optimal number of average gate fidelities measured respect to random Clifford unitaries.
	Moreover, for general unital quantum channels we provide an explicit expansion into a unitary 2-design,  
	allowing for a practical and guaranteed reconstruction also in that case. 
	As a side result, we obtain a new statistical interpretation of the unitarity -- a figure of merit that characterises the coherence of a 
	process.
	In our proofs we exploit recent representation theoretic insights on the Clifford group, 
	develop a version of Collins' calculus with Weingarten functions for integration over the Clifford group, 
	and combine this with proof techniques from compressed sensing.
\end{abstract}


\maketitle

\small
\tableofcontents
\normalsize
\break


\section{Introduction}
\vspace{-.2em}
As increasingly large and complex quantum devices are being built and the development of fault tolerant quantum computation is moving forward, it is critical to develop tools to refine our control of these devices. 
For this purpose, several improved methods for characterizing quantum processes have been developed in recent years. 

These improvements can be grouped into two broad categories. 
The first category includes techniques such as \emph{randomised benchmarking (RB)}
\cite{KnillBenchmarking,MagGamEmer,MagGamEmer2,WalGraHar15,Wal17,Helsen17,WalFla14}
and \emph{gate set tomography (GST)} \cite{blume2013robust}, 
which are more robust to state preparation and measurement (SPAM) errors. 
These techniques work by performing long sequences of random quantum operations, 
measuring their outcomes, and checking whether the resulting statistics are consistent 
with some physically-plausible model of the system. 
In this way, one can characterise a quantum gate in terms of other quantum gates, in 
a way that is insensitive to SPAM errors.

The second category \cite{FlaGroLiu12,BalKalDeu14,KliKueEis15,KliKueEis17,HolBauCra15} 
provides more detailed tomographic information. It
includes techniques such as \emph{compressed sensing}
\cite{GroLiuFla10,Gro11,Liu11,ShaKosMoh11,KalevKosutDeutsch15,Kue15,KabKueRau15}, 
\emph{matrix product state tomography} \cite{CraPleFla10,Lanyonetal17}, 
and learning of local Hamiltonians and tensor network states \cite{SilLanPou11,LanPou12}. 
These methods exploit the sparse, low-rank or low entanglement 
structure that is present in many of the physical states and processes that occur in nature. 
These techniques are less resource-intensive than conventional tomography, and therefore 
can be applied to larger numbers of qubits. 
Convex optimization techniques, such as semidefinite programming, are then used to reconstruct the underlying quantum state or process. 

A recent line of work \cite{KimSilRya14,KimLiu16}
has attempted to unify these two approaches to a quantum process tomography scheme, that is both robust to SPAM errors, and can handle large numbers of qubits 
(provided the quantum process has some suitable structure). 
To achieve this goal, it turns out that the proper design of the measurements is crucial. 
SPAM-robust methods such as randomised benchmarking are known to require some kind of 
computationally-tractable group structure, such as that found in the Clifford group.
Clifford gates are motivated by their abundant appearance in many practical applications, such as fault-tolerant quantum computing \cite{NieChu10,MagGamEmer2}.

In contrast, compressed sensing methods typically require measurements with less structure in this context, in that their $4$th-order moments are close to those of the uniform Haar measure. 
Thus, the key technical question is whether the seemingly conflicting requirements of sufficient randomness and desired structure in the measurements can be combined.

In this work, we show that the answer is indeed yes. 
In layman's terms, we demonstrate that Clifford-group based measurements 
are also sufficiently unstructured that they can be used for compressed sensing. 
Thus, we develop methods for quantum process tomography that are resource efficient, 
robust with respect to SPAM and other errors, 
and use measurements that are already routinely acquired in many experiments. 

In more detail, we provide procedures for the reconstruction from so-called \emph{average gate fidelities} (\AGFs), which are the quantities that are measured in randomised benchmarking. 
It was established that the unital part of general quantum channels can be reconstructed from \AGFs\ relative to a maximal linearly independent subset of Clifford group operations \cite{KimSilRya14}. 
We generalise this result by noting that the Clifford group can be replaced by an arbitrary unitary $2$-design and also explicitly provide an analytic form of the reconstruction. 

Our main result is a practical reconstruction procedure for quantum channels that are close to being unitary. 
Let $d$ be the Hilbert space dimension, so that a unitary quantum channel can be described by roughly $d^2$ scalar parameters.
The protocol is rigorously guaranteed to succeed using essentially order of $d^2$
 \AGFs\ with respect to randomly drawn Clifford gates, and we also prove it to be stable against errors in the \AGF\ estimates. 
In this way we generalise a previous recovery guarantee \cite{KimLiu16} from \AGFs\ with $4$-designs to ones with the more relevant Clifford gates. 

Conversely, we prove that the sample complexity of our reconstruction procedure is optimal in a simplified measurement setting. 
Here, we assume that independent copies of the channel's Choi state are measured and use direct fidelity estimation \cite{FlaLiu11,SilLanPou11} and information theoretic arguments \cite{FlaGroLiu12} to show that the dimensional scaling of our reconstruction error is optimal up to log-factors. 
As a side result, we also find a new interpretation of the \emph{unitarity} \cite{WalGraHar15} -- a figure of merit that captures the coherence of noise. 
We show that this quantity can be estimated directly from \AGFs, rather than simulating purity measurements \cite{WalGraHar15}.

In summary, we provide a protocol for quantum process tomography that fulfils all of the following desiderata: 
\vspace{.5em}
\begin{enumerate}[label={(\roman*)}]
\item It should be based on physically reasonable and feasible measurements,
\item make use of them in a sample optimal fashion, 
\item exploit structure of the expected/targeted channel (here low Kraus rank reflecting
quantum gates), and
\item be stable against SPAM and other possible errors. 
\end{enumerate}
In this sense, we expect our scheme to be of high importance and practically useful 
in actual experimental settings in future quantum technologies \cite{Roadmap}.
It adds to the information obtained from mere randomised benchmarking in 
that it provides actionable advice, especially regarding coherent errors.
Such advice is particularly relevant for fault tolerant quantum computation: Refs.~\cite{KueLonFla15,Wal15} indicate that it is coherent errors that lead to an enormous mismatch between average errors, which are estimated by randomised benchmarking, and worst-case errors, reflected by fault tolerance thresholds.

Our main technical contributions are results for the second and fourth moments of \AGF\ measurements with random Clifford gates. 
For the second moment we provide an explicit formula improving over the previous lower bound \cite{KimLiu16}. 
In the case of trace-preserving and unital maps, our analysis gives rise to a tight frame condition. 
In order to prove a bound on the fourth moment, we derive -- as a more universal new technical tool -- a general integration formula for the fourth-order diagonal tensor representation of the Clifford group. The proof builds on recent results on the representation theory of the multi-qubit Clifford group \cite{ZhuKueGra16,HelWalWeh16,gross2017schur}. 
Our result is the Clifford analogue to 
Collins' integration formula for the unitary group \cite{Col03,CollinsSniady} for fourth orders, which we expect to also be useful in other applications.
In the following, we present the precise formulation of our results. The proofs and technical contributions are given in \detailsandproof\suppfootnote.  

\section{Main results}\label{sec:mainresults}
A linear map from the set of Hermitian operators on a $d$-dimensional Hilbert space to itself is referred to as \emph{map}. 
A quantum channel is a completely positive map that in addition preserves the trace of a Hermitian operator and, thus, maps quantum states to quantum states. 
A map is unital if the identity operator (equivalently, the maximally mixed state) is a fixed point of the map. 
We define the \emph{average gate fidelity} (\AGF) between a map $\mc X$ and a \emph{quantum gate} 
(i.e.\ a unitary quantum channel) 
$\mc U: \rho \mapsto U\rho\,  U\ad$ associated with a unitary matrix $U \in U(d)$ as
\begin{equation}
	\Favg(\mc U, \mc X) = \int \mathrm d \psi \sandwich{\psi}{U\ad\mc X(\ketbra{\psi}{\psi})U}{\psi},
\label{eq:average_fidelity}
\end{equation}
where the integral is taken according to the uniform (Haar) measure on state vectors. 

The Clifford group constitutes a particularly important family of unitary gates that feature prominently in state-of-the-art quantum architectures. 
Moreover, it was shown that for many-qubit systems (i.e.\ $d=2^n$), any unital and trace-preserving map is fully characterised by its \AGFs\ \eqref{eq:average_fidelity} with respect to the Clifford group \cite{KimSilRya14}. A detailed analysis of the geometry of unital channels was previously given in Ref.\ \cite{MenWol09}. 
There, it was shown that a quantum channel is unital if and only if it can be written as an affine combination of unitary gates. 
(\emph{Affine} here means that the expansion coefficients sum to $1$. 
Unlike \emph{convex} combinations, they are however not restrict to being non-negative.)
Motivated by the result for Clifford gates, one can ask more generally: 
What are the sets of unitary gates that span the set of unital and trace-preserving maps? 

A general answer to this question can be given using the notion of unitary $t$-designs. 
Unitary $t$-designs \cite{DanCleEme09,GrossAudenaertEisert:2007} (and their state-cousins, spherical $t$-designs \cite{DelGoeSei77,RenBluSco04}, respectively) are discrete subsets of the unitary group $\U(d)$ (resp. complex unit sphere) that are evenly distributed in the sense that their average reproduces the Haar (resp. uniform) measure over the full unitary group (resp. complex unit sphere) up to the $t$-th moment. 
The multi-qubit Clifford group, for example, forms a \emph{unitary $3$-design} \cite{Zhu15,Web15,KueGro15}. 
For spherical designs, a close connection between informational completeness for quantum state estimation and the notion of a $2$-design has been established in Ref.~\cite{RenBluSco04}, see also 
Refs.\ \cite{Sco06,App05,GroKraKue15_partial}. 
A similar result holds for quantum process estimation, and is the starting point of our work.
Indeed, the following is essentially due to Ref.~\cite{ScottDesigns}. We give a concise proof in form of the slightly more general Theorem \ref{thm:affcomb_conv} in \refdetailsandproof~\ref{sec:tight_frame}.

\begin{proposition}[{Informational completeness and unitary designs}]\label{prop:affcomb}
	Let $\{ \mc U_k\}_{k=1}^N$ be the gate set of a unitary $2$-design, represented as channels. 
	Every unital and trace-preserving map $\mc X$ can be 
	written 
	as an affine 
	combination 
	$\mc X = \frac1N \sum_{k=1}^N  c_k(\mc X) \mc U_k$
	of the $\mc U_k$'s. 
	The coefficients are given by $c_k(\mc X) = C \Favg(\mc U_k, \mc X) - \frac Cd + 1$, where $C = d(d+1)(d^2 -1)$. 
\end{proposition}
Hence, every unital and trace-preserving map is uniquely determined by the \AGFs\ with respect to an arbitrary unitary $2$-design.

  Clifford gates are a particularly prominent gate set with this 2-design feature. However, its cardinality scales super-polynomially in the dimension $d$. For explicit characterisations, this is far from optimal. 
However, in certain dimensions there exist subgroups of the Clifford group with cardinality proportional to $d^4$ that also form a 2-design \cite{Chau:2005, GrossAudenaertEisert:2007}. More generally, order of $d^4\log(d)$ Clifford gates drawn i.i.d.\ uniformly at random are an approximate $2$-design \cite{AmbainisEtAl:2009}. From Proposition~\ref{prop:affcomb}, we expect that  such randomly generated approximate $2$-designs yield approximate reconstruction schemes for unital channels. 

Our main result focuses on the particular task of reconstructing multi-qubit unital channels that are close to being unitary, i.e.\ well-approximated by a channel of Kraus rank equal to one.
Techniques from low-rank matrix reconstruction 
\cite{FazHinBoy01,GroLiuFla10,Gro11,FlaGroLiu12,KabKueRau15,BalKalDeu14}
allow for exploiting this additional piece of information in order to reduce the number of \AGFs\ required to uniquely reconstruct an unknown unitary gate. 

Suppose we are given a list  of  $m$ \AGFs\ 
\begin{equation}
f_i = \Favg(\mc C_i, \mc X) + \epsilon_i 	\label{eq:Fmeasurements}
\end{equation}
 -- possibly corrupted by additive noise $\epsilon_i$ -- between the unknown unitary gate $\mc X$ and Clifford gates $\mc C_i$ that are chosen uniformly at random.
In order to reconstruct $\mc X$ from these observations, we propose to perform a least-squares fit over the set of unital quantum channels, i.e.

\begin{equation}\label{eq:algorithm}
\begin{split}
	\operatorname{minimise}\quad &\sum_{i=1}^m (\Favg(\mc C_i,\mc Z) - f_i)^2 \\
	\operatorname{subject\ to}\quad &\text{$\mc Z$ is a unital quantum channel}.
\end{split}
\end{equation}
We emphasise that this is an efficiently solvable convex optimisation problem.
The feasible set is convex since it is the intersection of an affine subspace (unital and trace-preserving maps) and a convex cone (completely positive maps).

Valid for multi-qubit gates ($d=2^n$), our second main result states that this reconstruction procedure is guaranteed to succeed with (exponentially) high probability, provided that the number $m$ of \AGFs\ is proportional (up to a $\log(d)$-factor) to the number of degrees of freedom in a general unitary gate. 
The error of the reconstructed channel is measured with the Frobenius norm in Choi representation $\norm{\argdot}$, see \detailsandproof\ for details. 
Here, we give a concise statement for the case of unitary gates.
A more general version -- Theorem~\ref{thm:Choi-recovery} in \detailsandproof\ -- shows that the result can be extended to cover approximately unitary channels.
\\

\begin{theorem}[{Recovery guarantee for unitary gates}]
\label{thm:recoveryguarantee} 
\hfill\\
Fix the dimension $d=2^n$.
Then,
\begin{equation}
	m \geq c d^2 \log(d)
\end{equation}
noisy \AGFs\ with randomly chosen Clifford gates suffice with high probability (of at least $1- \mathrm{e}^{-\gamma m}$) to reconstruct \emph{any} 
unitary quantum channel $\mathcal{X}$ via \eqref{eq:algorithm}. 
This reconstruction is stable in the sense that the minimiser $\mc Z^\sharp$ of \eqref{eq:algorithm} is guaranteed to obey
 \begin{equation}\label{eq:errorbound}
 	\norm{\mc Z^\sharp - \mc X} \leq \tilde{C} \frac{d^2}{\sqrt{m}} \norm{\epsilon}_{\ell_2}.
 \end{equation}
 The constants $\tilde C,c, \gamma>0$ are independent of $d$. 
 \end{theorem}
\noindent We note the following:
\begin{enumerate}[nosep,leftmargin=0em,labelwidth=*,align=left,label={(\roman*)}]
\item Eq.~\eqref{eq:errorbound} shows the protocol's inherent stability to additive noise. This stability, combined with the robustness of randomised benchmarking against SPAM errors, results in an estimation procedure that is potentially more resource-intensive, but considerably less susceptible to experimental imperfections and systematic errors than many other reconstruction protocols \cite{FlaLiu11,FlaGroLiu12,KliKueEis17}.

\item The proof can be verbatim adapted to an optimisation of the $\ell_1$-norm instead of the $\ell_2$-norm in Eq.~\eqref{eq:algorithm}, resulting in a slightly stronger error bound.

\item The theorem achieves a quadratic improvement (up to a $\log$-factor) over the minimal number of \AGFs\ required for a naive reconstruction via linear inversion for the case of noiseless measurements. 
But what is the number of measurements required to obtain the \AGFs\ and to suppress the effect of the measurement noise $\epsilon$ in the reconstruction error \eqref{eq:errorbound}? 
For randomised benchmarking setups a fair accounting of all involved errors is beyond the scope of the current work.
But in order to show that the scaling of the noise term in our reconstruction error \eqref{eq:errorbound} is essentially optimal, we consider the conceptually simpler measurement setting where the channel's Choi state is measured directly. 
In \refdetailsandproof~\ref{sec:optimality} we prove upper and lower bounds to the minimum number of channel uses sufficient for a reconstruction via Algorithm~\eqref{eq:algorithm} with reconstruction error \eqref{eq:errorbound} bounded by $\recerror>0$. 
This number of channel uses scales as ${d^4}/{\recerror^2}$ up to log-factors. 
The upper bound relies on direct fidelity estimation \cite{FlaLiu11}. 
In order to establish a lower bound we extend information theoretic arguments from Ref.~\cite{FlaGroLiu12} to rank-$1$ measurements. 
\item Finally, we note that the reconstruction \eqref{eq:algorithm} can be  practically calculated using standard convex optimization packages. 
A numerical demonstration is shown in Figure~\ref{fig:numerical_reconstruction} and discussed in more detail in \refdetailsandproof~\ref{sec:numerics}. 
There we also show that measuring \AGFs\ with respect to Clifford unitaries seems to be comparable to Haar-random measurements, even in the presence of noise. This confirms an observation that was already mentioned in Ref.~\cite{KimLiu16}.
\end{enumerate}

\begin{figure}
\def\mywidth{1}
	\input{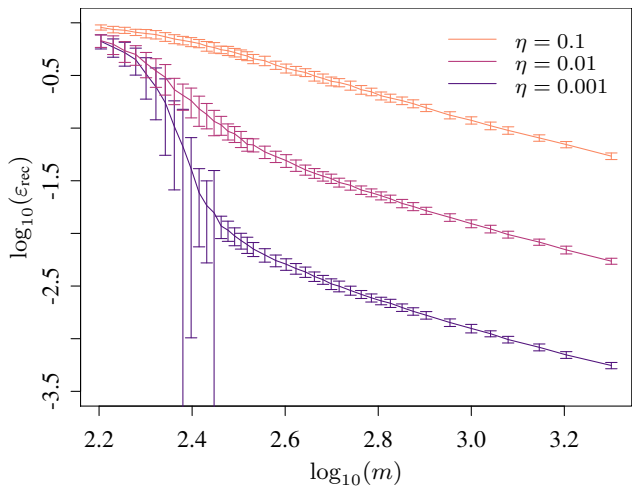}
	\caption{
		Reconstruction of a Haar random $3$-qubit channel using the optimization \eqref{eq:algorithm}: 
		The plots show the dependence of the observed average reconstruction error 
		$\recerror \coloneqq \norm{\mc Z^\sharp - \mc X}$ on 
		the number of \AGFs\ $m$ for different noise strengths 
		$\eta\coloneqq  \norm{\epsilon}_{\ell_2}$. 
		The error bars denote the observed standard deviation. 
		The averages are taken over $100$ samples of random {i.i.d.} measurements and channels (non-uniform). 
		The Matlab code and data used to create these plots can be found on GitHub \cite{our_git_repo}.
	}
	\label{fig:numerical_reconstruction}
\end{figure}

The proof of Theorem~\ref{thm:recoveryguarantee} is presented in \refdetailsandproof~\ref{ssec:proof_of_the_recovery_guarantee}. The \AGFs{} can be interpreted as expectation values of certain observables, which are unit rank projectors onto directions that correspond to elements of the Clifford group. In contrast, most previous work on tomography via compressed sensing feature observables that have full rank, e.g.~tensor products of Pauli operators. Since we now want to utilize observables that have unit rank, a different approach is needed. One approach, developed by a subset of the authors in \cite{KimLiu16} is to use strong results from low rank matrix reconstruction and phase retrieval \cite{CanStroVor13,CanLi13,GroKraKue15_partial,KueRauTer15,KabKueRau15}. These methods \cite{KueRauTer15,KabKueRau15} require measurements that look sufficiently random and unstructured, in that their $4$th-order moments are close to those of the uniform Haar measure. 
The multi-qubit Clifford group, however, does constitute a 3-design, but not a 4-design. 
In Ref.~\cite{KimLiu16} this discrepancy is partially remedied by imposing additional  constraints (a ``non-spikiness condition'', see also Ref.\ \cite{KraLiu16}) on the unitary channels to be reconstructed. 
In turn, their result also required these constraints to be included in the algorithmic reconstruction which renders the algorithm impractical \footnote{The cardinality of the Clifford group grows superpolynomially in the Hilbert space dimension $d$. Therefore, the non-spikiness with respect to the Clifford group quickly corresponds to a demanding number of constraints. In fact, about $10^8$ and $10^{13}$ constraints are already required for $3$-qubits and $4$-qubits, respectively.}. Moreover, important classes of channels, e.g. Pauli channels, do in general not satisfy this condition. 
Here, we overcome these issues by appealing to recent works that fully characterise the fourth moments of the Clifford group \cite{ZhuKueGra16,HelWalWeh16}. 
In order to apply these results, we develop an integration formula for fourth moments over the Clifford group. 
This formula is analogous to the integration over the unitary group know as Collins' calculus with Weingarten functions \cite{Col03}; see \refdetailsandproof~\ref{sec:rep_theory}.  
Equipped with this new representation theoretic technique we show in \refdetailsandproof~\ref{subsec:fourthmoment} that the deviation of the Clifford group from a unitary 4-design is -- in a precise sense -- mild enough for the task at hand. 

Our final result addresses the \emph{unitarity} of a quantum channel. Introduced by Wallman et al.~\cite{WalGraHar15}, the unitarity is a measure for the coherence of a (noise) channel $\mc E$.
It is defined to be the average purity of the output states of a slightly altered channel $\mathcal{E}'$
\footnote{$\mc E'$ is defined so that $\mc E' (\Id
) = 0$ and 
	$\mc E' (X) = \mc E(X) - { \tr (\mc E (X)}/{\sqrt{d}} \Id
$ for all traceless $X$.
	}
\begin{equation}
u(\mathcal{E}) = \int \rmd  \psi \tr \left( \mc E' \left( | \psi \rangle \! \langle \psi | \right)^\dagger \mc E' \left( | \psi \rangle \! \langle \psi | \right) \right)
\label{eq:unitarity}
\end{equation}
that flags the absence of trace-preservation and unitality. 
The unitarity can be estimated efficiently by using techniques similar to randomised benchmarking \cite{FenWalBuo16}. 
It is also an important figure of merit when one aims to compare the \AGF\ of a noisy gate implementation to its diamond distance \cite{KueLonFla15, Wal15} -- a task that is important for certifying fault tolerance capabilities of  quantum devices.

Although useful, the existing definition of the unitarity \eqref{eq:unitarity} is arguably not very intuitive. Here, we try to (partially) amend this situation by providing a simple statistical interpretation:
\begin{theorem}
[{Operational interpretation of unitarity}]
\label{thm:unitarity}
Let $\left\{\mathcal{U}_k \right\}_{k=1}^N$ be the gate set of a unitary 2-design. Then, for all hermicity preserving maps $\mc X$
\begin{equation}
\mathrm{Var} \left[ \Favg \left(\mathcal{U}_k, \mathcal{X} \right) \right] = \frac{ u (\mathcal{X})}{d^2 (d+1)^2},
\label{eq:unitarity_variance}
\end{equation}
where the variance is computed with respect to $\mathcal{U}_k$ drawn randomly from the
unitary 2-design.
\end{theorem}
The proof of the theorem is given in \refdetailsandproof~\ref{sec:tight_frame}. 
Note that the variance is taken with respect to unitaries drawn from the unitary 2-design and not the variance of the average fidelity with respect to the input state as calculated, e.g.\ in Ref.\ \cite{MagesanEtAl:2011}.

\section{Conclusion and outlook}\label{sec:conclusions}
In this work we address the crucial task of characterising quantum channels.
We do so by relying on \AGFs\ of the quantum channel 
of interest with simple-to-implement Cliffords. 
More specifically, we start by noting that
(i) the unital part of any quantum channel can be written in terms of a unitary $2$-design 
with expansion coefficients given by \AGFs. 
As a consequence, for certain Hilbert space dimensions $d$, the unital part can be reconstructed from $d^4$ \AGFs\ with Clifford group operations by a straight-forward and stable expansion formula. 
(ii) As the main result, we prove for the case of unitary gates that the reconstruction can be practically done using only essentially order of $d^2$ 
random \AGFs\ with Clifford gates.
In a simplified measurement setting, we show that this setting is provably resource optimal
in terms of the number of channel invocations.
For the proof, we derive a formula for the integration of fourth moments over the Clifford group, which is similar to Collins' calculus with Weingarten functions. 
This integration formula might also be useful for other purposes. 
(iii) We prove that the unitarity of a quantum channel, which is a measure for the coherence of noise \cite{WalGraHar15}, has a simple statistical interpretation: 
It corresponds to the variance of the \AGF\ with unitaries sampled from a unitary $2$-design.

The focus of this work is on the reconstruction of quantum gates. 
Here, the assumption of unitarity
considerably simplifies the representation-theoretic effort for 
establishing the fourth moment bounds required for applying strong existing proof techniques from low rank matrix 
recovery. 
These extend naturally to higher Kraus ranks and we leave this generalisation to future work. Existing results \cite{KueZhuGro16a,KueZhuGro16b} indicate that the deviation of the Clifford group from a unitary 4-design may become more pronounced when the rank of the states/channels in question increases. This may lead to a non-optimal rank-scaling of the required number of observations $m$. 
In fact, a straightforward extension of Theorem~\ref{thm:recoveryguarantee} to the Kraus rank-$r$ case already yields a recovery guarantee with a scaling of $m \sim r^5 d^2 \log(d)$.

Practically, it is important to explore how this protocol behaves when applied to data obtained from interleaved randomised benchmarking experiments. In Ref. \cite{KimSilRya14}, the authors show how to use interleaved randomised benchmarking experiments to measure the \AGF\ between a known Clifford and the combined process of an unknown gate concatenated with the average Clifford error process. In order to obtain tomographic information about the isolated unknown gate, the authors had to do a linear inversion of the average Clifford error. However, in most cases, we  expect the average Clifford error to be close to a depolarizing channel which has very high rank. Thus, 
building on our intuition obtained for quantum states \cite{RioGroFla16} and using our techniques, we could obtain a low-rank approximation to the combined unknown gate and average Clifford error, which under the assumption of a high rank Clifford error, would naturally pick out the coherent part of the unknown gate. 


\section{Details and proofs} \label{sec:details}
In this section we provide proofs and further details of the results of the work.
Section~\ref{sec:rep_theory}--\ref{subsec:fourthmoment} develop the prerequisites to prove the recovery guarantee, Theorem~\ref{thm:recoveryguarantee}, in Section~\ref{ssec:proof_of_the_recovery_guarantee}. The optimality of this result is addressed in Section~\ref{sec:optimality}. 
The expansion of unital maps in terms of a unitary $2$-design, Proposition~\ref{prop:affcomb}, is derived in Section~\ref{sec:tight_frame}. 
In Section~\ref{sec:unitarity}, we show that the unitarity of a hermiticity preserving map can be expressed as the variance of its average gate fidelity with respect to a unitary $2$-design. 
We also discuss possible implications. 
Finally, Section~\ref{sec:numerics} provides further details of the numerical demonstration of the protocol. 


We start by specifying the notation that is used subsequently. For a vector space $V$ we denote the space of its endomorphisms by $\Lin(V)$. In particular, let $H_d$ denote the space of hermitian operators on a $d$-dimensional complex Hilbert space.
We label the vector space of endomorphisms on $H_d$ by $\Lin(H_d)$ and denote its elements with calligraphic letters. 
For every map $\mathcal{X} \in \Lin(H_d)$, we define its adjoint 
$\mathcal{X}^\dagger \in \Lin(H_d)$ with respect to the Hilbert-Schmidt inner product $(\cdot, \cdot)$ on $H_d$. 
We denote the subset of \emph{completely positive} maps by 
$\operatorname{CP}(H_d) \subset \Lin(H_d)$. 
\emph{Quantum channels} are elements of $\CP(H_d)$ that are 
\emph{trace preserving} (TP), i.e.\
$\tr \left( \mc {E}(X) \right) = \tr (X)$ for all  
$X \in H_d$. 
This condition is equivalent to the identity matrix 
$\Id \in H_d$ being a fixed point of the adjoint channel, 
$\mc E\ad (\Id) = \Id$.
Similarly, a map (or channel) $\mc E$ that itself has the identity as a fixed-point, $\mc E(\Id) = \Id$,  is called \emph{unital}. The affine subspace of TP and unital maps is denoted by $\Lin_\utp(H_d) \subset \Lin(H_d)$. 
We further denote the linear hull of $\Lin_\utp(H_d)$ by $\Lin_{\overline{\utp}}(H_d)$.

Most of our results feature a norm on $\Lin (H_d)$, which is naturally induced on by the average gate fidelity (\AGF) \eqref{eq:average_fidelity} in the following way. We define the inner product on $\Lin(H_d)$ as
\begin{equation}\label{eq:inner_product}
 	(\mc X, \mc Y) = \frac{d+1}{d} \Favg(\mc X, \mc Y) - \frac{1}{d^2}( \mc X(\Id), \mc Y(\Id))
\end{equation}
and denote the induced norm on $\Lin(H_d)$ by $\| \mc X \|^2 = \left( \mc X, \mc X \right)$.  The pre-factors are chosen such that unitary channels $\mc U \in \Lin(H_d)$ have unit norm. 

Note that  this inner product is proportional to the previously defined Hilbert-Schmidt inner product applied to the
Choi and Liouville representations:
\begin{align}
\left( \mc X, \mc Y \right) = \left( J(\mc X), J(\mc Y) \right) = \frac{1}{d^2} \left(\mc L(\mc X),\mc L (\mc Y) \right),
\label{eq:inner_product_choi}
\end{align}
see Refs.\ \cite{HorHorHor99,Niel02} and also \cite[Proposition~1]{KueLonFla15}. 
We choose the convention that Choi matrices of quantum channels have unit trace, i.e.~$\Tr(J(\mc X))=1$. Furthermore, for $X \in H_d$ we will encounter the  Schatten norms $\tnorm{X} = \Tr[\sqrt{XX\ad}]$, $\fnorm{X} = \sqrt{\Tr(XX\ad)}$ and $\snorm{X} = \sqrt{\mu_\text{max}(XX\ad)}$, where $\mu_\text{max}(Y)$ denotes the maximum eigenvalue of a Hermitian matrix $Y$. For a vector $y \in \RR^m$ and $q \in \NN$ the $\ell_q$-norm is defined by $\norm{y}_{\ell_q} = (\sum_{i=1}^m |y_i|^q)^{1/q}$.  

For a map $\mc T: H_d \to H_d$ we define the random variable 
\begin{equation}\label{eq:SmcT}
	S_{\mc T} = d^2 (\mc T, \mc U)
\end{equation}
where $\mc U$ is a unitary channel $\mc U(X) = U X U\ad$ with $U$ either chosen uniformly at random  from the full unitary group $\U(d)$, or the Clifford group $\Cl(d)$, depending on the context. The main technical ingredients for the the proofs of our main results
are an expression for the second and fourth moment of $S_{\mc T}$.
 To this end, 
an integration formula for the first four moments over the Clifford group is developed in Section~\ref{sec:rep_theory}. We then derive an explicit expression for the second moment of $S_{\mc T}$ in Section~\ref{subsec:secondmoment} and an upper bound on the fourth moment of $S_{\mc T}$ 
in Section~\ref{subsec:fourthmoment}. 
These bounds are essential prerequisites for applying strong techniques from low-rank matrix reconstruction to prove our recovery guarantee, Theorem~\ref{thm:recoveryguarantee}, for unitary gates in Section~\ref{ssec:proof_of_the_recovery_guarantee}.

\subsection{An integration formula for the Clifford group} \label{sec:rep_theory}

One of the main technical ingredients of the proof is an explicit formula for integrals of the diagonal action of the Clifford group $\Cliff(d)$. More precisely, for a unitary representation $R : G \to \End(V)$ of a subgroup $G \subset \U(d)$ carried by a vector space $V$, we define $E_R: \End(V) \to \End(V)$ (``twirling'') as 
\begin{equation}\label{eq:E_R}
	E_R(A) = \int _G R(g) A R(g)\ad d\mu(g),
\end{equation}
where $\mu$ is the invariant measure induced by the Haar measure on $\U(d)$.

For $V = (\CC^d)^{\otimes n}$ we denote the diagonal action of a subgroup $G$ of $\GL(\CC^d)$ by $\Delta^n_{G}: G \to \GL(V)$, i.e.
\begin{equation}
	\Delta^n_{G}: U \mapsto \underbrace{U \otimes \ldots \otimes U}_{n \text{ times}}.
\end{equation}
Note that if $G$ is a subgroup of the unitary group $\U(d)$ then $\Delta^n_G$ is a unitary representation. The main result of this chapter is an explicit expression for $E_{\Delta_{\Cliff(d)}^4}(A)$ for arbitrary $A \in \End(V)$.  

For $E_{\Delta_{\U(d)}^n}(A)$, where the integration is carried out over the entire unitary group, an explicit formula was derived in Refs.~\cite{Col03,CollinsSniady}. It is instructive to review the result of  Ref.\ \cite{CollinsSniady} and its proof first. Our derivation of the analogous expression for the Clifford group follows the same strategy and makes use of many of the intermediate results. 

\subsubsection{Integration over the unitary group \texorpdfstring{$\U(d)$}{U(d)} }

To state the result we have to introduce notions from the representation theory of $\Delta^n_{\U(d)}$ which can be found, e.g., in Refs.~\cite{CollinsSniady,GoodmanWallach:2009,FouHar91,Col03}. Schur-Weyl duality relates the irreducible representations of the diagonal action of $\GL(V)$ to the irreducible representations of the natural action of the symmetric group $S_n$  on $V$. 
Recall that  the representation $\Delta^n_{\U(d)}$ decomposes into irreducible representations $\Delta^\lambda_{\U(d)}: \U(d) \to \GL(W_\lambda)$ labelled by partitions $\lambda = (\lambda_1, \lambda_2, \ldots, \lambda_{l(\lambda)})$ of $n$ into $l(\lambda)\leq d$ integers, i.e. $\sum_{i = 1}^{l(\lambda)} \lambda_i = n$. For short, we denote  a partition of $n$ by $\lambda \vdash n$ and dimensions of the Weyl-modules $W_\lambda$ by $D_\lambda$.

Let $\{\ket i\}_{i=1}^d$ be an orthonormal basis of $\CC^d$. 
We define the representation $\pi_{S_n}^{d}: S_n \to \GL(V)$ by linearly extending 
\begin{equation}
	\pi^d_{S_n}(\tau): \ket {i_1} \otimes \ldots \otimes \ket {i_k} \mapsto \ket {i_{\tau^{-1}(1)}} \otimes \ldots \otimes \ket {i_{\tau^{-1}(k)}}. 
\end{equation}

The irreducible representations of $\pi^d_{S_n}$, $\pi^\lambda_{S_n}: S_n \to \GL(S_\lambda)$ are also labelled by  partitions $\lambda \vdash n$. The dimensions of the Specht-modules $S_\lambda$ are denoted by $d_\lambda$.
Since the actions of $\Delta^n_{\U(d)}$ and $\pi_{S_n}^d$ commute, they induce a representation of $\U(d) \times S_n$ on $(\CC^d)^{\otimes n}$ that decomposes into irreducible representations as follows:

\begin{theorem}[Schur-Weyl decomposition]
\label{thm:SchurWeyl}
	The action of $\U(d) \times S_n$ on $(\CC^d)^{\otimes n}$ is multiplicity free and $(\CC^d)^{\otimes n}$ decomposes into irreducible components as 
	\begin{equation}\label{eq:SchurWeylDecomposition}
		(\CC^d)^{\otimes n} \cong \bigoplus_{\lambda \vdash n, l(\lambda) \leq d} W_\lambda \otimes S_\lambda
	\end{equation}
	on which $\U(d) \times S_n$ acts as $\Delta^\lambda_{\U(d)} \otimes \pi_{S_n}^\lambda$. 
\end{theorem}

We denote the orthogonal projections on $W_\lambda \otimes S_\lambda$ by $P_\lambda$ and the character on the irreducible representation  $\pi^\lambda_{S_n}$ of $S_n$ by $\chi^\lambda (\pi) \coloneqq \Tr(\pi^\lambda_{S_n}(\pi))$. The orthogonal projectors can be written as
\begin{equation}\label{eq:projasperm}
	P_\lambda = \frac{d_\lambda}{n!} \sum_{\sigma \in S_n} \chi^\lambda(\sigma)  \pi_{S_n}^d(\sigma),
\end{equation}
see, e.g.\ Ref.~\cite[Eq.~(12.10)]{Wigner1959Grouptheory}.
In terms of these projectors $E_{\Delta_{\U(d)}^n}(A)$ can be calculated using the following theorem.

\begin{theorem}[Integration over the unitary group $\U(d)$]\label{thm:intU}
	Let $A \in \mc \End(V)$. Then, for $R = \Delta^n_{\U(d)}$ and $G = \U(d)$,
	\begin{equation}
	\begin{split}
		&E_{\Delta^n_{\U(d)}}(A) \\
		&\quad\quad= \frac{1}{n !} \sum_{\tau \in S_n} \Tr( A \pi^d_{S_n}(\tau))\, \pi^d_{S_n}(\tau^{-1}) \sum_{\lambda \vdash n,\ l(\lambda) \leq d} \frac{d_\lambda}{D_\lambda} P_\lambda.
	\end{split}
	\end{equation}
\end{theorem}

This formula differs slightly from the original statement presented in Ref.~\cite{CollinsSniady}.
 The more common formulation presented there follows from evaluating the expression of Theorem~\ref{thm:intU} using a standard tensor basis of $\End(V)$ \footnote{This way of stating the result of Ref.\ \cite{CollinsSniady} was brought to our attention by study notes of K.~Audenaert}. However, here we have opted for a presentation of Theorem~\ref{thm:intU} that is easier to generalise beyond the full unitary group.

In the remainder of this section, we present a proof of Theorem~\ref{thm:intU} following  the strategy of Ref.\ \cite{CollinsSniady}. The commutant of a subset $\mc A \subset \End(V)$ is the subset of $\End(V)$ defined by 
\begin{equation}
	\Comm(\mc A) = \{B \in \End(V) \mid  B A = A B \quad \forall A \in \mc A \} . 
\end{equation}

It is straight-forward to verify the following well-known properties of $E_R$:
\begin{lemma}[Properties of $E_R$]\label{lem:propE}
Let $R$ be a unitary representation of a subgroup $ G \subseteq \U(d)$.
	Then, for all $A \in \End(V)$ and $B \in \Comm(R(G))$, the map $E_R$ (defined in Eq.~\eqref{eq:E_R}) fulfils 
\begin{align}
\Tr(E_R(A)) =& \Tr(A),\label{lem:propE:trace} \\
E_R(A B) =& E_R(A) B,  \label{lem:propE:modmorph} \\
E_R(A) \in & \Comm(R(G)).  \label{lem:propE:comm}
\end{align}
\end{lemma}

The last statement of Lemma~\ref{lem:propE} implies that $E_{\Delta_{\U(d)}}^n(A)$ is in the commutant of $\Delta_{\U(d)}^n$ for all $A \in \End(V)$. Using the decomposition of Theorem~\ref{thm:SchurWeyl} and Schur's Lemma we therefore conclude that $E_{\Delta_{\U(d)}^n}(A)$ acts as the identity on the Weyl-modules, 
\begin{equation}
	E_{\Delta_{\U(d)}^n}(A) = \sum_{\lambda \vdash n, l(\lambda) \leq d} \Id_{D_\lambda} \otimes E_\lambda
\end{equation} 
with $E_\lambda  \in \End(S_\lambda)$. In general, the direct sum of  endomorphisms acting on the irreducible representations of a group is isomorphic to the group ring which consists of formal (complex) linear  combinations of the group elements \cite[Propositon~3.29]{FouHar91}.
We denote the group ring of $S_n$ by $\CC[S_n]$. 

To derive an explicit expression of the coefficient of the expansion of $E_{\Delta_{\U(d)}^n}(A)$  in $\CC[S_n]$, we introduce the map $\Phi: \End(V) \to \End(V)$
\begin{equation}\label{eq:Phi}
	\Phi(A) = \sum_{\sigma \in S_n } \Tr(A \pi^d_{S_n}(\sigma^{-1}))\pi^d_{S_n}(\sigma). 
\end{equation}
We will make use of the following properties of the map $\Phi$. 
\begin{lemma}[Properties of $\Phi$] \label{lem:propPhi}
For all $A \in \End(V)$ and $B \in \Comm(\Delta_{\U(d)}^n)$
\begin{align}
\Phi(A) =& \Phi(E_{\Delta_{\U(d)}^n}(A)), \label{lem:propPhi:E} \\
\Phi(B) =& B\Phi(\Id), \label{lem:propPhi:Bimod}\\
\Phi(\Id)^{-1} =& \frac{1}{n!} \sum_{\lambda \vdash n, l(\lambda) \leq d} \frac{d_\lambda}{D_\lambda} P_\lambda.\label{lem:propPhi:Id}
\end{align}
\end{lemma}

\begin{proof}{}
	\begin{enumerate}
		\item Since $\pi^d_{S_n}(\sigma^{-1})$ is in $\Comm(\Delta^n_{\U(d)})$ for all $\sigma \in S_n$, we can apply  Lemma~\ref{lem:propE}  to get
		\begin{equation}
		\begin{aligned}
			\Tr(E_{\Delta^n_{\U(d)}}(A) \pi^d_{S_n}(\sigma^{-1}))
			 &=  
			 \Tr(E_{\Delta^n_{\U(d)}}(A \pi^d_{S_n}(\sigma^{-1}))) 
			 \\	
			 &=  \Tr(A  \pi^d_{S_n}(\sigma^{-1})) 
			\, ,
		\end{aligned}
		\end{equation}
		which establishes the first statement.
		\item Since the commutant is isomorphic to the group ring, it suffices to proof the statement for all $B = \pi^d_{S_n}(\tau)$ with $\tau \in S_n$. In this case, using the cyclicity of the trace for the first equality, we find
		\begin{equation}
		\begin{aligned}
			\Phi(\pi^d_{S_n}(\tau)) &=   \sum_{\sigma \in S_n } \Tr(\pi^d_{S_n}(\sigma^{-1})\pi^d_{S_n}(\tau) )\pi^d_{S_n}(\sigma) \\
			&= \sum_{\sigma \in S_n} \Tr(\pi^d_{S_n}( \tau\sigma^{-1})) \pi^d_{S_n}(\sigma) \\
			&=  \sum_{\sigma \in S_n} \Tr(\pi^d_{S_n}(\sigma^{-1})) \pi^d_{S_n}( \sigma \tau) \\
			&= \pi^d_{S_n}( \tau)\sum_{\sigma \in S_n} \Tr(\pi^d_{S_n}(\sigma^{-1})) \pi^d_{S_n}( \sigma).
		\end{aligned}
		\end{equation}
		Here we have used that $\pi^d_{S_n}(\tau \sigma) = \pi^d_{S_n}(\sigma) \pi^d_{S_n}(\tau)$ for all $\tau, \sigma \in S_n$. 
		\item 
		Using Theorem~\ref{thm:SchurWeyl} (Schur-Weyl duality), we can rewrite $\Phi(\Id)$ as 
		\begin{equation}
		\begin{aligned}
			\Phi(\Id) &=\sum_{\sigma \in S_n} \Tr( \pi^d_{S_n}(\sigma^{-1})) \pi^d_{S_n}(\sigma) \\
				&= \sum_{\sigma \in S_n} \sum_{\lambda \vdash n, l(\lambda) \leq d} D_\lambda \Tr(\pi_\lambda(\sigma^{-1}))\pi^d_{S_n}(\sigma)   \\
				&=   \sum_{\lambda \vdash n, l(\lambda) \leq d}  D_\lambda  \sum_{\sigma \in S_n}\chi^\lambda(\sigma) \pi^d_{S_n}(\sigma) .
		\end{aligned}
		\end{equation}
		 The explicit expression \eqref{eq:projasperm} for the projectors identifies 
		 $\Phi(\Id)$ as 
		\begin{equation}
		 	\Phi(\Id) = n! \sum_{\lambda \vdash n, l(\lambda) \leq d} \frac{D_\lambda}{d_\lambda} P_\lambda.
		 \end{equation} 
		 Since the $\{P_\lambda\}$ are a complete set of orthogonal projectors, the inverse of $\Phi(\Id)$ is given by
		 \begin{equation}
		 	\Phi(\Id)^{-1} = \frac{1}{n!} \sum_{\lambda \vdash n, l(\lambda) \leq d} \frac{d_\lambda}{D_\lambda} P_\lambda.
		 \end{equation}
	\end{enumerate}
\end{proof}

We are now in position to give a concise proof of Theorem~\ref{thm:intU}:
\begin{proof}[Proof of Theorem~\ref{thm:intU}]
	From Eqns.~\eqref{lem:propPhi:E} and \eqref{lem:propPhi:Bimod} we conclude $\Phi(A) = \Phi(E_{\Delta^n_{\U(d)}} (A)) = E_{\Delta^n_{\U(d)}}(A) \Phi(\Id)$ and, thus, 
	$ E_{\Delta^n_{\U(d)}}(A) = \Phi(A) \Phi(\Id)^{-1}$. 
	Inserting the expression \eqref{lem:propPhi:Id} for $\Phi(\Id)^{-1}$ and the definition \eqref{eq:Phi} of $\Phi$ yields the expression of the theorem. 
\end{proof}

\subsubsection{Integration over the Clifford group}
We now turn our attention to the Clifford group and aim at an analogous result to Theorem~\ref{thm:intU} for $E_{\Delta^4_{\Cl(d)}}(A)$ with $A\in \End(V)$. As the former result for the unitary group, the result  for the Clifford group heavily relies on a characterisation of the commutant of $\Delta^4_{\Cl(d)}$. The required results for the Clifford group were derived in Ref.\ \cite{ZhuKueGra16} and apply to multi-qubit dimensions $d = 2^n$.
This paper introduces the orthogonal projection
\begin{equation}\label{eq:Qdef}
	Q = \frac{1}{d^2} \sum_{k=1}^{d^2} W_k^{\otimes 4}
\end{equation}
where $W_1,\ldots,W_{d^2} \in \End \left( \mathbb{C}^d \right)$ are the multi-qubit Pauli matrices. In fact, the $d^2$-dimensional range of $Q$ forms a particular stabiliser code. We denote by $Q^\perp = \Id - Q$ the orthogonal projection onto the complement of this stabiliser code. The orthogonal projection $Q$ commutes with every $\pi^d_{S_4}(\sigma)$, $\sigma \in S_4$. Thus, $Q$ acts trivially on the Specht modules $S_\lambda$ in the Schur-Weyl decomposition \eqref{eq:SchurWeylDecomposition}.
Following the notation conventions from Ref.~\cite{ZhuKueGra16}, we
denote the subspace of the Weyl module $W_\lambda$ that intersects with the range of $Q$ by $W^+_\lambda$ and its dimension as $D^+_\lambda$. Analogously, the orthogonal complement of $W^+_\lambda$ shall be $W^-_\lambda$ with dimension $D^-_\lambda$.  
We are now ready to state the main result of this section. 

\begin{theorem}[Integration over the Clifford group $\Cl(d)$]\label{thm:intCl}
	Let 
	 $A \in \mc \End(V)$. Then, 
\begin{equation}\begin{split}
	E_{\Delta^4_{\Cl(d)}}(A) 
	&= \frac{1}{4!} \sum_{\lambda \vdash 4, l(\lambda) \leq d} d_\lambda  \sum_{\sigma \in S_4} 
	\\
	&\times \bigg[ \frac{1}{D^+_\lambda} \Tr(A Q \pi^d_{S_4}(\sigma^{-1})) Q  \\
	&  + \frac{1}{D^-_\lambda}\Tr(A Q^\perp \pi^d_{S_4}(\sigma^{-1}))Q^\perp 	\bigg] \\
	&\times  \pi^d_{S_4}(\sigma)P_\lambda.
\end{split}\end{equation}
\end{theorem}

To set-up the proof we summarise the necessary results of Ref.~\cite{ZhuKueGra16} in the following theorem:
\begin{theorem}[Representation theory of the Clifford group \cite{ZhuKueGra16}]
\label{thm:SchurWeylClifford}
	Whenever $W^\pm_\lambda$ are non-trivial, the action of $\Cl(d) \times S_4$ on $(\CC^d)^{\otimes 4}$ is multiplicity free and $(\CC^d)^{\otimes 4}$ decomposes into irreducible components
	\begin{equation}\label{eq:SchurWeylDecompositionClifford}
		(\CC^d)^{\otimes 4} \cong \bigoplus_{\lambda \vdash 4, l(\lambda) \leq d} (W^+_\lambda \otimes S_\lambda) \oplus (W^-_\lambda \otimes S_\lambda),
	\end{equation}
	on which $\Cl(d) \times S_4$ acts as $\Delta^\lambda_{\Cl(d)} \otimes \pi_{S_4}^\lambda$. 

	The dimensions of $W^+_\lambda$ are of polynomials in $d$ of degree $4$ and the dimensions of $W^-_\lambda$ are either vanishing or polynomials in $d$ of degree $2$. 
\end{theorem}

From Theorem~\ref{thm:SchurWeylClifford} we learn that an element of the commutant of the diagonal action of the Clifford group $\Delta^4_{\Cl(d)}$ can be written in the form 
\begin{equation}\label{eq:cliffcommform}
	B = Q\bigoplus_{\lambda \vdash 4, l(\lambda) \leq d} (\Id_{D_\lambda} \otimes B^+_\lambda) + Q^\perp \bigoplus_{\lambda \vdash 4, l(\lambda) \leq d} (\Id_{D_\lambda} \otimes B^-_\lambda),
\end{equation}
where $B^\pm_\lambda \in \End(S_\lambda)$ are linear operators acting on the Specht modules $S_\lambda$. 

To expand elements of $\Comm(\Delta^4_{\Cl(d)})$, we  define the map  $\tilde\Phi: \End(V) \to \End(V)$, $\tilde\Phi(A) = \Phi(AQ)Q + \Phi(AQ^\perp)Q^\perp$ with $\Phi$ from \eqref{eq:Phi}. The map $\tilde\Phi$ has properties comparable to the map $\Phi$, but is adapted to the diagonal representation of the Clifford group.
\begin{lemma} \label{lem:propTildePhi}
For all $A \in \End(V)$ and $B \in \Comm(\Delta_{\Cl(d)}^4)$
\begin{align}
\tilde\Phi(A) =& \tilde\Phi(E_{\Delta_{\Cl(d)}^4}(A))	\label{lem:propTildePhi:E}, 	\\
\tilde\Phi(B) =& B\tilde\Phi(\Id), \label{lem:propTildePhi:Bimod}\\
\tilde\Phi(\Id)^{-1} =& \frac{1}{4!} \sum_{\lambda\vdash 4, l(\lambda) \leq d} d_\lambda P_\lambda \left[ \frac{1}{D_\lambda^+} Q + \frac{1}{D_\lambda^-} Q^\perp \right].\label{lem:propTildePhi:Id}
\end{align}

\end{lemma}

\begin{proof}
\quad
	\begin{enumerate}
		\item Since $Q\pi^d_{S_4}(\sigma^{-1})$ and $Q^\perp \pi^d_{S_4}(\sigma^{-1})$  are in $\Comm(\Delta^4_{\Cl(d)})$ for all $\sigma \in S_4$, we can again apply  Lemma~\ref{lem:propE}  to get
		$ \Tr(E_{\Delta^4_{\Cl(d)}}(A) Q \pi^d_{S_4}(\sigma^{-1})) =  \Tr(E_{\Delta^4_{\Cl(d)}}(A Q \pi^d_{S_4}(\sigma^{-1}))) =  \Tr(A Q \pi^d_{S_4}(\sigma^{-1}))$ and likewise for $Q^\perp$ instead of $Q$. Inserting this in the definition of $\tilde\Phi$ yields the first statement.
		\item From the expansion of elements $B \in \Comm(\Delta_{\Cl(d)}^4)$ in \eqref{eq:cliffcommform}, we conclude that $B$ can be expressed as $B = QB_1 + Q^\perp B_2$, where  $B_1$ and $B_2$ are in the group ring $\CC[S_4]$. 
		Hence, it suffices to show the statement, $\tilde\Phi(B) = B\tilde\Phi(\Id)$, for $B = Q \pi^d_{S_4}(\sigma)$ and $B=Q^\perp \pi^d_{S_4}(\sigma)$. In the first case, we find
		\begin{equation}
		\begin{aligned}
			\tilde\Phi( Q \pi^d_{S_4}(\sigma)) &= \Phi( Q \pi^d_{S_4}(\sigma))Q \\
			&= \Phi( Q \Id)Q \pi^d_{S_4}(\sigma) \\
			&= \tilde\Phi( \Id)Q \pi^d_{S_4}(\sigma) \, ,
		\end{aligned}
		\end{equation}
		where property~\eqref{lem:propE:modmorph} from Lemma~\ref{lem:propE} has been used in the second step.
		The proof of $Q^\perp$ is analogous. 
		\item Using the decomposition \eqref{eq:SchurWeylDecompositionClifford} of Theorem~\ref{thm:SchurWeylClifford}, we can calculate
 		\begin{equation}\begin{split}
 			\tilde\Phi(\Id) = &\sum_{\lambda  \vdash 4, l(\lambda) \leq d} \sum_{\sigma \in S_4}  \chi_{\pi_{S_4}^d} (\sigma^{-1}) \pi_{S_4}^d(\sigma) \\ 
 			&\times \left[ D_\lambda^+
		 	 Q   + D_\lambda^- Q_\lambda^\perp  \right] \\
 	 		&= 4! \sum_{\lambda} \frac{1}{d_\lambda} P_\lambda \left[ D_\lambda^+
 	 		Q   + D_\lambda^- Q^\perp  \right],
		\end{split}\end{equation}
		where the last line follows again from the expression \eqref{eq:projasperm} for the projectors. 
		Inverting this expression yields
		\begin{equation}
			\tilde\Phi(\Id)^{-1} = \frac{1}{4!} \sum_{\lambda} d_\lambda P_\lambda \left[ \frac{1}{D_\lambda^+} Q + \frac{1}{D_\lambda^-} Q^\perp \right].
		\end{equation}
	\end{enumerate}
\end{proof}

With these statements for the Clifford group at hand, we can proceed to prove Theorem~\ref{thm:intCl}.

\begin{proof}[Proof of Theorem~\ref{thm:intCl}]
Eq.~\eqref{lem:propTildePhi:E} in Lemma~\ref{lem:propTildePhi} and \ref{lem:propTildePhi:Bimod} in Lemma~\ref{lem:propTildePhi} can be combined to conclude  $\tilde\Phi(A) = \tilde\Phi(E_{\Delta^4_{\Cl(d)}} (A)) = E_{\Delta^4_{\U(d)}}(A) \tilde\Phi(\Id)$ and, thus, 
	$ E_{\Delta^4_{\Cl(d)}}(A) = \tilde\Phi(A) \tilde\Phi(\Id)^{-1}$. The expression for $\tilde\Phi(\Id)^{-1}$ was derived in Lemma~\ref{lem:propTildePhi}, Eq.~\eqref{lem:propTildePhi:Id}. Together with the definition of $\tilde\Phi$ the expression of the theorem follows after some simplification. 
\end{proof}

\subsection{The second moment}\label{subsec:secondmoment}
The main result of this section is the following expression for the second moment of $S_{\mc T}$ defined in Eq.~\eqref{eq:SmcT}.
We shall use this statement multiple times in the proofs of our main results. 
\begin{lemma}[The $2$-nd moment for $\U(d)$]\label{lem:secondMoment}
Let $\mathcal{T}:H_d \to H_d$ be a map. Then
\begin{equation}
		\begin{aligned}
		&\EE_{U \sim \operatorname{Haar}(\U(d))}[S^2_{\mc T}]  \\
		&= \frac{1}{d^2 -1} \Big\{ 
		d^2 \norm{\mc T}^2 + \Tr(\mathcal{T}(\Id))^2 \\& \quad\quad\quad\quad\quad - \frac{1}{d} \left(\fnorm{\mc T(\Id)}^2 +\fnorm{\mc T\ad(\Id)}^2\right)\Big\},
	\end{aligned}
\end{equation}
for $S_\mc T$ defined in Eq.~\eqref{eq:SmcT}.
\end{lemma}

For trace-annihilating and $\Id$-annihilating maps, one arrives at a much simpler expression:
\begin{corollary}[Expression for trace-annihilating and $\Id$-annihilating maps]\label{lem:secondMomentBound}
Let $\mathcal{T} \in \operatorname{V}_\utpn$ be a map that is trace-annihilating and $\Id$-annihilating. Then the second moment of $S_{\mathcal{T}}$ is
	\begin{equation}
		\EE_{U \sim \operatorname{Haar}(\U(d))}[S^2_{\mc T}] 
		= \frac{d^2}{d^2 -1} 
		\norm{\mc T}^2 .
	\end{equation}
\end{corollary}
\begin{proof}
	 This follows directly from Lemma~\ref{lem:secondMoment} and the observation that $\mathcal{T}$ being  trace-annihilating translates to $\Tr(\mc T(\Id))) = 0$ and $\fnorm{\mc T\ad(\Id)} = 0$ and  $\mathcal{T}$ being $\Id$-annihilating further requires $\fnorm{\mc T(\Id)} = 0$. 
\end{proof}

Before proving Lemma~\ref{lem:secondMoment}, we derive a general expression for the $k$-th moment of $S_{\mc T}$.  To this end, recall that by Choi's theorem an endomorphism $\mc T$ of $H_d$ (i.e.\ a hermiticity preserving map) can be decomposed as 
\begin{equation}\label{eq:mapdecomposition}
	\mc{T}(X) = \sum_{i = 1}^r \lambda_i T_i X T\ad_i,
\end{equation}
where $\lambda_i \in \RR$ and  $T_1,\ldots,T_r$ are linear operators with unit Frobenius norm.
In this decomposition, the random variable 
$S_{\mc T}$ from Eq.~\eqref{eq:SmcT}, with $\mc U(X) = UXU\ad $ takes the form 
\begin{equation}
	S_{\mc T} = d^2 (\mc T, \mc U) = \sum_{i=1}^r \lambda_i | \Tr(U\ad T_i)|^2 
\end{equation}
and its $k$-th moment can be expressed as follows:
\begin{lemma}[$k$-th moment of $S_{\mc T}$]\label{lem:generalKthMoment}
	For $k \in \mathbb{N}$ and $T_i$ defined by Eq.~\eqref{eq:mapdecomposition}
	we have
	\begin{equation}
	\begin{aligned}
		&\EE_{U \sim \operatorname{Haar}(\U(d))}[S^k_{\mc T}]  \\
			&=\sum_{i_1, \ldots, i_k = 1}^r \lambda_{i_1} \cdots \lambda_{i_k}  
		 \frac{1}{k !} \sum_{\tau \in S_k}   \sum_{\lambda \vdash k,\ l(\lambda) \leq d} \frac{d_\lambda}{D_\lambda} \\
		 & \quad 
		 \times\Tr\left[
		 \bigotimes_{j = 1}^k T\ad_{i_{\tau(j)}} 
		  P_\lambda \bigotimes_{j = 1}^k T_{i_j} \right].
	\end{aligned}
	\end{equation}
\end{lemma}

\begin{proof}
	We can rewrite the $k$-th unitary moment of $S_{\mc T}$ as 
	\begin{equation}
	\begin{aligned}
		&\EE_{U \sim \operatorname{Haar}(\U(d))}[S^k_{\mc T}]  \\
			&= \EE_{U}\sum_{i_1, \ldots, i_k = 1}^r \lambda_{i_1} \cdots \lambda_{i_k} |\Tr(U\ad T_{i_1})|^2\cdots  |\Tr(U\ad T_{i_k})|^2 \\
			&=  \EE_{U} \sum_{i_1, \ldots, i_k = 1}^r \lambda_{i_1}\cdots \lambda_{i_k} \\
			&\quad\quad\times \Tr\left[\bigotimes_{j = 1}^k T_{i_j}\ad\, U^{\otimes k} \right] \Tr\left[ {U\ad}^{\otimes k} \bigotimes_{j = 1}^k T_{i_j}\right] \\
			&= \sum_{i_1, \ldots, i_k = 1}^r \lambda_{i_1} \cdots \lambda_{i_k} \\
			&\quad\quad\times \sum_{m,n = 1}^{d^k} \bra{m} \bigotimes_{j = 1}^k T\ad_{i_j} E_{\Delta_{\U(d)}^k}(\ketbra{m}{n}) \bigotimes_{j = 1}^k T_{i_j} \ket{n} 
	\end{aligned}
	\end{equation}
	where in the last line we evaluated the trace in an orthonormal basis $\{ \ket{m} \mid m \in \{ 1, \ldots, d^k\}\}$ for $(\CC^d)^{\otimes k}$.
Using the expression for $E_{\Delta^k_{\U(d)}}$ of Theorem~\ref{thm:intU} we get 
	\begin{equation}
	\begin{aligned}
		&\EE_{U \sim \operatorname{Haar}(\U(d))}[S^k_{\mc T}]  \\
		&= \sum_{i_1, \ldots, i_k = 1}^r \lambda_{i_1} \cdots \lambda_{i_k}  
		 \frac{1}{k !} \sum_{\tau \in S_k}   \sum_{\lambda \vdash k,\ l(\lambda) \leq d} \frac{d_\lambda}{D_\lambda} \\
		 & \qquad \times\Tr\left[\pi^d_{S_k}(\tau) 
		 \bigotimes_{j = 1}^k T\ad_{i_j} \pi^d_{S_k}(\tau^{-1})
		  P_\lambda \bigotimes_{j = 1}^k T_{i_j} \right] \\
		  &= \sum_{i_1, \ldots, i_k = 1}^r \lambda_{i_1} \cdots \lambda_{i_k}  
		 \frac{1}{k !} \sum_{\tau \in S_k}   \sum_{\lambda \vdash k,\ l(\lambda) \leq d} \frac{d_\lambda}{D_\lambda} \\
		 & \qquad \times\Tr\left[
		 \bigotimes_{j = 1}^k T\ad_{i_{\tau(j)}} 
		  P_\lambda \bigotimes_{j = 1}^k T_{i_j} \right].
	\end{aligned}
	\end{equation}	
\end{proof}

\begin{proof}[Proof of Lemma~\ref{lem:secondMoment}]
	\newcommand{\scranti}{{\Yboxdim{4pt}\,\yng(1,1)}}
	\newcommand{\scrsym}{{\Yboxdim{4pt}\,\yng(2)}}

	We evaluate the expression of Lemma~\ref{lem:generalKthMoment} for the case $k=2$. To this end recall that the irreducible representations of $S_2$ are the symmetric (\scrsym\,) and antisymmetric representation (\scranti\,). The central projections are given by $P_\scrsym = \frac{1}{2} ( 1 + \FF)$ and $P_\scranti = \frac{1}{2} ( 1 - \FF)$ \cite{FouHar91}, where $\FF$ is the bipartite flip operator $\mathbb{F}: (\CC^d)^{\otimes 2} \to (\CC^d)^{\otimes 2}$, $\ket x \otimes \ket y \mapsto \ket y \otimes \ket x$.
	The dimensions are $d_\scrsym = d_\scranti = 1$, $D_\scrsym = \frac{d (d-1)}{2}$
	and $D_\scranti = \frac{d (d+1)}{2}$. 
For $A,B \in H_d^{\otimes 2}$ we introduce the following short-hand notation 
	\begin{equation}
		\Gamma_{AB} \coloneqq  \sum_{i,j}^r \lambda_i \lambda_j  \Tr\left[A(T\ad_i\otimes T\ad_j ) B(T_i \otimes T_j)\right].
	\end{equation}
	Rearranging the terms in the first statement of the Lemma~\ref{lem:generalKthMoment} then yields
	\begin{align}
		&\!\!\!\!\!\! \EE_{U \sim \operatorname{Haar}(\U(d))}[S^2_{\mc T}]  
		\\
		&= 
		\frac{1}{4}\Bigg\{
		\left[\frac{1}{D_\scrsym} + \frac{1}{D_\scranti}\right]\left[\Gamma_{\Id \Id} + \Gamma_{\FF \FF} \right] 
		\\
		&\qquad
			+ \left[\frac{1}{D_\scrsym} - \frac{1}{D_\scranti}\right]\left[\Gamma_{\FF \Id} + \Gamma_{\Id \FF} \right]\Bigg\} 
		\\
		&= 
		\frac{1}{d^2-1} \left\{ \Gamma_{\Id \Id} + \Gamma_{\FF \FF} - \frac{1}{d} \left( \Gamma_{\Id \FF} + \Gamma_{\FF \Id}\right) \right\} . \label{eq:second_moment_aux1}
	\end{align}
	The four $\Gamma$-terms can be evaluated explicitly. For the first term, we obtain
	\begin{equation}
	\begin{aligned}
		\Gamma_{\Id \Id} &= \sum_{i,j = 1}^r \lambda_i \lambda_j \fnorm{T_i}^2 \fnorm{T_j}^2  \\
		 &=\left(\sum_i \lambda_i  \Tr( T_i \Id T_i^\dagger) \right)^2  \\
		 &= \Tr(\mc T(\Id))^2.
	\end{aligned}
	\end{equation}
	The second terms reads
	\begin{equation}
	\begin{aligned}
		\Gamma_{\FF \FF} &= \sum_{i,j = 1}^r \lambda_i \lambda_j |\Tr(T\ad_i T_j)|^2 \\
		&= d^2 \norm{\mc T}^2
	\end{aligned}
	\end{equation}
	and the third term can be written as 
	\begin{equation}
	\begin{aligned}
		\Gamma_{\FF \Id}
		&= \sum_{i,j=1}^r \lambda_i \lambda_j \tr \left( T\ad_i T_i T\ad_j T_j \right) \\
		&= \fnorm{\mathcal{T}\ad (\Id )}^2.
	\end{aligned}
	\end{equation}
Moreover, a computation that closely resembles this reformulation yields
	$\Gamma_{\Id \FF} = \fnorm{\mc T (\Id)}^2$ and the claim follows.
\end{proof}

\subsection{A fourth moment bound}\label{subsec:fourthmoment}

The main result of this section is an upper bound for the fourth moment of $S_{\mc T}$ when $\mc U$ is a Haar random Clifford operation. 
To gain some intuition,  let us first derive an upper bound on the fourth moment taken with respect to the full unitary group. Note that a similar bound has already been derived in Ref.~\cite{KimLiu16}. 
\begin{lemma}[$4$-th moment bound for $\U(d)$]\label{lem:fourthMomentU}
	Let $\mc T: H_d \to H_d$ be a map. Then for $S_\mc T$ defined in Eq.~\eqref{eq:SmcT}
	 \begin{align}
		\EE_{U \sim \operatorname{Haar}(\U(d))}[S^4_{\mc T}] 
		\leq C \tnorm{J(\mc T)}^4
	\end{align}
	with some constant $C > \frac{1}{3}$ independent of the dimension $d$. 
\end{lemma}

\begin{proof}
Applying Cauchy-Schwarz to an individual summand on the right hand side of Lemma~\ref{lem:generalKthMoment} yields for all $k$
	\begin{equation}
	\begin{aligned}
		\left|\Tr\left[
		 \bigotimes_{j = 1}^k T\ad_{i_{\tau(j)}} 
		  P_\lambda \bigotimes_{j = 1}^k T_{i_j} \right] \right|
		  &\leq \fnorm{P_\lambda \bigotimes_{j = 1}^k T_{i_{\tau(j)}}  }\fnorm{P_\lambda \bigotimes_{j = 1}^k T_{i_{j}}} \\
		  &\leq \fnorm{\bigotimes_{j = 1}^k T_{i_{\tau(j)}}  }\fnorm{\bigotimes_{j = 1}^k T_{i_{j}}} \\
		  &= \prod_{j=1}^k \fnorm{T_{i_j}}^2,
	\end{aligned}
	\end{equation}
which is independent of the permutation $\tau \in S_k$.
We may therefore conclude
	\begin{equation}\label{eq:umoments:upperbound}
	\begin{aligned}
		&\EE_{U \sim \operatorname{Haar}(\U(d))}[S^k_{\mc T}]  \\
			&\leq \sum_{i_1, \ldots, i_k = 1}^r    \prod_{j=1}^k  \left|\lambda_{i_j}\right|  \fnorm{T_{i_j}}^2 \sum_{\lambda \vdash k,\ l(\lambda) \leq d} \frac{d_\lambda}{D_\lambda}.
	\end{aligned}
	\end{equation}

	From Theorem~\ref{thm:SchurWeylClifford} we observe that for $k=4$
	\begin{equation}
		\sum_{\lambda \vdash 4,\ l(\lambda) \leq d } \frac{d_\lambda}{D_\lambda} \leq \frac{C}{d^4}
	\end{equation}
	for some constant $C > \frac{1}{3}$ independent of $d$. 
	Thus, Eq.~\eqref{eq:umoments:upperbound} implies the desired bound.
\end{proof}

In an analogous way we can derive a sufficient bound on the fourth moment of $S_{\mc T}$ when the average is performed over the Clifford group. The result will be stated in Lemma~\ref{lem:fourthMomentCliff}. To get the correct dimensional pre-factors in the bound, we have to rely on particular properties of the  projection $Q$ of Eq.~\eqref{eq:Qdef} appearing in the representation theory of the fourth order diagonal action of Clifford group in Theorem~\ref{thm:intCl}. The following technical result takes care of this issue.

\begin{lemma}[Properties of the projection $Q$]\label{lem:Qnormbound}
For $\{T_l\}_{l=1}^r \subset \End(\CC^d)$ and $Q$ defined in Eq.~\eqref{eq:Qdef}
\begin{equation}
	\fnorm{Q \bigotimes_{j=1}^4 T_{i_j} Q}
	\leq \frac{1}{d} \prod_{j=1}^4 \fnorm{T_{i_j}}.
\end{equation}
\end{lemma}

This bound is tight.
In fact, one can show that it is saturated if all $T_i$'s are chosen to be the same stabiliser state. The proof of Lemma~\ref{lem:Qnormbound} requires two other properties of multi-qubit Pauli matrices $W_1, \ldots, W_{d^2}$. The first property is summarised by the following lemma.

\begin{lemma}[Magnitude of multi-qubit Pauli matrices]\label{lem:PauliMagnitudeBound}
	For $A, B \in \End(\CC^d)$, 
	\begin{equation}
		\Tr(W_j A W_k B) \leq \fnorm{A}\fnorm{B}
	\end{equation}
	for all $j,k \in \{1, \ldots, d^2\}$. 
\end{lemma}
\begin{proof}
This statement follows directly from Cauchy-Schwarz and the unitary invariance of the Frobenius norm:
\begin{equation}
\begin{aligned}
\Tr \left( W_j A W_k B \right)
&= \left( B^\dagger, W_j A W_k \right) \\
&\leq \| B^\dagger \|_2 \| W_j A W_k \|_2 \\
&= \| B \|_2 \| A \|_2.
\end{aligned}
\end{equation}
\end{proof}

The second property is that the two multi-qubit flip operator $\mathbb{F}$ can be expanded in terms of tensor products of Pauli matrices.

\begin{lemma}[Multi-qubit flip operator in terms of Pauli matrices]\label{lem:flip}
	\begin{equation}
		\mathbb{F} = \frac{1}{d} \sum_{i=1}^{d^2} W_i^{\otimes 2}. 
	\end{equation}
\end{lemma}

\begin{proof}
The re-normalised Pauli matrices form an orthonormal basis of $H_d$:
\begin{equation}
X = \frac{1}{d} \sum_{k=1}^d W_k \Tr \left( W_k X \right) \quad \forall X \in H (\CC^n).
\end{equation}
We can extend this to a basis of $H_d^{\otimes 2}$ by considering all possible tensor products of Pauli matrices. Expanding the flip operator in this basis yields
\begin{equation}
\begin{aligned}
\mathbb{F}
&= \frac{1}{d^2} \sum_{k,l=1}^{d^2} W_k \otimes W_l \Tr \left( \mathbb{F} W_k \otimes W_l \right) \\
&= \frac{1}{d^2} \sum_{k,l=1}^{d^2} W_k \otimes W_l d \delta_{k,l} = \frac{1}{d} \sum_{k=1}^{d^2} W_k^{\otimes 2}
\end{aligned}
\end{equation}
as claimed.
\end{proof}
We are now equipped to prove Lemma~\ref{lem:Qnormbound}. 
\begin{proof}[Proof of Lemma~\ref{lem:Qnormbound}]
	We start by inserting the definition of $Q$, \eqref{eq:Qdef}. Fixing {w.l.o.g.} an order of the indices, we obtain 
	\begin{align}
		&\!\!\!\! \Tr\left[ Q \bigotimes_{j=1}^4 T_j Q \bigotimes_{j=1}^4 T_j\ad\right] \\
		&= \frac{1}{d^4} \sum_{k,l=1}^{d^2} \prod_{j=1}^4 \Tr\left[W_k T_j W_l T_j\ad\right] \\
		&= \frac{1}{d^4} \sum_{k,l=1}^{d^2}
		c_{k,l}(T_1)c_{k,l}(T_2)c_{k,l}(T_3)c_{k,l}(T_4),
		\label{eq:incs}
	\end{align}
	where we defined $c_{k,l}(T_j) \coloneqq \Tr(W_k T_j W_l T_j\ad) \in \mathbb{C}$. These numbers obey
	\begin{equation}
	\begin{aligned}
		\overline{c_{k,l}(T_j)} =& \overline{\Tr \left( W_k T_j W_l T_j^\dagger \right)}
		= \Tr \left( \left( W_k T_j W_l T_j^\dagger \right)^\dagger \right) \\
		=& \Tr \left( T_j W_l^\dagger T_j^\dagger W_k \right) = c_{k,l}(T^\dagger_j).
	\end{aligned}
	\end{equation}
	In addition, Lemma~\ref{lem:PauliMagnitudeBound} implies
	\begin{align}
		|c_{k,l}(T_j)|^2 =& \left| \Tr \left( W_k T_j W_l T_j^\dagger \right) \right|^2 \leq \| T_j \|_2^4.  \label{eq:csAbsoluteValueBound}
	\end{align}

	Equation~\eqref{eq:incs} can be viewed as a complex-valued inner product between two $d^2$-dimensional vectors indexed by $k$ and $l$. This expression can be upper bounded by the Cauchy-Schwarz inequality:
	\begin{align}
		& \frac{1}{d^4} \sum_{k,l=1}^{d^2} c_{k,l}(T_1) c_{k,l}(T_2) c_{k,l} (T_3) c_{k,l} (T_4) \\
 		&=\frac{1}{d^4} \sum_{k,l=1}^{d^2} \overline{ c_{k,l}(T_1^\dagger) c_{k,l}(T_2^\dagger)}
		c_{k,l}(T_3) c_{k,l}(T_4) \\
		&\leq \frac{1}{d^2} \sqrt{\frac{1}{d^2} \sum_{k,l} \left| c_{k,l} (T_1^\dagger) c_{k,l}(T_2^\dagger) \right|^2}
		\nonumber 
		\\
		&\quad\quad\times \sqrt{\frac{1}{d^2} \sum_{k,l} \left| c_{k,l} (T_3) c_{k,l}(T_4) \right|^2}. \label{eq:cssqrts}
	\end{align}
	The first square-root can be bounded in the following way
	\begin{equation}
	\begin{aligned}
		&\!\!\!\! \sqrt{\frac{1}{d^2} \sum_{k,l} \left| c_{k,l} (T_3) c_{k,l}(T_4) \right|^2} \\
		&\leq  \sqrt{\| T_1^\dagger \|_2^4 \frac{1}{d^2} \sum_{k,l} c_{k,l} (T_2^\dagger) } \\
		&= \| T_1 \|_2^2 \sqrt{ \frac{1}{d^2} \sum_{k,l} \Tr \left( W_k T_2^\dagger W_l T_2 \right)^2} \\
		&= \| T_1 \|_2^2 \sqrt{ \Tr \left( \frac{1}{d} \sum_k W_k^{\otimes 2} (T_2^\dagger)^{\otimes 2} \frac{1}{d} \sum_l W_l^{\otimes 2} T_2^{\otimes 2} \right)} \\
		&= \| T_1 \|_2^2	 \sqrt{ \Tr \left( \mathbb{F}\, (T_2^\dagger)^{\otimes 2}\, \mathbb{F}\, T_2^{\otimes 2} \right)} \\
		&= \| T_1 \|_2^2 \sqrt{ \Tr \left( T_2^\dagger T_2 \right)^2} = \| T_1 \|_2^2 \| T_2 \|_2^2.
	\end{aligned}
	\end{equation}	
	Here, we have applied the magnitude bound \eqref{eq:csAbsoluteValueBound} for $c_{k,l}(T_1^\dagger)$ in the second line and applied Lemma~\ref{lem:flip}. 

	The second square root can be bounded in a complete analogous fashion, i.e.
	\begin{equation}
		\sqrt{\frac{1}{d^2} \sum_{k,l} \left| c_{k,l} (T_3) c_{k,l}(T_4) \right|^2}
		\leq \| T_3 \|_2^2 \| T_4 \|_2^2.
	\end{equation}
	Inserting both bounds into Eq.~\eqref{eq:cssqrts} yields the desired claim.
\end{proof}

Having established Lemma~\ref{lem:Qnormbound}, we will now state the bound on the fourth moment of $S_{\mc T}$ when the average is performed over the Clifford group. 
\begin{lemma}[$4$-th moment bound for $\Cliff(d)$]\label{lem:fourthMomentCliff}
	Let $\mc T: H_d \to H_d$ be a map. For $S_\mc T$ defined in Eq.~\eqref{eq:SmcT}, it holds
	\begin{equation}
		\EE_{U \sim \operatorname{Haar}(\Cliff(d))}[S^4_{\mc T}]  
			\leq C \tnorm{J(\mc T)}^4,
	\end{equation}
where $\| \cdot \|_1$ denotes the trace (or nuclear) norm and  the constant $C>0$
	is independent of $d$.
\end{lemma}

\begin{proof}
	As for the unitary group, we can rewrite the $k$-th moment of $S_{\mc T}$ for the Clifford group as
	\begin{equation}
	 \begin{aligned}
		&\EE_{U \sim \operatorname{Haar}(\Cliff(d))}[S^k_{\mc T}]  \\
			&= \sum_{i_1, \ldots, i_k = 1}^r \lambda_{i_1} \cdots \lambda_{i_k} \sum_{m,n = 1}^{d^k} \\
			&\times \bra{m} \bigotimes_{j = 1}^k T\ad_{i_j} E_{\Delta_{\Cliff(d)}^k}(\ketbra{m}{n}) \bigotimes_{j = 1}^k T_{i_j} \ket{n} 
	\end{aligned}
	\end{equation} 
	using a basis $\{ \ket{m} \mid m \in \{ 1, \ldots, d^k\}\}$ for $(\CC^d)^{\otimes k}$.
	 The expression for $E_{\Delta^4_{\Cliff(d)}}$ with $k=4$ was derived in Theorem~\ref{thm:intCl}. It implies that
	 \begin{equation}
	 \begin{aligned}
		&\EE_{U \sim \operatorname{Haar}(\Cliff(d))}[S^4_{\mc T}]  \\
		  &= \sum_{i_1, \ldots, i_k = 1}^r \lambda_{i_1} \cdots \lambda_{i_k}  
		 \frac{1}{4 !} \sum_{\tau \in S_k}   \sum_{\lambda \vdash k,\ l(\lambda) \leq d} d_\lambda \\
		 & \times \Bigg\lbrace 
		 \frac{1}{D^+_\lambda}\Tr\left[Q
		 \bigotimes_{j = 1}^4 T\ad_{i_{\tau(j)}} Q
		  P_\lambda \bigotimes_{j = 1}^4 T_{i_j} \right] \\
		  &+  \frac{1}{D^-_\lambda}\Tr\left[Q^\perp
		 \bigotimes_{j = 1}^4 T\ad_{i_{\tau(j)}} Q^\perp
		  P_\lambda \bigotimes_{j = 1}^4 T_{i_j} \right] \Bigg\rbrace.
	\end{aligned}	
	\end{equation}
We may bound the first trace term by
	\begin{equation}
	\begin{aligned}
		&\left|\Tr\left[
		 Q\bigotimes_{j = 1}^4 T\ad_{i_{\tau(j)}} 
		  QP_\lambda \bigotimes_{j = 1}^4 T_{i_j} \right]\right|  \\
		  &\leq \fnorm{P_\lambda Q \bigotimes_{j = 1}^4 T_{i_{\tau(j)}} Q }\fnorm{P_\lambda Q\bigotimes_{j = 1}^4 T_{i_{j}}Q} \\
		  &\leq \fnorm{Q\bigotimes_{j = 1}^4 T_{i_{\tau(j)}} Q }\fnorm{Q\bigotimes_{j = 1}^4 T_{i_{j}}Q} \\
		 &\leq \frac{1}{d^2}\prod_{j=1}^4 \fnorm{T_{i_j}}^2,
	\end{aligned}
	\end{equation}
where we have used Cauchy-Schwarz and applied 
Lemma~\ref{lem:Qnormbound} in the last line. 
For the second trace term a looser bound suffices:
	\begin{equation}
	\fnorm{Q^\perp \bigotimes_{j = 1}^k T_{i_{\tau(j)}} Q^\perp} \leq \prod_{j=1}^k \fnorm{T_{i_j}}
	\end{equation}
	 for all $\tau \in S_4$. This follows directly from Cauchy-Schwarz.
	Altogether we conclude that
	\begin{equation}
	\begin{aligned}
		&\EE_{U \sim \operatorname{Haar}(\Cliff(d))}[S^4_{\mc T}]  \\
			&\leq \sum_{i_1, \ldots, i_4 = 1}^r    \prod_{j=1}^4  |\lambda_{i_j}|  \fnorm{T_{i_j}}^2 \sum_{\lambda \vdash k,\ l(\lambda) \leq d} d_\lambda \left[ \frac{1}{d^2D^+_\lambda}+\frac{1}{D^-_\lambda}\right] \\
			& \leq C \tnorm{J(\mc T)}^4
	\end{aligned}
	\end{equation} 
	with some constant $C > 0$ independent of $d$. The last step follows from the dimensions given in Theorem~\ref{thm:SchurWeylClifford}. 
\end{proof}

\subsection{Proof of Theorem~\ref{thm:recoveryguarantee} (recovery guarantee)}\label{ssec:proof_of_the_recovery_guarantee}

We consider the following measurements: For a map $\mc X \in \Lin(H_d)$ the measurement outcomes $f \in \RR^m$ are given by 
\begin{equation}
\begin{aligned}\label{eq:FmeasurementsExplicit}
	f_i &= \Favg(\mc C_i,\mc X) + \epsilon_i  \\
		&= \frac{1}{d+1} \left[ d\, ( \mc C_i, \mc X) 
		+ \frac{1}{d} \Tr(\mc X\ad ( \Id))\right]  + \epsilon_i,
\end{aligned}
\end{equation}
where $\mathcal C_i$ are random Clifford channels and $\epsilon \in \RR^m$ accounts for additional additive noise. 

To make use of the proof techniques developed for low rank matrix reconstruction \cite{KueRauTer15,KabKueRau15}, we will in the following work in the Choi representation of channels.
This has the advantage, that the Kraus rank directly translates to the familiar matrix rank. We define the Choi matrix of a map $\mc X \in  \Lin(H_d)$ as 
\begin{equation}\label{eqn:choi}
	J(\mc X) = (\mc X \otimes \Id)( \ketbra{\psi}{\psi}),
\end{equation}
where $| \psi \rangle =d^{-1/2} \sum_{k=1}^d \ket k \otimes \ket k \in \CC^d \otimes \CC^d$ is the maximally entangled state vector. 
The Choi matrix of a map is positive semi-definite if and only if the map is completely positive. We denote the cone of positive semi-definite matrices by $\operatorname{Pos}_{d^2}$.
A channel $\mc X$ is trace-preserving and unital if 
and only if both partial traces of the Choi matrix yield the maximally mixed state, i.e.\ $\Tr_1(J(\mc X))= \Tr_2(J(\mc X)) = \Id/d$. 
We will denote the set of Choi matrices that correspond to channels in $\Lin_\utp$ by $J(\Lin_\utp)$. 
Furthermore, we define $J(\operatorname{V}_\utpn)$ as the set of Choi matrices corresponding to trace- and identity-annihilating channels, i.e., both partial traces of operators in $J(\operatorname{V}_\utpn)$ vanish.
Moreover, recall that the inner product on $\Lin_\utp$ we introduced in \eqref{eq:inner_product} coincides with the Hilbert-Schmidt inner product of the corresponding Choi matrices \eqref{eq:inner_product_choi}.
Adhering to this correspondence, we slightly abuse notation and 
use $(\mc X, \mc Y)$ and $(J(\mc X), J(\mc Y))$ interchangeably. 

To formalise the robustness of our reconstruction we need to introduce the following notation. 
For a Hermitian matrix $Z\in H_d$ let $\lambda$ be the largest eigenvalue with an eigenvector $v$. 
We write $Z|_1 = \lambda \ketbra v v$ for the best unit rank approximation to $Z$
and $Z|_c \coloneqq Z-Z|_1$ denotes the corresponding ``tail''. 

In terms of the Choi matrix of $\mc X$ the measurement outcomes $f \in \RR^m$ read
\begin{align}\label{eq:ChoiMeasurements}
	f_i &= \frac{1}{d+1} \left[ d\, ( J(\mc C_i), J(\mc X)) 
		+ \Tr(J(\mc X))\right]  + \epsilon_i,
\end{align}
The underlying linear measurement map $\mathcal A: H_{d^2} \to \mathbb{R}^m$ is given by 
\begin{equation}\label{eq:measurements}
	\mathcal A_i(X) = \frac{1}{d+1}\left[d (J(\mc C_i),X) + \Tr(X)\right].
\end{equation}
 Since unital and trace preserving maps $\mc X$ have trace normalised Choi matrices the second trace-term of the measurement map is just a constant shift. We also define the set of measurement matrices $\{A_i\}_{i=1}^b$ that encode the measurement map as $\mc A_i(X) = (A_i, X)$:
$A_i = \frac{d}{d+1}\left[ J(\mc C_i) + \Id/d\right]$, where each $\mathcal{C}_i$ is a gate that is chosen uniformly at random (according to the Haar measure) from the multi-qubit Clifford group. 

In the Choi representation, we want to consider the optimisation problem
\begin{equation}\label{eq:Choi-algorithm}
\begin{split}
	\operatorname*{minimise}_Z \quad &\norm{\mathcal{A}(Z)- f}_{\ell_q} \\
	\operatorname{subject\ to}\quad &Z \in J(\Lin_\utp) \cap \operatorname{Pos}_{d^2},
\end{split}
\end{equation}
where we allow the minimisation of an arbitrary $\ell_q$-norm. The optimisation problem \eqref{eq:algorithm} is equivalent to \eqref{eq:Choi-algorithm} for $q=2$. 

We are interested in using the optimisation procedure \eqref{eq:Choi-algorithm} for the recovery of unitary quantum channels.  
In this section, we will derive the following recovery guarantee:

\begin{theorem}[Recovery guarantee]\label{thm:Choi-recovery}
	Let $\mc A: H_{d^2} \to \RR^m$ be the measurement map  \eqref{eq:measurements} with 
	\begin{equation}
		m \geq c d^2 \log(d).
	\end{equation}
	Then, for all  $X \in J(L_\utp) \cap \operatorname{Pos}_{d^2}$ given 
	 noisy observations $f = \mc A(X) + \epsilon \in \RR^m$, the minimiser $Z^\sharp$ of the optimisation problem \eqref{eq:Choi-algorithm} fulfils for $p \in \{1, 2\}$
	\begin{equation}
\begin{aligned}\label{eq:unitary_recerror}
	\norm{Z^\sharp - X}_p  
		\leq \tilde{C}_1  \tnorm{X|_c}  +  2 \tilde{C}_2 d^2 m^{-1/q} \norm{\epsilon}_{\ell_q}
\end{aligned}
\end{equation}
	with probability at least $1-\e^{-c_f m}$ over the random measurements. 
	The constants $\tilde C_1, \tilde C_2, c, c_f>0$  only depend on each other.
\end{theorem}

The recovery guarantee of Theorem~\ref{thm:recoveryguarantee} is the special case of Theorem~\ref{thm:Choi-recovery} for $q=2$ and $p=2$ restricted to measurements of a unitary quantum channel. In contrast, the more general formulation of Theorem~\ref{thm:Choi-recovery} allows for a violation of the unit rank assumption. The first term \eqref{eq:unitary_recerror} is meant to absorb violations of this assumption into the error bound.
We note in passing that the choice of $p=1$ actually yields a tighter bound compared to $p=2$. 

More generally, one can ask for a recovery guarantee if the measured map $X$ can not be guaranteed to be unital or trace preserving. From Eq.~\ref{eqn:utpprojection} one observes that as long as the map $X$ is trace normalised the measured \AGFs\ are identical to the average fidelities of the projection $X_\utp$ of $X$ onto the affine space of unital and trace-preserving maps. But since $X_\utp$ is not necessarily positive, it is not straight-forward to apply Theorem~\ref{thm:Choi-recovery} to $X_\utp$. We expect the reconstruction algorithm to recover the trace-preserving and unital part of an arbitrary map. The reconstruction error \eqref{eq:unitary_recerror} is expected to  additionally feature a term proportional to the distance of $X$ to the intersection of $L_\utp$ with the cone $\operatorname{Pos}_{d^2}$ of positive semi-definite matrices.

Another way to proceed is to use a trace-norm minimisation subject to unitality, trace-preservation and the data constraints $\norm{\mathcal A(Z) - f}_{\ell_q} < \eta$. The derivation of Theorem~\ref{thm:Choi-recovery} readily yields a recovery guarantee for the trace-norm minimisation that is essentially identical to Theorem~\ref{thm:Choi-recovery}. See Ref.~\cite{KabKueRau15} for details on the argument. The main difference is that such a recovery guarantee does not need to assume complete positivity of the map that is to be reconstructed. Correspondingly, the result of the trace-norm minimisation is not guaranteed to be positive semi-definite. This implies that the robustness of this
algorithm against violations of the unitality and trace-preservation is different compared to \eqref{eq:Choi-algorithm}. For example, the \AGFs\ of a not necessarily unital or trace-preserving map $\mc X$ to unitary gates coincide with the \AGFs\ of its unital and trace-preserving part $\mc X_\utp$ as long as $X$ is still normalised in trace-norm. This is a consequence of Eq.~\ref{eqn:utpprojection}. Thus, a trace-norm minimisation will reconstruct $X_\utp$ up to an error given by $\tnorm{J(\mc X_\utp)|_c}$ and noise.
We leave a more extensive study of the robustness of the discussed reconstruction algorithms against violations of this particular model assumption to future work. 

The proof of the recovery guarantee relies on establishing the so-called \emph{null space property (NSP)} for the measurement map $\mc A$. We refer to Ref.\ \cite{rauhut_book} for a history of the term. The  NSP  ensures injectivity, i.e.~informational completeness, of the measurement map $\mc A$ restricted to the matrices that should be recovered. Informally, for our purposes, a measurement map $\mathcal{A}: H_{d^2} \to \mathbb{R}^m$ obeys the NSP if no unit rank matrix in $J(\operatorname{V}_\utpn)$ is in the kernel (nullspace) of $\mathcal{A}$. 

\begin{definition}[Robust NSP, Definition~3.1 in Ref.~\cite{KabKueRau15}]\label{def:NSP}
$\mathcal{A}: H_{d^2} \to \mathbb{R}^m$ satisfies the \emph{null space property (NSP)} with respect to $\ell_q$ with constant $\tau >0$ if for all $X \in J(\operatorname{V}_\utpn)$ 
\begin{equation} \label{eq:NSPcondition} 
	\fnorm{ X|_1} \leq \frac{1}{2} \tnorm{X|_c} + \tau \| \mathcal{A} (X) \|_{\ell_q}.
\end{equation}
\end{definition}
The factor $1/2$ in front of the first term of \eqref{eq:NSPcondition} is only one possible choice. In fact, one can instead introduce a constant with value in $(0, 1)$. The constants appearing in Theorem~\ref{thm:Choi-recovery} then depend on the specific value of the pre-factor. In particular, the different choices of the pre-factor in the definition of the NSP result in different trade-offs between the constant $c$ that appears in the sampling complexity and the constant $\tilde C_1$ that decorates the model-mismatch term in the reconstruction error. For the simplicity, we leave these dependencies implicit. 

The main consequence of the NSP that we require is captured by the following reformulation of Theorem~12 of \cite{KabKueRau15}. 
\begin{theorem}\label{thm:NSPconsequence}
	Fix  $p \in \{1,2\}$ and let $\mathcal A: H_{d^2} \to \RR^m$ satisfy the NSP with constant  $\tau > 0$.  Then, for all $Y, Z \in J(\Lin_\utp)$
	\begin{equation}\begin{split}
		\norm{Z-Y}_p &\leq \frac{9}{2} \left[ \tnorm{Z} - \tnorm{Y} + 2 \tnorm{Y|_c} \right] \\
	 	&+ 7 \tau \norm{\mathcal{A}(Z - Y)}_{\ell_q}. 
	 \end{split}\end{equation} 
\end{theorem}

In fact, the measurement $\mathcal{A}$ of \eqref{eq:ChoiMeasurements} obeys the NSP. More precisely:
\begin{lemma}\label{lem:NSPforA}
	Let $\mathcal{A} : H_{d^2} \to \RR^m$ be the measurement map defined in \eqref{eq:measurements} with $m \geq c d^2 \log(d)$. Then $\mc A$ obeys the NSP property with constant $\tau = C^{-1} d(d+1)m^{-1/q}$ with probability of at least $1- \e^{-c_f m}$. The constants $C, c, c_f > 0$ only depend on each other.
\end{lemma}
The proof of Lemma~\ref{lem:NSPforA} is developed in the subsequent section.
\begin{proof}[Proof of Theorem~\ref{thm:Choi-recovery}]
With the requirements of Lemma~\ref{lem:NSPforA} we can apply Theorem~\ref{thm:NSPconsequence} and set  $Z = Z^\sharp$, the reconstructed result of the algorithm, as well as $Y = X$. 
The theorem's statement then reads
\begin{equation}\label{eq:rec:pnorm}
\begin{aligned}
	\norm{Z^\sharp - X}_p  
		&\leq 9 \tnorm{X|_c}  \\
		&+ 7\tau \norm{\mathcal{A}(Z^\sharp - X)}_{\ell_q}, \\
\end{aligned}
\end{equation}
because $\| X \|_1 = \| Z \|_1 = 1$ is true for arbitrary Choi matrices of (trace-preserving) quantum channels.
The second term is dominated by
\begin{equation}\label{eq:rec:secterm}
\begin{aligned}
 	\norm{\mathcal{A}(Z^\sharp- X)}_{\ell_q}  
 		&\leq  \left[ \norm{\mathcal{A}( X - Z^\sharp)+ \epsilon}_{\ell_q} +  \norm{\epsilon}_{\ell_q} \right]  \\
 		&\leq 2 \norm{\epsilon}_{\ell_q},
 \end{aligned}
\end{equation}
 where the last step follows from $Z^\sharp$ being the minimiser of \eqref{eq:Choi-algorithm}. Thus, we can replace it  by any point in the feasible set  including $X$ on the right hand side of the first line.
Inserting \eqref{eq:rec:secterm} and the NSP constants of Lemma~\ref{lem:NSPforA} into \eqref{eq:rec:pnorm} the assertion of the theorem follows. 
\end{proof}

In the remainder of this section, we will establish the NSP for our measurement matrix $\mathcal A$ as summarised in Lemma~\ref{lem:NSPforA}. 

\subsubsection*{Establishing the null space property}

To prove Lemma~\ref{lem:NSPforA} at the end of this section we start with deriving a criterion for the NSP property following the approach taken in Refs.~\cite{KabKueRau15,KliKueEis17}. 

\newcommand{\eurSet}{\ensuremath{\mathbf\Omega}}

\begin{lemma}\label{lem:NSPcriterion}
	A map $\mathcal{A}: H_{d^2} \to \RR^m$ obeys the null space property with respect to $\ell_q$-norm with constant $\tau > 0$ if 
	\begin{equation}
		\inf_{X \in \eurSet} \norm{\mathcal A(X)}_{\ell_1} \geq \frac{m^{1 - 1/q}}{\tau}
	\end{equation}
	with 
	\begin{equation}\nonumber
		\eurSet \coloneqq \{ Z \in J(\operatorname{V}_\utpn) \mid \fnorm{Z|_1} \geq 
		\frac{1}{2} \tnorm{Z|_c}, \fnorm{Z} = 1  \} \, .
	\end{equation}
\end{lemma} 

\begin{proof}
For matrices $X$ with the property $\fnorm{X|_1} \leq \frac{1}{2} \tnorm{X|_c}$ the NSP condition \eqref{eq:NSPcondition} is satisfied independently of the map $\mathcal{A}$. 
Hence, to establish the NSP for a specific map $\mathcal{A}$ it suffice to show that the condition \eqref{eq:NSPcondition} holds for all $X \in \eurSet = \{ Z \in J(\operatorname{V}_\utpn) \mid \fnorm{Z|_1} \geq \frac{1}{2} \tnorm{Z|_c}, \fnorm{Z} = 1  \}$. The additional assumption of $\fnorm{Z} = 1$ is no restriction since both sides of \eqref{eq:NSPcondition} are absolutely homogeneous functions of the same degree. By definition, for all $X \in \eurSet$ we have $\fnorm{X|_1} \leq \fnorm{X} \leq 1$. Therefore, for $X \in \eurSet$ 
\begin{equation}
	\norm{\mc A(X)}_{\ell_q} \geq \frac{1}{\tau}
\end{equation}
implies the NSP condition \eqref{eq:NSPcondition}.
Using the norm inequality $\norm{x}_{\ell_q} \geq m^{1/q - 1} \norm{x}_{\ell_1}$ yields the criterion of the lemma. 
\end{proof}

Recall that every rank-$r$ matrix $X$ obeys $\tnorm{X}^2/\fnorm{X}^2 \leq r$. This motivates thinking of the matrices of $\eurSet$ as having \emph{effective unit rank} since the norm ratio bounded in $\mathcal{O}(1)$. More precisely, the following statement holds:
\begin{lemma}[Ratio of 1 and 2-norms] \label{lem:effective_low_rank}
Every matrix $X \in \eurSet$ has \emph{effective unit rank} in the following sense:
\begin{equation}
\frac{ \| X \|_1^2}{\| X \|_2^2} \leq 9.
\end{equation}
\end{lemma}

\begin{proof} From $\|X|_1\|_2\leq 1$ and the definition of $\eurSet$ it follows that 
$
	\|X|_1\|_2 + \frac{1}{2} \|X|_1\|_1 \leq \frac{3}{2}.
$
Hence
$
	\frac{1}{2} \|X|_1\|_2+ \|X|_1\|_1 \leq 3
$.
Therefore,
we have that
$
	\|X\|_1 \leq   \|X|_1\|_1+ \|X|_c\|_1\leq 3
$
from which the assertion follows, because every $X \in \eurSet$ has unit Frobenius norm.
\end{proof}

In summary, we want to prove a lower bound on the $\ell_q$-norm of the measurement outcomes for trace- and identity annihilating channels with effective unit Kraus rank. The proof uses Mendelson's small ball method. See Ref.~\cite[Lemma~9]{KliKueEis17} for details of the method as it is stated here, which is a slight generalisation of Tropp's formulation \cite{Tro15} of the original method developed in Refs.~\cite{Men14,KolMen14}. 
Mendelson's proof strategy requires multiple ingredients.  These necessary ingredients will become obvious from the following theorem, which can be found in Ref.~\cite{Tro15} and lies at the heart of the small ball method. 
\begin{theorem}[Mendelson's small ball method]
\label{lem:Mendelson}
Suppose that $\mathcal{A}$ contains $m$ measurements of the form $f_k = \Tr[A_kX]$ where each $A_k$ is an independent copy of a random matrix $A$. Fix $E \subseteq J(\operatorname{V}_\utpn)$ and $\xi >0$ and define
\begin{align}
W_m (E; A)
&\coloneqq 
\EE \left[ \sup_{Z \in E} \Tr \left( Z H \right) \right]
	\, , \quad  
	H = \frac{1}{\sqrt{m}} \sum_{k=1}^m \epsilon_k A_k, \label{eq:empirical_width}
	\\
Q_\xi (E; A ) 
&\coloneqq 
\inf_{Z \in E} \PP \left[ \left| \Tr \left[ A Z \right] \right| \geq \xi \right] \label{eq:marginal_tail_fct},
\end{align}
where the $\epsilon_k$'s are {i.i.d.} Rademacher random variables, {i.e.} are uniformly distributed in $\{-1,1\}$.  
Then, with probability  of at least $1- \e^{-2t^2}$, where $t \geq 0$,
\begin{equation*}
\inf_{Z \in E} \| \mathcal{A}(Z) \|_{\ell_1} \geq 
\sqrt m \left( \xi \sqrt{m} Q_{2\xi} (E; A) - 2 W_m (E; A) - \xi t \right).
\end{equation*}
\end{theorem}

A lower bound of $\norm{\mathcal A(X)}_{\ell_1}$ thus requires two main ingredients:
1.) a lower bound on the so-called \emph{mean empirical width} $W_m(E; A)$ and 2.) an upper bound on the so-called \emph{marginal tail function} $Q_{2\xi}(E;A)$. We will derive those bounds for $E=\eurSet$ and our measurement map $\mc A$ at hand.

\paragraph*{Bound on the mean empirical width.}
With a different normalisation the following statement is derived in Ref.~\cite{KimLiu15}. 
\begin{lemma}\label{lem:Wm}
Fix $d=2^n$ and suppose that the measurement matrices are given by $A_i = \frac{d}{d+1}\left[ J(\mc C_i) + \Id / d\right]$ with a gate $\mc C_i$  chosen uniformly from the Clifford group for all $i$. 
Also, assume that $m \geq  d^2 \log (d)$. 
Then
\begin{equation}
W_m (\eurSet,A)
\leq 
\frac{24}{d+1} \sqrt{\log(d)}.
\end{equation}
\end{lemma}

The proof is analogous to the one in Refs.\ \cite{KueRauTer15,KimLiu16,KliKueEis17}. 
In order to adjust the normalisation we provide a short summary. 
\begin{proof} 
For $Z \in \eurSet$ it holds that 
\begin{equation}
	(A_i, Z) = \frac{d}{d+1} (J(\mc C_i), Z).
\end{equation}
The constant shift by the identity matrix does not appear hear since every $Z \in \eurSet$ is trace-less. 
Thus, we can set $H= \frac{d}{\sqrt{m}(d+1)} \sum_{i=1}^m \epsilon_i J(C_i)$. 
Applying H{\"o}lder's inequality for Schatten norms to the definition of the mean empirical width yields
\begin{equation}
W_m(\eurSet, A) \leq \sup_{Z \in \eurSet} \| Z \|_1 \mathbb{E} \| H \|_\infty \leq 3\, \mathbb{E}\| H \|_\infty, \label{eq:Wm_aux1}
\end{equation}
where we have used the effective unit rank of $Z$, Lemma~\ref{lem:effective_low_rank}.
Also, the $\epsilon_i$'s in the definition of $H$ form a Rademacher sequence. 
The non-commutative Khintchine inequality, see e.g\ \cite[Eq.~(5.18)]{Ver10}, 
can be used to bound this sequence
\begin{align}
\mathbb{E}_{\epsilon_i,C_i} \| H \|_\infty
\leq &  \frac{d}{d+1} \sqrt{\frac{2 \log (2 d^2)}{m} \mathbb{E}_{C_i} \left\| \sum_{i=1}^m J(C_i)^2  \right\|_\infty } \label{eq:Wm_aux2}
\end{align}
and $J(C_i)^2 = J(C_i)$ further simplifies the remaining expression. Moreover, 
$
\mathbb{E} \left[ J(C_i) \right] = \frac{1}{d^2} \mathbb{I}$, $\| J(C_i) \|_\infty = 1
$
and a Matrix Chernoff inequality for expectations (with parameter $\theta =1$), see e.g.\ \cite[Theorem~5.1.1]{Tro12} implies
\begin{align}
\mathbb{E}_{C_i} \left\| \sum_{i=1}^m J(C_i) \right\|_\infty
\leq & \left( \mathrm{e}-1 \right) \frac{m}{d^2} + \log (d^2)
\leq 4 \frac{m}{d^2},
\end{align}
where the second inequality follows from the assumption $m \geq d^2 \log (d)$. 
Inserting this bound into Eq.~\eqref{eq:Wm_aux2} yields 
\begin{equation}
	\mathbb{E}_{\epsilon_i, C_i} \| H \|_\infty 
	\leq 
	\frac d{d+1} \sqrt{\frac{8 \log (2d^2)}{d^2}}
\end{equation}
and the claim follows from combining this estimate with the bound~\eqref{eq:Wm_aux1} and  $\log (2d^2) \leq 4 \log (d)$. 
\end{proof}

\paragraph*{Bound on the marginal tail function.}
Here, we establish an anti-concentration bound to the marginal tail function. 
The precise result is summarised in the following statement. 
\begin{lemma}
\label{lem:marginal_tail_function}
	Suppose the random variable $A \in H_d$ is given by 
	$A = \frac{d}{d+1}\left[ J(\mc C) + \Id /d \right]$,
	where $\mc C$ is a Clifford channel drawn uniformly from the Clifford-group $\Cl(d)$. For $0\leq \xi \leq \frac{1}{d(d+1)}$ it holds that
	\begin{equation}
		Q_\xi(\eurSet, A) 
		\geq 
		\frac{1}{\hat C} \left(1 - d^2(d+1)^2 \xi^2 \right)^2,
	\end{equation}
	 where $\hat C$ is the constant from Lemma~\ref{lem:fourIsSmallerThanTwoSquared}. 
\end{lemma}
This statement follows from applying the Paley-Zygmund inequality to the non-negative random variable $S_{\mc T}^2$ defined in Eq.~\eqref{eq:SmcT}.
For this purpose, we will make use of the bounds on the second and fourth moment of $S_{\mc T}$ derived in Section~\ref{subsec:secondmoment} and Section~\ref{subsec:fourthmoment}, respectively. 
In particular, we establish the following relation between the second and fourth moment of $S_{\mc T}$. This is one of the technical core result of this work.
\begin{lemma}\label{lem:fourIsSmallerThanTwoSquared} Let $\mc T \in V_\utpn{}$ be a map  with $J(\mc T)$ of effective unit rank, i.e.~$\fnorm{J(\mc T)}^2 \leq c \tnorm{J(\mc T)}^2$ with some constant $c > 0$,  then
	\begin{equation}
		\EE_{U \sim \operatorname{Haar}(\Cl(d))} [S_{\mc T}^4] \leq \hat C\, \EE _{U \sim \operatorname{Haar}(\Cl(d))}[S^2_{\mc T}]^2
	\end{equation}
	for some constant $\hat C$ independent of the dimension $d$.
\end{lemma}
\begin{proof}
	Since the Clifford group is a unitary $3$-design \cite{Zhu15,Web15}, Corollary~\ref{lem:secondMomentBound} implies 
	\begin{equation}
		\EE _{U \sim \operatorname{Haar}(\Cl(d))}[S^2_{\mc T}] \geq \fnorm{J(\mc T)}^2. 
	\end{equation}
	Furthermore, the effective unit rank assumption, $\tnorm{J(\mc T)}^2 \leq c \fnorm{J(\mc T)}^2$, together with Lemma~\ref{lem:fourthMomentCliff} yields for the fourth moment
	\begin{equation}
		\EE _{U \sim \operatorname{Haar}(Cl(d))}[S^4_{\mc T}] \leq \hat C \fnorm{J(\mc T)}^4
	\end{equation}
	for some constant $\hat C  = c C > 0$ independent of $d$. Combining these two equations, the statement of the proposition follows. 
\end{proof}

Note that with the help of Lemma~\ref{lem:fourthMomentU} one arrives at the same conclusion for the moments of $S_{\mc T}$ when the average is taken over the unitary group. This reproduces the previous technical core result of Ref.\ \cite{KimLiu16}. 
\begin{proof}[Proof of Lemma~\ref{lem:marginal_tail_function}]
	In the following we always understand by $\mc T$ the map in $\Lin(H_d)$ with Choi matrix $T = J(\mc T)$.  In terms of the random variable $S_{\mc T} = d^2 \, \Tr[T J(\mc C)]$, Eq.~\eqref{eq:SmcT}, the marginal tail function can be expressed as
	\begin{equation}
		Q_\xi(\eurSet ,A) =  \inf_{T\in \eurSet} \Pr\left[ \frac{|S_{\mc T}|}{d(d+1)} \geq \xi \right].
	\end{equation}
	 Here we again used that every $Z \in \eurSet$ is trace-less. Consequently, the shift by the identity matrix in the measurements $A_i$ vanishes.
	 Using Lemma~\ref{lem:fourIsSmallerThanTwoSquared}, the theorem follows by a straight-forward application of the Paley-Zygmund inequality,
	 \begin{equation}
	\begin{aligned}
			&\phantom{=.} \inf_{T \in \eurSet} \Pr\left[ \frac{1}{d(d+1)} |S_{\mc T}| \geq \xi \right] 
			\\
			&= 
			\inf_{T \in \eurSet} \Pr\left[ 
				\frac{1}{d^2(d+1)^2} S_{\mc T}^2 \geq \frac{\EE[S_{\mc T}^2] }{d^2(d+1)^2}\, {\tilde\xi}^2 ] 
				\right]  
			\\
			&\geq 
			(1 - \tilde \xi^2)^2 \frac{\EE[S_{\mc T}^2]^2}{\EE[S_{\mc T}^4]} \geq \frac{1}{\hat C} (1-\tilde \xi^2)^2,
	\end{aligned}
	\end{equation}
	where $\hat C > 0$ and 
	$\tilde \xi = \frac{d(d+1)}{\sqrt{\EE[S_{\mc T}^2]} } \, \xi$
	is required to fulfil $\tilde \xi \in [0,1]$. 
	According to Corollary~\ref{lem:secondMomentBound} and the normalisation of $T \in \eurSet$ we have
	$\tilde \xi
	= 
	\frac{d(d+1)\, \xi}{\fnorm{T}} = d(d+1)\, \xi$. 
\end{proof}

\paragraph*{Completing the proof of Lemma~\ref{lem:NSPforA}}
We are finally in position to deliver the proof for the NSP of $\mathcal{A}$. With the bounds on the mean empirical width, Lemma~\ref{lem:Wm}, and the marginal tail function, Lemma~\ref{lem:marginal_tail_function}, Mendelson's small ball method, Theorem~\ref{lem:Mendelson}, yields the following lemma:

\begin{lemma}
\label{lem:NSPCritforA}
Suppose that $\mathcal{A}$ contains 
\begin{equation}\label{eq:m_0}
	m \geq m_0 = c \, d^2 \log(d)
\end{equation}
measurements of the form 
$f_k = \Tr[A_k X]$ 
where each 
$A_k = \frac{d}{d+1} 
\left[J(\mc C_i) + \Id/d\right]$ is given by an independent and uniformly random Clifford unitary channel $\mc C_i$. 
Fix 
$\eurSet \subset J(\operatorname{V}_\utpn)$ 
as defined in Lemma~\ref{lem:NSPcriterion}.
Then 
\begin{equation}
\inf_{Z \in \eurSet} \| \mathcal{A}(Z) \|_{\ell_1} 
\geq 
C \frac{m}{d(d+1)}
\end{equation}
with probability at least $1-\e^{-c_f m}$ over the random measurements. 
The constants $C,c,c_f>0$ only depend on each other. 
\end{lemma}

\begin{proof}
Combining the Lemmas~\ref{lem:Mendelson}, \ref{lem:Wm}, and~\ref{lem:marginal_tail_function} yields with probability at least $1-\e^{-2t^2}$ that 

\begin{equation}
\begin{aligned}
	& \phantom{=.} 
	\inf_{\mc Z \in \eurSet} \| \mathcal{A}(\mc Z) \|_{\ell_1}
	\\
	&\geq 
	\sqrt m \left( \frac{\xi \sqrt m}{\hat C} \left( 1-(d(d+1)\xi)^2 \right)^2 - \frac{48}{d+1} \sqrt{\log(d)} - \xi t \right)
	\\
	&\geq
	\frac{\sqrt m}{d+1} \left( c_1 \frac{\sqrt{m}}{d} - 48 \sqrt{\log(d)} - \frac{t}{2d} \right)
\end{aligned}
\end{equation}
where we have chosen $\xi = \frac{1}{2d(d+1)}$. 
The statement follows from the scaling \eqref{eq:m_0} of $m$. 
\end{proof}

From Lemma~\ref{lem:NSPCritforA} and Lemma~\ref{lem:NSPcriterion} the assertion of Lemma~\ref{lem:NSPforA} directly follows. 

\subsection{Sample optimality in the number of channel uses} \label{sec:optimality}
The compressed sensing recovery guarantees, Theorem~\ref{thm:recoveryguarantee} and Theorem~\ref{thm:Choi-recovery}, focus on the minimal number of \AGFs\ $m$ that are required for the reconstruction of a unital and trace-preserving quantum channel using the reconstruction procedure \eqref{eq:algorithm} and \eqref{eq:Choi-algorithm}, respectively. This can be regarded as the number of measurement settings. But already the measurement of single fidelities up to some desired additive error will require a certain number of repetitions of some experiment. 
Therefore, to quantify the total measurement effort a more relevant figure of merit is the minimum number of channel uses $M$ required for taking all the data used in a reconstruction. 

We will show that the equivalent algorithms \eqref{eq:algorithm} and \eqref{eq:Choi-algorithm} reach an optimal parametric scaling of the required number of channel uses  in a simplified measurement setting.
To this end, we first combine the direct fidelity estimation protocol of Ref.~\cite{FlaLiu11} with our recovery strategy to provide an upper bound on the number of channel uses required for the reconstruction of a unitary gate up to a constant error. Second, following the proof strategy of Ref.~\cite[Section III]{FlaGroLiu12}, we derive a lower bound on the number of channel uses required by any POVM measurement scheme of \AGFs\ with Clifford gates and any subsequent reconstruction protocol that only relies on these \AGFs. 

\subsubsection{Measurement setting}
In order to obtain an optimality result we consider a measurement setting that is arguably simpler than the one in randomised benchmarking and more basic from a theoretical perspective. 
We consider a unitary channel $\mc U$ given by a unitary $U \in \U(d)$ and measurements given by Clifford channels $\mc C_i$ with $C_i\in \Cl(d)$. 
Using the identities \eqref{eq:inner_product} and \eqref{eq:inner_product_choi}
the \AGFs\ $\Favg(\mc C_i, \mc X)$ are determined by
\begin{equation}\label{eq:lower_bound_measurements}
	f_i 
	= 
	\left( J(\mc C_i), J(\mc X) \right)
	=
	\frac{1}{d^2} |\Tr[C_i U]|^2 \, .
\end{equation}
In this section, we consider $U / \sqrt{d}$ as a pure state vector in $\CC^{d} \otimes \CC^d$, i.e., as the state vector corresponding to the Choi state of the channel $\mc U$. 
This state can be prepared by applying the operation $U$ to one half of a maximally entangled state. 

\subsubsection{An upper bound from direct fidelity estimation} \label{sec:optimality:upperbound}

We will now derive an upper bound on the number of channel uses required in the reconstruction scheme \eqref{eq:Choi-algorithm}.
We note that our measurement values \eqref{eq:lower_bound_measurements} are also fidelities of the quantum state vectors $U/\sqrt{d}$ and $C_i/\sqrt d$ and use \emph{direct fidelity estimation} \cite{FlaLiu11} (see also \cite{SilLanPou11}) to estimate these fidelities. 
Importantly, each $C_i/\sqrt d$ is a stabiliser state and we view it as the ``target state'' in the direct fidelity estimation protocol \cite{FlaLiu11}. 
Then $C_i/\sqrt d$ is a \emph{well-conditioned state} with parameter $\alpha=1$. 
One of the main statements of Ref.~\cite{FlaLiu11} is that the fidelity $f_i$ can hence be estimated from $\mu \geq \mu_0$ many Pauli measurements, where
$\mu_0 \in \landauO\left( \frac{\log(1/\delta_0)}{\recerrorF^2} \right)$. Here, 
$\delta_0>0$ is the maximum failure probability,
and $\recerrorF>0$ is the accuracy up to which the fidelity $f_i$ is estimated. 
This implies that the estimation error is bounded as 
\begin{equation}
	\recerrorF \in \landauO\left( \frac{\sqrt{\log(1/\delta_0)}}{\sqrt{\mu_0}} \right) \, .
\end{equation}

For our channel reconstruction, we measure $m \in \tilde \landauO(d^2)$ many fidelities, each up to error $\recerrorF$, see Theorem~\ref{thm:recoveryguarantee}.
For a maximum failure probabilities of the single fidelity estimations $\delta_0$
and a desired failure probability $\delta$ of all the $m$ estimations it is sufficient to require
$\delta \leq m \delta_0$, since $(1-\delta_0)^{m} \geq 1-m \delta_0$. 
Moreover, in order for the reconstruction error 
\eqref{eq:errorbound}
to be bounded as 
\begin{equation}
	\hat C \frac{d^2}{\sqrt{m}}\norm{\epsilon}_{\ell_2} \leq \recerror \, ,
\end{equation}
where $\norm{\epsilon}_{\ell_2} \leq \sqrt{m}\, \recerrorF$,
we require
\begin{align}
	\hat C \frac{d^2}{\sqrt{m}} \norm{\epsilon}_{\ell_2}
	&\leq 
	C_2 d^2 \frac{\sqrt{\log(m/\delta)}}{\sqrt{\mu_0}} \leq \recerror \, .
\end{align}
Thus, a constant bound $\recerror$ of the reconstruction error can be achieved with a number of channel uses $M$ in 
\begin{equation}\label{eq:fidelity_estimation_bound}
\landauO\left( \frac{d^4 \log(m/\delta)}{\recerror^2} \right)
\subset
\tilde \landauO\left( \frac{d^4}{\recerror^2} \right) \, .
\end{equation}

\subsubsection{Information theoretic lower bound on the number of channel uses} \label{sec:optimality:lowerbound}

In this section we derive a lower bound on the number of channel uses that holds in a general POVM framework. Up to log-factors, it has the same dimensional scaling as the upper bound \eqref{eq:fidelity_estimation_bound} from direct fidelity estimation. 

We extend the arguments of Ref.~\cite[Section~III]{FlaGroLiu12} to prove a lower bound on the number of channel uses required for QPT of unitary channels from measurement values of the form \eqref{eq:lower_bound_measurements}. 
We consider each of these values to be an expectation value in a binary POVM measurement setting given by the unit rank projector 
$J(\mc C_i)$ are applied to the Choi state $J(\mc U)$. 
Then we are in the situation of \cite[Section~3]{FlaGroLiu12}, which proves a lower bound for the \emph{minimax risk} -- a prominent figure of merit for statistical estimators.

Let us summarise this setting. 
We denote by $\DM \subset H_d$ the set of density matrices and by $\meas$ the set of all two-outcome positive-operator-valued measurements (POVMs), each of them given by a projector $\pi \in H_d$. 
Next, we assume that we measure $M$ copies of an unknown state $\rho \in \DM$ in a sequential fashion.
By $Y_i$ we denote the binary random variable that is given by choosing the $i$-th measurement $\pi_i \in \meas$ and measuring $\rho$. 
These are mapped to an estimate $\hat \rho(Y_1, \dots, Y_M) \in H_d$. 
Any such estimation protocol is specified by the estimator function $\hat \rho$ and a set of functions $\{\Pi_i\}_{i \in [M]}$ that correspond to the measurement choices, where $\Pi_i(Y_1, \dots Y_{i-1}) \in \meas$, i.e., the $i$-th measurement choice $\Pi_i$ only depends on previous measurement outcomes. 
Let $\recerrorTr>0$ be the maximum trace distance error we like to tolerate between the estimation $\hat \rho$ and $\rho$. 
Then the \emph{minimax risk} is defined as 
\begin{equation}\label{eq:def:minimax_risk}
	R^\ast(M,\recerrorTr) 
	\coloneqq
	\inf_{\substack{\hat \rho\\ \Pi_1, \dots, \Pi_M }} \sup_{\rho \in \DM} \,
		\PP\left[ \TrNorm{\hat \rho(Y)- \rho}>\recerrorTr \right],
\end{equation}
where we denote by $Y$ the vector consisting of all random variables $Y_i$. An estimation protocol $(\hat \rho, \{\Pi_i\}_{i\in [M]})$ minimising the minimax risk has the smallest possible worst-case probability over the set of quantum states. 

The following theorem provides a lower bound on the mini\-max risk for the estimation of the Choi matrix of a unitary gate from unit rank measurements. 

\begin{theorem}[Lower bound, unit rank measurements]\label{thm:lower_bound_states}
Fix a set $\mc M$ of rank-$1$ measurements. 
For $\recerrorTr>0$ the minimax risk \eqref{eq:def:minimax_risk} of measurements of $M$ copies 
is bounded as
\begin{equation}
	R^\ast(M,\recerrorTr) 
	\geq 
	1- c_1 \frac{\log(d)\log(|\mc M|)}{d^4 (1-\recerrorTr/2)^2}M- \frac{c_2}{d^2(1-\recerrorTr^2)} \, ,
\end{equation}
where $c_1$ and $c_2$ are absolute constants. 
\end{theorem}

Before providing a proof for this theorem let us work out its consequences. If the measurements project onto Clifford unitaries, we get the following lower bound on the minimax risk. 

\begin{corollary}[Lower bound, Clifford group]
Let $\recerrorTr>0$ and consider measurements of the form \eqref{eq:lower_bound_measurements} given by Clifford group unitaries  on $M$ copies. 
Then the minimax risk \eqref{eq:def:minimax_risk} is bounded as
\begin{equation}
	R^\ast(M,\recerrorTr) 
	\geq 
	1- c_3 \frac{\log(d)^3}{d^4 (1-\recerrorTr/2)^2}M- \frac{c_2}{d^2(1-\recerrorTr^2)} \, ,
\end{equation}
where $c_3$ and $c_2$ are absolute constants. 
\end{corollary}

\begin{proof}
The cardinality of the $n$-qubit Clifford group ($d = 2^n$) is bounded as 
\begin{equation} \label{eq:cardinality:Cl}
	|\Cl(d)| 
	= 
	2^{n^2+2n}\prod_{j=1}^n(4^j-1)
	< 
	2^{2n^2+4n} \, 
\end{equation}
\cite{CalderbankEtAl:1996}. This implies that in case of our Clifford group measurements we have 
$\log(|\mc M|)< 2\log(d)^2 + 4 \log(d)$. 
\end{proof}

In every meaningful measurement and reconstruction scheme the minimax risk needs to be small. 
The corollary implies that, in the case of Cliffords, the number of copies $M$ need to scale with the dimension as
\begin{equation}
	M \in \Omega\left( \frac{d^4}{\log(d)^3} \right) \, ,
\end{equation}
where we have assumed $\recerrorTr>0$ to be small. 
This establishes a lower bound on the number of channel uses that every POVM measurement and reconstruction scheme requires for a guaranteed successful recovery of unitary channels from \AGFs\ with respect to Clifford unitaries. 

From the argument as it is presented here it is not possible to extract the optimal parametric dependence of the number of channel uses $M$ on the desired reconstruction error $\recerrorTr$. For quantum state tomography such bounds were derived in Ref.~\cite{HaaHarJi15} by extending the argument of Ref.~\cite{FlaGroLiu12} and constructing different $\recerrorTr$-packing nets. By adapting the $\recerrorTr$-packing net constructions of Ref.~\cite{HaaHarJi15} to unitary gates  
one might be able to derive  a optimal parametric dependence of $M$ on $\recerrorTr$. But it is not obvious how one can incorporate the restriction of the measurements to unit rank in the argument of Ref.~\cite{HaaHarJi15}. We leave this task to future work. 

In the remainder of this section we prove Theorem~\ref{thm:lower_bound_states}. 
The proof proceeds in two steps. At first we derive a more general bound on the minimax risk, Lemma~\ref{lem:minimax}, that follows mainly from combining Fano's inequality with the data processing inequality, see e.g.~\cite{CoverThomas:2012}. This is a slight generalization of Lemma~1 of Ref.~\cite{FlaGroLiu12} adjusted to the situation where the outcome probabilities of the POVM  measurements do not necessarily  concentrate around the value $1/2$. 
Lemma~\ref{lem:minimax} assumes the existence of an $\recerrorTr$-packing net for the set of unitary gates whose measurement outcomes are in a small interval to establish a lower bound on the minimax risk. 
Hence, in order to complete the proof, we have to establish the existence of a suitable packing net, Lemma~\ref{lem:net}, in a second step. Combining the general bound of Lemma~\ref{lem:minimax} and the existence of the packing net of Lemma~\ref{lem:net}, the proof of Theorem~\ref{thm:lower_bound_states} follows. 

We begin with the general information theoretic bound on the minimax risk. 

\begin{lemma}[Lower bound to the minimax risk]
\label{lem:minimax}
Let $\recerrorTr>0$ and 
$0 < \alpha < \beta \leq 1/2$.
Assume that there are states $\rho_1, \dots , \rho_s \in \Pos_{D}$ 
and orthogonal projectors $\pi_1, \dots, \pi_n \in \Pos_{D}$
such that 
\begin{align}\label{eq:minimax:net_anticon}
\TrNorm{\rho_i - \rho_j} \geq \recerrorTr
\\
\Tr[\pi_k \rho_i] \in [\alpha, \beta] \label{eq:minimax:net_nonbiased}
\end{align}
for all $i\neq j \in [s]$ and $k\in [n]$. 
Then the minimax risk \eqref{eq:def:minimax_risk} of $M$ single measurements is bounded as
\begin{equation}
	R^\ast(M,\recerrorTr) \geq 1- \frac{M (h(\beta)-h(\alpha))+1}{\log(s)} \, ,
\end{equation}
where $h$ denotes the binary entropy. 
\end{lemma}

\begin{proof}
We start by following the proof of \cite[Lemma~1]{FlaGroLiu12}:
Let $X$ be the random variable uniformly distributed over $[s]$ and let
$Y_1, \dots, Y_M$ be the random variables describing the $M$ single POVM measurements performed on $\rho_X$. 
Consider any estimator $\hat \rho$ of the state $\rho_X$ from the measurements $Y$ 
and define
\begin{equation}
	\hat X(Y) \coloneqq \argmin_{i\in [s]} \TrNorm{\hat \rho(Y) - \rho_i} \, .
\end{equation}
Then, for all $i \in [s]$,
\begin{equation}
	\PP[\TrNorm{\hat \rho(Y) - \rho_i} \geq \recerrorTr]
	\geq
	\PP[\hat X(Y) \neq X].
\end{equation}
Following Ref.~\cite{FlaGroLiu12}, we combine Fano's inequality and the data processing inequality for the mutual information $I(X;Z) = H(X) - H(X|Z)$, where $H$ denotes the entropy and conditional entropy, to obtain
\begin{align}
	\PP[\hat X(Y) \neq X] 
	&\geq 
	\frac{H(X|\hat X(Y))-1}{\log(s)} 
	\\
	&\geq 
	1-
	\frac{I(X;Y) +1}{\log(s)} \, .
\end{align}
Now we start deviating from Ref.~\cite{FlaGroLiu12}. 
We use that $I(X;Y)=I(Y;X)$, the chain rule, and the definition of the conditional entropy to obtain
\begin{align}
	&\phantom{\geq.}
	\PP[\hat X(Y) \neq X] 
	\\
	&\geq 
	1- \frac{ H(Y) -H(Y|X) +1 }{\log(s)} 
	\\
	&=
	1- \frac 1{\log(s)} \biggl( 
			\sum_{j=1}^M \Bigl\{ H(Y_j|Y_{j-1}, \dots, Y_1) 
			\\ 
			& \qquad \qquad 
			- \frac 1s \sum_{i=1}^s  H(Y_j|Y_{j-1}, \dots, Y_1, X=i) \Bigr\}
			+1 \biggr) .
\end{align}
Now we use that 
\begin{equation}
	H(Y_j|Y_{j-1}, \dots, Y_1, X=i)
	\geq 
	h(\alpha)
\end{equation}
and 
\begin{equation}
	H(Y_j|Y_{j-1}, \dots, Y_1) 
	\leq 
	h(\beta) \, ,
\end{equation}
where $h$ is the binary entropy, to arrive at
\begin{align}
	\PP[\hat X(Y) \neq X] 
	&\geq 
	1- \frac{ 
			M (h(\beta) -h(\alpha))+1 }{ \log(s) } 
	\\
	&\geq 
		1- \frac{ 
			M (h(\beta) -h(\alpha))+1 }{ \log(s) } .
\end{align}
\end{proof}

To apply Lemma~\ref{lem:minimax} we need to proof the existence of an $\recerrorTr$-packing net $\{\rho_i\}_{i=1}^s$ consisting of unitary quantum gates with the properties \eqref{eq:minimax:net_anticon} and \eqref{eq:minimax:net_nonbiased}. The construction of such a suitable $\recerrorTr$-packing net will use the fact that the modulus of the trace of a Haar random unitary matrix is a sub-Gaussian random variable. This can be
viewed as a non-asymptotic version of a classic result by Diaconis and Shahshahani \cite{DiaSha94}: the trace of a Haar random unitary matrix in $\U(d)$ is a complex Gaussian random variable in the limit of infinitely large dimensions $d$. 

\paragraph*{The trace of Haar random unitaries is sub-Gaussian.}
The statement follows from the fact that the moments of the modulus of the trace of a Haar random unitary are dominated by the moments of a Gaussian variable. 

\begin{proposition}\label{lem:S_Umoments}
	For all $d,k\in\ZZ_+$
	\begin{equation}
		\EE_{U \sim \operatorname{Haar}(\U(d))} \left[ |\Tr[U]|^{2k} \right] \leq k!,
	\end{equation}
	with equality if $k \leq d$. 
\end{proposition}

\begin{proof}
Denote by $S \coloneqq |\Tr(U)|^2$  the random variable with $U \in \U(d)$ drawn from the Haar measure. 
Let $\{\ket n\}_{n=1}^{d^k}$ be an orthonormal basis of $(\CC^d)^{\otimes k}$. 
The $k$-th moment of $S$ is given by 
\begin{align}
	\EE[S^k] &= \sum_{n, m=1}^{d^{k}} \sandwich{n}{U^{\otimes k}}{n}\sandwich{m}{(U\ad)^{\otimes k}}{m}.
\end{align}
Applying Theorem~\ref{thm:intU}, we get 
\begin{align}
	\EE[S^k] &= \frac{1}{k!} \sum_{n,m=1}^{d^k} \sum_{\tau \in S_k} \sum_{\lambda  \vdash k, l(\lambda )\leq d} \frac{d_\lambda }{D_\lambda } \\
	&\times \sandwich{m}{\pi^d_{S_k}(\tau )}{n} \sandwich{n}{\pi^d_{S_k}(\tau ^{-1})P_\lambda }{m}\\
	&= \frac{1}{k!} \sum_{\tau  \in S_k} \sum_{\lambda \vdash k, l(\lambda ) \leq d} \frac{d_\lambda }{D_\lambda } \Tr(\pi^d_{S_k}(\tau ) \pi^d_{S_k}(\tau ^{-1}) P_\lambda) \\
	&= \sum_{\lambda \vdash k, l(\lambda ) \leq d} \frac{d_\lambda }{D_\lambda } \Tr(P_\lambda ).
\end{align}
Since $\Tr(P_\lambda ) = d_\lambda D_\lambda $, we conclude
\begin{equation}
	\EE[S^k] = \sum_{\lambda \vdash k, l(\lambda ) \leq d} d_\lambda ^2 
	\leq \sum_{\lambda \vdash k} d_\lambda^2 = k!\;.
\end{equation}
The last equality can be seen from the orthogonality relation of the characters of the symmetric group, see e.\,g.\@ Ref.~\cite[Chapter~2]{FouHar91} for more details. Note that the second inequality is saturated in the case where $k \leq d$ since in this case the restriction $l(\lambda) \leq d$ is automatically fulfilled. 
\end{proof}

As a simple implication of the previous lemma is that the random variable $S=|\Tr(U)|^2$ has subexponential tail decay.

\begin{lemma}\label{lem:S_tail}
Let $S$ be a real-valued random variable that obeys $\mathbb{E} \left[ |S|^k \right] \leq k!$ for all $k \in \mathbb{N}$. Then, the right tail of $X$ decays at least subexponentially. For any $t \geq 0$,
\begin{equation*}
\mathbb{P} \left[ S \geq t \right] \leq \mathrm{e}^{-\kappa t+2},
\end{equation*}
with $\kappa = 1- \frac{1}{2 \mathrm{e}}$.
\end{lemma}

 This is a consequence of a standard result in probability theory that can be found in many textbooks, e.g.\ \cite{Ver12} and \cite[Section~7.2]{FouRau13}. We present a short proof here in order to be self-contained.

\begin{proof}
We use Markov's inequality, Proposition~\ref{lem:S_Umoments}, and Stirling's bound 
$k! \leq \e\,\sqrt{k}\, k^k\, \e^{-k}$ to obtain
for any $k \in \mathbb{N}$
\begin{align}
\PP[S\geq k] 
&\leq 
\frac{\EE[|S|^k]}{k^k}
\leq 
\frac{k!}{k^k}
\leq  
\mathrm{e} \sqrt{k} \mathrm{e}^{-k}.
\end{align}
In order to prove the tail bound, we choose $t \geq 0$ arbitrary and let $k$ be the largest integer that is smaller or equal to $t$ ($k = \lfloor t \rfloor$). Then
\begin{align*}
\mathrm{Pr} \left[ S \geq t \right]
\leq \mathrm{Pr} \left[ S \geq k \right]
\leq \mathrm{e} \sqrt{k} \mathrm{e}^{-k} 
\leq \mathrm{e}^{-\kappa k + 1} 
\leq \mathrm{e}^{-\kappa t + 1+\kappa}.
\end{align*}
Here, we have used $\sqrt{k} \mathrm{e}^{-k} \leq \mathrm{e}^{-\kappa k}$ and $t \leq k+1$. 
\end{proof}

Random variables with subgaussian tail decay -- \emph{subgaussian random variables} -- are closely related to random variables with subexponential tail decay: $X$ is subgaussian if and only if $X^2$ is subexponential.

Thus, Proposition~\ref{lem:S_Umoments} highlights that the trace of a Haar-random unitary is a subgaussian random variable. This is the aforementioned generalization of the classical result by Diaconis and Shashahani.

\paragraph*{A packing net with concentrated measurements.}

The proof of existence of an $\recerrorTr$-packing net to apply Lemma~\ref{lem:minimax} uses a probabilistic argument as in Ref.~\cite{FlaGroLiu12}. Here, the strategy is the following: We assume we are already given an $\recerrorTr$-packing net of a size $s-1$ that satisfies the desired concentration condition \eqref{eq:minimax:net_nonbiased}. We then show that a Haar random unitary gate also fulfils the concentration condition  and is $\recerrorTr$-separated from the rest of the net with strictly positive probability. Consequently, if one can be lucky to randomly arrive at a suitable $\recerrorTr$-packing net of size $s$ in this way then it must also exist. 

We start by deriving an anti-concentration result for the Choi matrix $J(\mc U)$  of a unitary channel given by a Haar random unitary $U$ in $\U(d)$.  

\begin{lemma}\label{lem:TrNormConcentration}
	Let $\mc V$ be a unitary gate. For all $\recerrorTr > 0$
	\begin{equation}
		\Pr_{U \sim \operatorname{Haar}(\U(d))}[\tnorm{J(\mc U) - J(\mc V)} \leq \recerrorTr ] 
		\leq  
		\e^{-\kappa d^2(1 - \recerrorTr/2)^2+2}
	\end{equation}
	with $\kappa>0$	being the constant from Lemma~\ref{lem:S_tail}. 
\end{lemma}

\begin{proof}
	Due to the unitary invariance of the trace-norm and the Haar measure, it suffice to show the statement for $\mc V = \Id$. 
	For a unitary channel with Choi-matrix $J(\mc U) = d^{-1} \operatorname{vec}(U) \operatorname{vec}(U\ad)^{t}$ and Kraus-operator $U \in \U(d)$ we have
	\begin{equation}
		\tnorm{J(\mc U) - J(\Id)} = 2 \sqrt{1 - \frac{1}{d^2}|\Tr(U)|^2} \geq 2\left( 1 - \frac{1}{d}|\Tr(U)|\right).
	\end{equation}
	For the first equation we calculate the set eigenvalues of $J(\mc U) - J(\Id)$, which is $\{ \pm \sqrt{1 - d^{-2}| \Tr(U)|^2}\}$. 
	Introducing the random variable  $S_U \coloneqq |\Tr(U)|^2$, we can rewrite the probability  as 
	\begin{align}
		\Pr[\tnorm{J(\mc U) - J(\Id)} \leq \recerrorTr] &\leq \Pr\left[2\left(1 - \frac{1}{d} \sqrt{S_U}\right) \leq \recerrorTr\right] \\
		&= \Pr\left[ S_U \geq d^2\left(1 -  \frac{\recerrorTr}{2}\right)^2 \right].
	\end{align}
	From Lemma~\ref{lem:S_tail} we know that 
	\begin{align}
		\Pr\left[S_U \geq  d^2\left(1 -  \frac{\recerrorTr}{2}\right)^2 \right] 
		\leq 
		\e^{-\kappa d^2(1 - \recerrorTr/2)^2+2}
	\end{align}
	from which the assertion follows.
\end{proof}

The anti-concentration result of Lemma~\ref{lem:TrNormConcentration} implies the existence of a large $\recerrorTr$-packing net $\mc N_\recerrorTr$ of unitary quantum channels. The desired concentration of the measurement outcomes can be established using Lemma~\ref{lem:S_tail}. In summary we arrive at the following assertion:

\begin{lemma}[Packing net with concentrated measurements]
\label{lem:net}
Let $0<\recerrorTr< 1/2$, $\kappa = 1-\frac{1}{2\e}$, and $C_1, \dots, C_K \in \U(d)$.
Then, for any number 
$s<\frac{1}{2}\e^{\kappa(1-\recerrorTr/2)^2 d^2-2}$,
there exist $U_1, \dots, U_s \in \U(d)$ such that
for all $i,j \in [s]$ with $i \neq j$ and for all $k \in [K]$
\begin{align}
\label{eq:PNcond}
	\TrNorm{J(\mc U_i ) - J(\mc U_j)} \geq \recerrorTr 
\, ,
\\
\label{eq:biasCond}
	\frac{1}{d^2} |\Tr[C_k^\dagger U_i]|^2 
	\leq  \frac{\log(2K) + 2}{\kappa d^2}
	\, .
\end{align}
\end{lemma}

\begin{proof}
As outlined above the existence of the described $\recerrorTr$-packing net follows inductively from the fact that if one adds a Haar random unitary gate $\mc U$ to an $\recerrorTr$-packing $\tilde{\mc N}_\recerrorTr$ of size $s-1$ that already fulfils all requirements of the lemma the resulting set $\tilde{\mc N}_\recerrorTr \cup \{\mc U\}$  has still a strictly positive probability to be an $\recerrorTr$-packing net with the desired concentration property \eqref{eq:biasCond}. 

We start with bounding the probability that the resulting set $\tilde{\mc N}_\recerrorTr \cup \{\mc U\}$ fails to be an $\recerrorTr$-packing net. Let us denote the probability that a Haar random $\mc U$ is not $\recerrorTr$-separated from $\tilde{\mc N}_\recerrorTr$ by $\bar{p}_\recerrorTr$. In other words, $\bar{p}_\recerrorTr$ is the probability that there exists $\mc V \in \tilde{\mc N}_\recerrorTr$ with 
\begin{equation}
	\tnorm{J(\mc U) - J(\mc V)} \leq \recerrorTr.
\end{equation}
Taking the union bound for all $\mc V \in \tilde{\mc N}_\recerrorTr$, Lemma~\ref{lem:TrNormConcentration} implies that 
\begin{equation}
	\bar{p}_{\recerrorTr} \leq s \e^{-\kappa d^2 ( 1 - \epsilon /2)^2 + 2}
\end{equation}
with $\kappa= 1-\frac{1}{2\e}$.
Thus, for $s < \frac{1}{2} \e^{-\kappa d^2 ( 1 - \epsilon /2)^2 + 2}$ we ensure that $\bar{p}_{\recerrorTr} < \frac{1}{2}$. 

We now also have to upper bound the probability $\bar{p}_c$ of $\mc U$ not having a concentration property 
\begin{equation}
	\frac{1}{d^2} |\Tr[C_k^\dagger U_i]|^2 
	\leq \beta
\end{equation}
 with respect to $K$ different unitaries $C_1, \ldots, C_K$. 

Using the unitary invariance of the Haar measure and taking the union bound, the tail-bound for the squared modulus of the trace of a Haar random unitary, Lemma~\ref{lem:S_tail}, yields 
\begin{equation}\label{eq:TrCU_concentration}
	\bar{p}_c \leq
	K \e^{-\kappa \beta d^2+2} 
\end{equation}
for $\beta \geq 2$. In order for $\bar{p}_c$ to be at most $1/2$, we need that 
\begin{align}
\beta \geq \frac{\log(2K) + 2}{\kappa d^2}.
\end{align}
In summary, we have established that $\bar{p}_\recerrorTr + \bar{p}_c < 1$ as long as $s < \frac{1}{2} \e^{-\kappa d^2 ( 1 - \epsilon /2)^2 + 2}$  and the achievable concentration is $\beta \geq (\log(2K) + 2)/(\kappa d^2)$. Hence, in this parameter regime there always exist at least one additional unitary gate extending the $\recerrorTr$-packing net. Inductively this proves the existence assertion of the lemma. 
\end{proof}

Having established a suitable $\recerrorTr$-packing net, we can now apply Lemma~\ref{lem:minimax} to derive the lower bound on the minimax-risk for the recovery of unitary gates from unit rank measurements of Theorem~\ref{thm:lower_bound_states}, the main result  of this section. 

\begin{proof}[Proof of Theorem~\ref{thm:lower_bound_states}]
We will apply Lemma~\ref{lem:minimax} with $\alpha=0$ and 
\begin{align}
	\beta = \frac{\log(2|\mc M|)+2}{\kappa d^2}
	\end{align}
and we use that 
$h(\beta) \leq 2\beta \log(1/\beta)$ for 
$\beta \leq 1/2$. 
Combining the Lemmas~\ref{lem:minimax} and~\ref{lem:net} we obtain 
\begin{align}
R^\ast(M,\recerrorTr)
&\geq 1- \frac{M h(c/d^2) + 1}
			 {(\kappa(1-\recerrorTr/2)^2 d^2+2)/\log(2) - 2}
\\
&\geq 1- \frac{2\frac{\log(2|\mc M|)+2}{\kappa d^2} \log\left(\frac{\log(2|\mc M|)+2}{\kappa d^2}\right)M + 1}
			 {d^2 (\kappa(1-\recerrorTr/2)^2 d^2+2)/\log(2) - 2} \, ,
\end{align}
where, in Lemma~\ref{lem:net} we have chosen $s$ to be the strict upper bound minus one. 
Finally, we simplify the bound by choosing large enough constants $c_1$ and $c_2$. 
\end{proof}

\subsection{Expansion of quantum channels in average gate fidelities}\label{sec:tight_frame} 
In this section, we give a instructive proof of the result of \cite{ScottDesigns} that the linear span of the \emph{unital} channels coincides with the linear span of the \emph{unitary} ones, even if one restricts to the unitaries from a unitary $2$-design. 
We also link this finding to \AGFs. On the way, we establish the simple formula of  Proposition~\ref{prop:affcomb} that allows for the reconstruction of unital and trace-preserving maps from measured \AGFs\ with respect to a arbitrary unitary $2$-design, e.g.~Clifford gates.  

In Lemma~\ref{lem:secondMoment} we derived an explicit expression for the second moment of the random variable $S_{\mc T} =d^2(\mc T, \mc U)$. For $\mc T \in \Lin_{\overline{\utp}}$, the linear hull of unital and trace-preserving maps, and $\mc U$ uniformly drawn from a unitary $2$-design the expression in fact indicates that a unitary $2$-design constitutes a Parseval frame for $\Lin_{\overline{\utp}}$. 
More abstractly, this observation stems from the general fact that irreducible unitary representations form Parseval frames on the space of endomorphisms of their representation space. For this reason it is instructive, to derive the connection explicitly in the `natural' representation-theoretic language. We begin with formalising the connection between irreducible representations and Parseval frames. 

\begin{lemma}[Irreps form a Parseval frame]\label{lem:irrep frame}
  Let $R: G \to \End(V)$ be an irreducible unitary representation of a group $G$.
  Then the set $\{\sqrt{\dim V}\,R(g)\}_{g\in G}$ forms a Parseval frame for the space $\End(V)$ equipped with the Hibert-Schmidt inner product $A, B \mapsto \Tr[A\ad B]$,
  in the sense that
  \begin{align}\label{eqn:projection algebra}
	  T_G(A)\coloneqq 
	  \dim(V)\,\int_G R(g) \Tr[R(g)\ad A]\, d\mu(g)  = A
  \end{align}
  for all $A\in \End(V)$.
\end{lemma}

\begin{proof}
	Since $\End(V)$ is generated as an algebra by $\{R(g)\}_{g\in G}$ (see e.g.\ \cite[Proposition 3.29]{FouHar91}), it suffices to show the statement for $A = R(g)$ with $g \in G$.
	Due to the invariance of the Haar measure, the map $T_G$ is covariant in the sense that $T_G(R(g) B) = R(g) T_G(B)$ for all $B \in \End(V)$. In particular, for $B=\Id$, we thus get $T_G(R(g)\Id)=R(g) T_G(\Id)$.
	With $\chi(g)=\tr R(g)$ the character of the representation, we have
	\begin{align}
		 T_G(\Id) = \dim(V)\,\int_G R(g) \bar\chi(g)\, d\mu(g)=\Id
	\end{align}
	from the well-known expression for projection onto a representation space in terms of the character, see e.g.~Ref.~\cite[Chapter 2.4]{FouHar91}.
 	Thus, we have established that $S_R(R(g))= R(g)$ for all $g \in G$. 
\end{proof}

Applying this lemma to unitary channels, we can derive the following expression for the orthogonal projection onto the linear hull of unital and trace-preserving maps. 
\begin{theorem}\label{thm:utpprojection}
	Let $\{ \mc U_k\}_{k=1}^N$ be a unitary $2$-design. The orthogonal projection onto the linear hull of unital and trace-preserving maps $\Lin_{\overline{\utp}}(H_d)$ is give by 
	\begin{equation}
		P_{{\overline{\utp}}}(\mc X) = \frac1N \sum_{k=1}^{N} c_{\mc U_k}(\mc X)\ \mc U_k
	\end{equation}
	with coefficients 
	\begin{align}
	c_{\mc U}(\mathcal{X}) 
	= C \Favg(\mc U, \mc X) - \frac{1}{d}\left(\frac{C}{d} - 1\right) \Tr(\mc X(\Id)) \, ,
\end{align}
where $C \coloneqq d(d+1)(d^2 -1)$. 
\end{theorem}
\begin{proof}
	Throughout the proof, we denote the unitary channel representing the unitary $U \in U(d)$ on space of Hermitian operators $H_d$ by $\mc U: \rho \mapsto U \rho\, U\ad$.
The vector space $H_d$ is a direct sum of the space $\mathcal{K}_0$ of trace-less hermitian matrices, and of $\mathcal{K}_1=\{ z \Id \}_{z\in \CC}$.
The group of unitary channels 
acts trivially on $\mathcal{K}_1$, and irreducibly on  
$\mathcal{K}_0$.
In particular, $\mathcal{U}$ is ``block-diagonal'' $\mathcal{U}=\mathcal{U}_0\oplus 1$ with respect to this decomposition, where $\mathcal{U}_0 \in \End(\mathcal{K}_0)$ is the irreducible $(d^2-1)$-dimensional block.
More generally, the projection of a map $\mathcal{X}$ onto the linear hull of unital and trace-preserving maps $\Lin_{\overline{\utp}}(H_d)$ is of the form $\mathcal{X}_0 \oplus x_1$. The map $\mathcal{X}_0 \oplus x_1$ is trace-preserving and unital if and only if $x_1 = \Tr(\mc X(\Id/d)) = 1$. 
For the map $\mc X \in \Lin(H_d)$ we have
\begin{align}
	\Tr[ \mc U^\dagger \mc X] 
	= 
	\Tr[\mathcal{U}_0^\dagger \mathcal{X}_0] + x_1.
\end{align}
Using this formula, Lemma~\ref{lem:irrep frame} for the choice $V=\mathcal{K}_0$, and the fact that a group integral over a non-trivial irrep vanishes \footnote{More explicitly, for $X \in H_d$ we can calculate $\int_{U(d)} \mc U(X)\, d\mu(U) = E^1_{\Delta_{U(d)}}(X) = \mathbb{O} \oplus 1$ using Theorem~\ref{thm:intU}.}
, we find
\begin{align}
	&\phantom{=.}(d^2-1)\int_{U(d)} \mathcal{U} \, \tr\left[\mathcal{U}^\dagger \mathcal{X}\right]  d\mu(U) \nonumber \\ 
	&=
	(d^2-1) 
	\int_{U(d)} 
		(\mathcal{U}_0\oplus 1)
		(
		  \Tr[
		  \mathcal{U}_0^\dagger
		  \mathcal{X}_0 ]
		  +x_1
		)
	\, d\mu(U)  \nonumber \\ \nonumber
	&=
	  (d^2-1) 
	  \int_{U(d)} 
		\mathcal{U}_0
		(
		\Tr[
		\mathcal{U}_0^\dagger
		\mathcal{X}_0 ]
		+x_1
		)
		\, d\mu(U) 
		\\
	&\phantom{=.}\oplus
	  (d^2-1) 
	  \int_{U(d)} 
		(
		\Tr[
		\mathcal{U}_0^\dagger
		\mathcal{X}_0 ]
		+x_1
		)
		\, d\mu(U)  
		\nonumber \\
	&=
	\mathcal{X}_0
	\oplus
	(d^2-1)x_1. \label{eq:integrate to zero}
\end{align}
Hence, for $\mc X \in \Lin_{\overline{\utp}}(H_d)$ we obtain the completeness relation
\begin{align}\label{eq:utpproj1}
\begin{split}
	\int_{U(d)} \mathcal{U} \, 
	&\left(
	(d^2-1)
	\Tr[\mathcal{U}^\dagger \mathcal{X}] +\frac{2 - d^2}{d} \Tr[\mc X(\Id)]\right)  d\mu(U) \\ 
	 &= \mc X.
	 \end{split}
\end{align}
For $\mc X$ in the ortho-complement of $\Lin_{\overline{\utp}}(H_d)$ the left hand side of Eq.~\eqref{eq:utpproj1} vanishes. The expression, thus, defines the orthogonal projection $P_{\overline{\utp}}$ onto $\Lin_{\overline{\utp}}$. The projection can be re-expressed in terms of the \AGF. 
With the help of Eqs.~(\ref{eq:inner_product}, \ref{eq:inner_product_choi}), 
\begin{equation}
\begin{aligned}
 	\Tr[\mathcal{U}^\dagger \mc X ]
	&= \left(\mc L(\mc U),\mc L (\mc X) \right) 
	= d^2 (\mc U, \mc X) \\
	&= d(d+1) \Favg(\mc U, \mc X) - \Tr(\mc X(\Id)). \label{eq:inner_product_trace}
\end{aligned}
\end{equation}
Hence,
\begin{align}
	\label{eqn:reproduce}
	P_{\overline{\utp}} (\mc X) = \int_{U(d)} c_{\mc U}(\mc X) \, \mathcal{U}\,  d\mu(U),
\end{align}
with expansion coefficients
\begin{align*}
	c_{\mc U}(\mathcal{X}) 
	&= d(d+1)(d^2-1) \Favg(\mc U,\mc X) \\
	&- \frac{1}{d}\left((d+1)(d^2-1)-1\right)\Tr(\mc X(\Id)) \\
	&= C \Favg(\mc U, \mc X) - \frac{1}{d}\left(\frac{C}{d} - 1\right) \Tr(\mc X(\Id)).
\end{align*}

Since the integrand in Eq.~\eqref{eqn:reproduce} is linear in $U^{\otimes 2}\otimes \bar U^{\otimes 2}$, the completeness relation continues to hold if the Haar integral is replaced by the average 
\begin{align}\label{eqn:reproduce design}
	\frac1N
	\sum_{k=1}^N 
	c_{\mc U_k}(\mc X)\,
	\mathcal{U}_{k} 
	=
	P_{\overline{\utp}}(\mathcal{X}) 
\end{align}
over any unitary $2$-design $\{\mc U_k\}_{k=1}^N$.
\end{proof}

In the proof, we have used that linear hull of the unital and trace-preserving maps $\Lin_{\overline{\utp}}$ is given by the space of block diagonal matrices $\Lin(\mathcal{K}_0)\oplus \Lin(\mathcal{K}_1)$.
If $\mathcal{X}$ is not unital and trace-preserving, the image $\mc X_{\overline{\utp}}$ will thus be equal to $\mathcal{X}$, with the off-diagonal blocks set to zero. 
In particular, the two-norm deviation of a map $\mathcal{X}$ from its projection onto $\Lin_{\overline{\utp}}$ is given by
\begin{equation}
\begin{split}
\label{eqn:utpprojection}
\| \mc X -  P_{\overline{\utp}}(\mc X) \|^2  = 
 \frac{1}{d^3}&
\biggl( \| \mc X (\Id
) \|_2^2 
\\
&+ \| \mc X^\dagger (\Id
 ) \|_2^2 - \frac{2}{d} \tr \left( \mc X (\Id
) \right)^2 \biggr).
\end{split}
\end{equation}

Based on the arguments used to establish Theorem~\ref{thm:utpprojection}, we can derive the following variant, which includes a converse statement.
\begin{theorem}[{Informational completeness and unitary designs}]\label{thm:affcomb_conv}
	Let 
	$\{ \mc U_k\}_{k=1}^N$ 
	be a set of unitary channels.
	Then the following are equivalent:

	\begin{enumerate}[label={(\roman*)}]
		\item
	Every unital and trace-preserving map $\mc X$ can be 
	written 
	as an affine 
	combination 
	$\mc X = \frac1N \sum_{k=1}^N  c_k(\mc X) \mc U_k$
	of the $\mc U_k$, with 
	coefficients given by 
		$c_k(\mc X) 
	= C \Favg(\mc U_k, \mc X) - \frac Cd + 1$, where $C = d(d+1)(d^2 -1)$. 
		\item
		The set 
	$\{  U_k\}_{k=1}^N$ 
		forms a unitary $2$-design.
	\end{enumerate}
\end{theorem}

\begin{proof}
To show that (ii) implies (i) we apply Theorem~\ref{thm:utpprojection}. From Eq.~\eqref{eq:utpproj1} we can read of that 
\begin{equation}
	\frac{1}{N}\sum_{k=1}^N c_k(\mc X) = \Tr[\mc X(\Id/d)] = 1. 
\end{equation}
Thus, the linear expansion of $\mc X$ in terms of the unitary $2$-design is affine.

It remains to establish the converse statement.  
Let $\{ \mc U_k \}_{k=1}^N$ be a set of unitary channels fulfilling 
\begin{align}\label{eq:ukcondition}
	\frac{1}{N}\sum_{k=1}^N \mc U_k \left((d^2-1) \Tr[\mc U_k\ad \mc X] + 2 - d^2 \right) = \mc X
\end{align}
for all $\mc X \in \Lin_{\utp}(H_d)$. 

A handy criterion for verifying that $\{ \mc U_k\}_{k=1}^N$ is a unitary $2$-design can be formulated in terms of its frame potential
\begin{equation}
	P = \frac1{N^2}\sum_{k,k'=1}^N|\Tr(U_k\ad U_{k'})|^4, 
\end{equation}
where again $U_k$ is the unitary matrix defining the unitary channel $\mc U_k$. 
A set of unitary gates is a unitary $2$-design if and only if $P=2$ \cite[Theorem~2]{GrossAudenaertEisert:2007}. In fact, Eq.~(\ref{eq:ukcondition}) allows to calculate the frame potential as follows.

Inserting $\mc X = 0 \oplus 1$ (the depolarising channel), we find that 
\begin{equation}\label{eq:depolarizing_average}
	\frac{1}{N} \sum_{k=1}^N \mc U_k = 0 \oplus 1. 
\end{equation}
Note that this implies that the set $\{ \mc U_k \}_{k=1}^N$ constitutes a unitary $1$-design. Therefore, Eq.~\eqref{eq:ukcondition} takes the form 
\begin{align}\label{eq:simplerukcondition}
	\frac{1}{N}\sum_{k=1}^N \mc U_k (d^2-1) \Tr[\mc U_k\ad \mc X] + 0 \oplus (2 -d^2) = \mc X
\end{align}
for all $\mc X \in \Lin_{\utp}(H_d)$. Let the left hand side of Eq.~\eqref{eq:simplerukcondition} define a linear operator $F: \mc X \mapsto F(\mc X)$. Then Eq.~\eqref{eq:simplerukcondition} implies 
\begin{align}
	\frac{1}{N}&\sum_{k' = 1}^N \Tr[\mc U_{k'}\ad F(\mc U_{k'})] \\
	&= \frac{d^2-1}{N^2} \sum_{k,k' = 1}^N |\Tr(U_{k'}\ad U_k)|^4 +  2 - d^2 \\
	&= d^2 
\end{align}
and hence
\begin{equation}
 	\frac{1}{N^2}\sum_{k,k' =1}^N |\Tr(U_{k'}\ad U_k)|^4 = 2.
\end{equation}
This completes the proof. 
\end{proof}

Note that for quantum channels, the affine expansion is \emph{almost} convex in the sense that
$c_k (\mathcal{X}) \geq  {2-d^2}/{N} \geq -{1}/{d^2}$.

\subsection{A new interpretation for the unitarity}\label{sec:unitarity}
%
In this section, we provide a proof for Theorem~\ref{thm:unitarity} and elaborate on its implications. The proof is most naturally phrased by decomposing the linear hull of unital and trace preserving maps $\Lin_{\overline{\utp}}$ into endomorphism acting on the spaces that carry irreducible representations of the unitary channels. In the proof of Theorem~\ref{thm:utpprojection} we have explicitly seen that the projection of any map $\mathcal{X}$ onto $\Lin_{\overline{\utp}}$ has the block-diagonal structure:
\begin{align*}
P_{\overline{\utp}} \left( \mathcal{X} \right) = \mathcal{X}_0 \oplus x_1, 
\end{align*}
where $x_1 = \Tr \left( \mathcal{X} \left( \Id/d \right) \right)$.
For channels that are already unital and trace preserving, this projection acts as the identity and $x_1 = 1$. Particular examples of this class are unitary channels $\mathcal{U} = \mathcal{U}_0 \oplus 1$ and the depolarizing channel $\mathcal{D} = \mathbb{O} \oplus 1$ acting as $\mathcal{D}(X) = \frac{\Tr (X)}{d} \Id$ on $X \in H_d$. 
Unitary channels are also special in the sense that they are normalised with respect to the  inner products defined in Eqs.~\eqref{eq:inner_product}, \eqref{eq:inner_product_choi} and \eqref{eq:inner_product_trace}:
\begin{align*}
d^2  = \tr \left[ \mathcal{U}^\dagger \mathcal{U} \right] 
= \left( \mathcal{L}(\mathcal{U}), \mathcal{L}(\mathcal{U}) \right)
= d^2 \left( \mathcal{U}, \mathcal{U} \right).
\end{align*}
In fact, unitary channels are the only maps with this property (provided that we also adhere to our convention of normalizing maps with respect to the trace-norm of the Choi matrix).
Combining this feature with the ``block diagonal'' structure of unitary channels yields
\begin{align*}
d^2 = \tr \left[ \mathcal{U}^\dagger \mathcal{U} \right]
= \tr \left[ \mathcal{U}_0^\dagger \oplus 1\; \mathcal{U}_0 \oplus 1 \right]
= 1 + \tr \left[ \mathcal{U}_0^\dagger \mathcal{U}_0 \right].
\end{align*}
This computation implies that a map $\mathcal{X}$ is unitary if and only if
\begin{equation*}
u (\mathcal{X}): = \frac{ \tr \left[ \mathcal{X}_0^\dagger \mathcal{X}_0 \right]}{d^2-1} \label{eq:unitarity2}
\end{equation*}
equals one.
Otherwise the \emph{unitarity} $u (\mathcal{X}) \in [0,1]$ is strictly smaller. 
For instance, $u (\mathcal{D}) = 0$ for the depolarizing channel.
This definition of the unitarity is equivalent to the one presented in Eq.~\eqref{eq:unitarity}, see \cite[Proposition~1]{WalGraHar15}. 
The argument outlined above succinctly summarises the main motivation for this figure of merit: it captures the coherence of a noise channel $\mathcal{X}$.

Equipped with this characterisation of the unitarity, we can now give the proof for the interpretation of the unitarity as the variance of the \AGF\ with respect to a unitary $2$-design.

\begin{proof}[Proof of Theorem~\ref{thm:unitarity}]
The unitarity $u (\mathcal{X})$ may be expressed as
\begin{align} 
\frac{\tr \left[ \mathcal{X}_0^\dagger \mathcal{X}_0 \right]}{d^2-1}
= \frac{\tr \left[ \left(\mathcal{X}_0 \oplus (d^2-1) x_1\right)^\dagger \mathcal{X} \right]}{d^2-1} - x_1^2.
\label{eq:unitarity_aux1}
\end{align}
Eq.~\eqref{eq:depolarizing_average} allows us to rewrite $x_1$ as 
an average over a unitary 1-design $\{\mc U_k\}_{k=1}^N$:
\begin{equation*}
x_1 = \tr \left[\left(\mathbb{O} \oplus 1\right)^\dagger \mathcal{X} \right]
= \frac{1}{N} \sum_{k=1}^N \tr \left[  \mathcal{U}_k^\dagger \mathcal{X} \right] = \mathbb{E}  \tr \left[ \mathcal{U}^\dagger \mathcal{X} \right] 
\end{equation*}
Let us now assume that the set $\left\{ \mathcal{U}_k \right\}_{k=1}^N$ is also a 2-design. Then, Eq.~\eqref{eq:integrate to zero} implies
\begin{equation*}
\frac{\left( \mathcal{X}_0 \oplus (d^2-1) x_1\right)^\dagger}{d^2-1}
= \sum_{k=1}^n \mathcal{U}_k^\dagger  \overline{\tr\left[ \mathcal{U}_k^\dagger \mathcal{X} \right]} = \mathbb{E}\;  \mathcal{U}^\dagger \tr \left[ \mathcal{X}^\dagger \mathcal{U} \right] 
\end{equation*}
Inserting both expressions into Eq.~\eqref{eq:unitarity_aux1} yields
\begin{align*}
u (\mathcal{X})
=& \tr \left[ \mathcal{X}^\dagger \mathbb{E}\;  \mathcal{U} \tr \left[ \mathcal{U}^\dagger \mathcal{X} \right] \right] - \left( \mathbb{E}  \tr \left[ \mathcal{X}^\dagger \mathcal{U} \right] \right)^2 \\
=& \mathbb{E} \; \left| \tr \left[ \mathcal{X}^\dagger \mathcal{U} \right] \right|^2  - \left( \mathbb{E} \tr \left[ \mathcal{X}^\dagger \mathcal{U} \right] \right)^2 \\
=& \mathrm{Var} \left[ \tr \left[ \mathcal{X}^\dagger \mathcal{U} \right] \right],
\end{align*}
where we have used linearity of the expectation value and the fact that the random variable $\tr \left[ \mathcal{X}^\dagger \mathcal{U} \right]$ is real-valued.
Finally, we employ the relation between $\tr \left[ \mathcal{U}^\dagger \mathcal{X} \right]$ and $\Favg(\mathcal{U},\mathcal{X})$ presented in Eq.~\eqref{eq:inner_product_trace} to conclude
\begin{align*}
u (\mathcal{X})
=& \mathrm{Var} \left[ \tr \left[ \mathcal{U}^\dagger \mathcal{X} \right] \right] \\
=& \mathrm{Var} \left[ d(d+1) \Favg(\mathcal{U},\mathcal{X}) - \tr ( \mathcal{X}(\Id)) \right] \\
=& \left( d (d+1) \right)^2 \mathrm{Var} \left[ \Favg(\mathcal{U},\mathcal{X} ) \right],
\end{align*}
because variances are invariant under constant shifts and depend quadratically on scaling factors. This establishes Theorem~\ref{thm:unitarity}.
\end{proof}

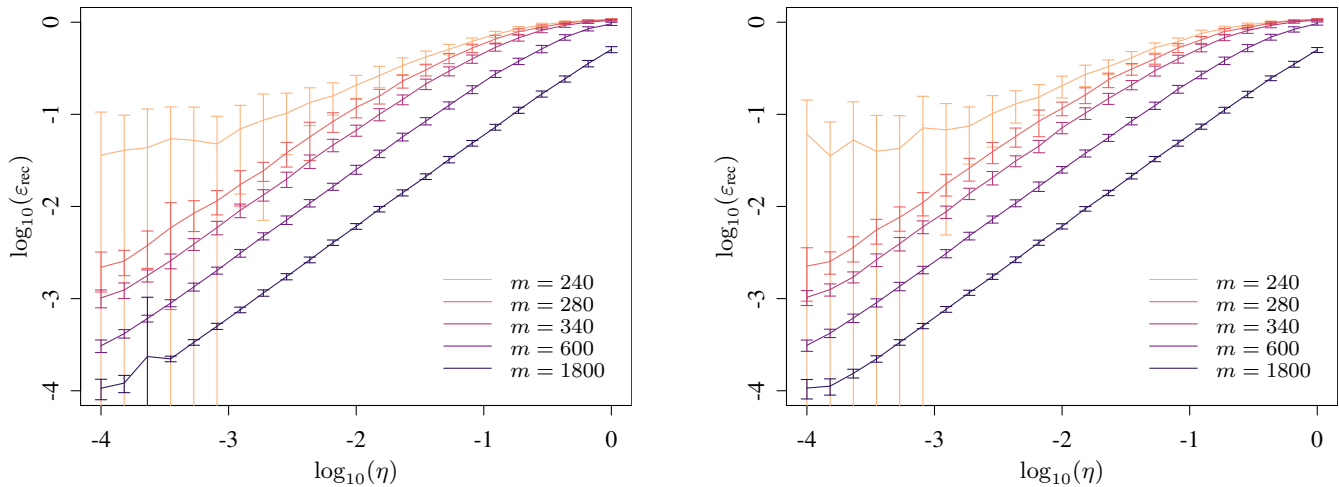
\begin{figure*}[t]
\def\mywidth{.5}
\subfloat
	{%
	\input{Deltaeta_3qubits_SDPT3.fig}
	}
\hfill
\subfloat
	{%
	\input{Deltaeta_3qubits_Haar_SDPT3.fig}
	}
	\caption{\label{fig:numerical_reconstruction_eta}
		Comparison of the reconstruction \eqref{eq:algorithm} from \AGFs\ \eqref{eq:Fmeasurements} with random Clifford unitaries (left) and Haar random unitaries (right). 
		The plots show the dependence of the observed average reconstruction error 
		$\recerror \coloneqq \norm{\mc Z^\sharp - \mc X}$, 
		on the noise strength 
		$\eta\coloneqq  \norm{\epsilon}_{\ell_2}$ for $3$ qubits and 
		different numbers of \AGFs\ $m$. 
		The error bars denote the observed standard deviation. 
		The averages are taken over $100$ samples of random {i.i.d.} measurements and channels (non-uniform). 
		The Matlab code and data used to create these plots can be found on GitHub \cite{our_git_repo}.
	}
	
\end{figure*}

We conclude this section with a more speculative note regarding the possible applications for Theorem~\ref{thm:unitarity}. 
A direct estimation procedure for the unitarity has been proposed in Ref.\ \cite{WalGraHar15}. Inspired by randomised benchmarking, this procedure is robust towards SPAM errors, but has other drawbacks: Estimating the purity of outcome states directly is challenging, because the operator square function is not linear. Although Wallman et al.\ have found ways around this issue, their approaches are not yet completely satisfactory. 

We propose an alternative approach based on Theorem~\ref{thm:unitarity}. 
It might be conceivable that techniques like importance sampling could be employed to efficiently estimate this variance -- and thus the unitarity -- from ``few'' samples. 
The fourth moment bounds computed here could potentially serve as bounds on the ``variance of this variance'' and help control the convergence.

\subsection{Numerical demonstrations}
\label{sec:numerics}
We emphasise that the main contributions of this work are of theoretical nature (we prove several Theorems).
Nonetheless, we would also like to demonstrate the practical feasibility of our reconstruction procedure \eqref{eq:algorithm} and discuss some of its subtleties.
The Matlab code for our numerical experiments can be found on GitHub \cite{our_git_repo}. 

Let $\mc X$ denote a unitary quantum channel. 
Given measurements $f_i$ from Eq.~\eqref{eq:ChoiMeasurements} with Clifford unitaries $C_i$ we approximately recover $\mc X$ using the semi-definite program (SDP) \eqref{eq:Choi-algorithm} with $q=2$. 
In the numerical experiments we draw a three-qubit unitary channel $\mc X$ uniformly at random, 
the $m$ Clifford unitaries for the measurements uniformly at random, and the noise $\epsilon \in \RR^m$ uniformly from a sphere with radius $\eta$, i.e., $\lTwoNorm{\epsilon}= \eta$.  

Then we solve the SDP using Matlab, CVX and SDPT3. The resulting average reconstruction error is plotted against the number of measurement settings $m$ and the noise strength $\eta$ in Figure~\ref{fig:numerical_reconstruction} and Figure~\ref{fig:numerical_reconstruction_eta} (left), respectively. As a comparison we run simulations for Haar random unitary measurements, see Figure~\ref{fig:numerical_reconstruction_eta} (right). 
We find that the measurements based on random Clifford unitaries perform equally well as measurements based on Haar random unitaries. 
This is in agreement with a similar observation made for the noiseless case by two of the authors in 
Ref.\ \cite{KimLiu16}.

We observed that sometimes the SDP solver cannot find a solution. 
We also tested the use of  Mosek instead of SDPT3. We find that the Mosek solver is faster but has more problems finding the correct solution. For the cases where the SDP solver does not exit with status ``solved'' we relax the machine precision on the equality constraints in the SDP \eqref{eq:Choi-algorithm} and change the measurement noise by a machine precision amount. 
More explicitly, for an integer $j\geq 0 $ we try to solve 
\begin{equation}
\begin{split}
	\operatorname*{minimise}_Z \quad &\norm{\mathcal{A}(Z)- f}_{\ell_2} 
	\\
	\operatorname{subject\ to}\quad &Z\geq 0, \\
									& \fnormB{\Tr_1(Z) - \frac\1 d} \leq 10^j\, \eps, \\
									& \fnormB{\Tr_2(Z) - \frac\1 d} \leq 10^j\, \eps
\end{split}
\end{equation}
where $\eps$ denotes the machine precision and $\Tr_1$ and $\Tr_2$ the partial traces on $\Lin(\CC^d \otimes \CC^d)$. 
We successively try to solve these SDPs for $j=0,1,2,\dots, 6$. 
Moreover, we change the measurement noise $\epsilon$ to $\epsilon' + \zeta$ in each of these trials, where each $\zeta_i = \eps\cdot g_i$ with $g_i \sim \mc N(0,1)$ is an independent normally distributed random number. 
For the Clifford type measurement (Figures~\ref{fig:numerical_reconstruction} and~\ref{fig:numerical_reconstruction_eta} left) a total of 24\,400 channels were reconstructed and $j$ was increased 1\,865 many times in total. 
For the Haar random measurement unitaries (Figure~\ref{fig:numerical_reconstruction_eta} left) a total of 12\,900 channels were reconstructed and $j$ was increased 950 times. 
So, we observed that with a probability of $\sim 7.5\%$ the SDP solver cannot solve the given SDP with machine precision constraints. 

Some of the error bars in the plots in Figures~\ref{fig:numerical_reconstruction} and~\ref{fig:numerical_reconstruction_eta} might seem quite large, which we would like to comment on. 
Note that in compressed sensing it is typical to have a phase-transition from having no recovery for too small numbers of measurements $m$ to having a recovery with very high probability once $m$ exceeds a certain threshold. 
This phase transition region becomes smeared out if the noise strength $\lTwoNorm{\epsilon}$ is increased. 
For those $m$ in the phase transition region the reconstruction errors are expected to fluctuate a lot, which we observe in the plots. 

The slope of the linear part of plots $\recerror(m)$ in Figure~\ref{fig:numerical_reconstruction} is roughly $\delta \recerror(m)/ \delta m \approx -1.3$. 
This means that the reconstruction error scales like 
$\recerror(m) \sim m^{-1.3}$, which is better than Theorem~\ref{thm:recoveryguarantee} suggests. 
The reason for this discrepancy is that the theorem also bounds systematic errors and even adversarial noise $\epsilon$ whereas in the numerics we have drawn $\epsilon_i$ uniformly from a sphere, i.e., $\epsilon_i$ are {i.i.d.} up to a rescaling.


\section*{Acknowledgements}
We thank Steven T.~Flammia, Christian Krumnow, Robin Harper, and Michae{\l} Horodecki for inspiring discussions and helpful comments.
RK is particularly grateful for a ``peaceful disagreement'' (we borrow this term from Ref.~\cite{mateus_blog}) with
Mateus Ara\'ujo that ultimately led to our current understanding of the relation between Clifford gates and tight frames.
We used code from references \cite{qetlab,randCliff,qip} to run our simulations.
The work of IR, RK, and JE was funded by 
AQuS, DFG (SPP1798 CoSIP, EI 519/9-1, EI 519/7-1, EI 519/14-1), the ERC (TAQ)
and the Templeton Foundation.
Parts of YKL's work were carried out in the context of the AFOSR MURI project ``Optimal Measurements for Scalable Quantum Technologies.''  
Contributions to this work by NIST, an agency of the US government, are not subject to US copyright. 
Any mention of commercial products does not indicate endorsement by NIST. 
DG's work has been supported by the Excellence Initiative of the German Federal and State Governments (ZUK 81), the ARO under contract W911NF-14-1-0098 (Quantum Characterization, Verification, and Validation),
Universities Australia and DAAD's Joint Research Co-operation Scheme (using funds provided by the German Federal Ministry of Education and Research),
and the DFG (SPP1798 CoSIP, project B01 of CRC 183).
The work of MK was funded by the National Science Centre, Poland within the project Polonez (2015/19/P/ST2/03001) which has received funding from the European Union’s Horizon 2020 research and innovation programme under the Marie Sk{\l}odowska-Curie grant agreement No 665778.

\bibliographystyle{apsrev4-1}
\input{Recovering_quantum_gates_precompiled.bbl}

\end{document}

%% file: Deltaeta_3qubits_SDPT3.fig.tex
\begin{tikzpicture}[x=1pt,y=1pt]
\definecolor{fillColor}{RGB}{255,255,255}
\path[use as bounding box,fill=fillColor,fill opacity=0.00] (0,0) rectangle (238.49,180.67);
\begin{scope}
\path[clip] ( 30.00, 30.00) rectangle (238.49,180.67);
\definecolor{drawColor}{RGB}{254,175,119}

\path[draw=drawColor,line width= 0.4pt,line join=round,line cap=round] (230.77,176.15) --
	(221.99,175.79) --
	(213.22,175.16) --
	(204.44,174.09) --
	(195.67,172.54) --
	(186.89,170.52) --
	(178.12,167.83) --
	(169.35,164.85) --
	(160.57,161.93) --
	(151.79,158.74) --
	(143.02,154.88) --
	(134.25,151.26) --
	(125.47,147.22) --
	(116.70,144.73) --
	(107.92,140.68) --
	( 99.14,138.04) --
	( 90.37,134.70) --
	( 81.60,128.99) --
	( 72.82,130.38) --
	( 64.05,131.01) --
	( 55.27,127.63) --
	( 46.50,126.63) --
	( 37.72,124.75);
\end{scope}
\begin{scope}
\path[clip] (  0.00,  0.00) rectangle (238.49,180.67);
\definecolor{drawColor}{RGB}{0,0,0}

\path[draw=drawColor,line width= 0.4pt,line join=round,line cap=round] ( 37.72, 30.00) -- (230.77, 30.00);

\path[draw=drawColor,line width= 0.4pt,line join=round,line cap=round] ( 37.72, 30.00) -- ( 37.72, 26.40);

\path[draw=drawColor,line width= 0.4pt,line join=round,line cap=round] ( 85.98, 30.00) -- ( 85.98, 26.40);

\path[draw=drawColor,line width= 0.4pt,line join=round,line cap=round] (134.25, 30.00) -- (134.25, 26.40);

\path[draw=drawColor,line width= 0.4pt,line join=round,line cap=round] (182.51, 30.00) -- (182.51, 26.40);

\path[draw=drawColor,line width= 0.4pt,line join=round,line cap=round] (230.77, 30.00) -- (230.77, 26.40);

\node[text=drawColor,anchor=base,inner sep=0pt, outer sep=0pt, scale=  1.00] at ( 37.72, 14.40) {-4};

\node[text=drawColor,anchor=base,inner sep=0pt, outer sep=0pt, scale=  1.00] at ( 85.98, 14.40) {-3};

\node[text=drawColor,anchor=base,inner sep=0pt, outer sep=0pt, scale=  1.00] at (134.25, 14.40) {-2};

\node[text=drawColor,anchor=base,inner sep=0pt, outer sep=0pt, scale=  1.00] at (182.51, 14.40) {-1};

\node[text=drawColor,anchor=base,inner sep=0pt, outer sep=0pt, scale=  1.00] at (230.77, 14.40) {0};

\path[draw=drawColor,line width= 0.4pt,line join=round,line cap=round] ( 30.00, 35.58) -- ( 30.00,175.09);

\path[draw=drawColor,line width= 0.4pt,line join=round,line cap=round] ( 30.00, 35.58) -- ( 26.40, 35.58);

\path[draw=drawColor,line width= 0.4pt,line join=round,line cap=round] ( 30.00, 70.46) -- ( 26.40, 70.46);

\path[draw=drawColor,line width= 0.4pt,line join=round,line cap=round] ( 30.00,105.34) -- ( 26.40,105.34);

\path[draw=drawColor,line width= 0.4pt,line join=round,line cap=round] ( 30.00,140.22) -- ( 26.40,140.22);

\path[draw=drawColor,line width= 0.4pt,line join=round,line cap=round] ( 30.00,175.09) -- ( 26.40,175.09);

\node[text=drawColor,rotate= 90.00,anchor=base,inner sep=0pt, outer sep=0pt, scale=  1.00] at ( 21.60, 35.58) {-4};

\node[text=drawColor,rotate= 90.00,anchor=base,inner sep=0pt, outer sep=0pt, scale=  1.00] at ( 21.60, 70.46) {-3};

\node[text=drawColor,rotate= 90.00,anchor=base,inner sep=0pt, outer sep=0pt, scale=  1.00] at ( 21.60,105.34) {-2};

\node[text=drawColor,rotate= 90.00,anchor=base,inner sep=0pt, outer sep=0pt, scale=  1.00] at ( 21.60,140.22) {-1};

\node[text=drawColor,rotate= 90.00,anchor=base,inner sep=0pt, outer sep=0pt, scale=  1.00] at ( 21.60,175.09) {0};

\path[draw=drawColor,line width= 0.4pt,line join=round,line cap=round] ( 30.00, 30.00) --
	(238.49, 30.00) --
	(238.49,180.67) --
	( 30.00,180.67) --
	( 30.00, 30.00);
\end{scope}
\begin{scope}
\path[clip] ( 30.00, 30.00) rectangle (238.49,180.67);
\definecolor{drawColor}{RGB}{254,175,119}

\path[draw=drawColor,line width= 0.4pt,line join=round,line cap=round] (230.77,175.88) -- (230.77,176.40);

\path[draw=drawColor,line width= 0.4pt,line join=round,line cap=round] (228.60,175.88) --
	(230.77,175.88) --
	(232.94,175.88);

\path[draw=drawColor,line width= 0.4pt,line join=round,line cap=round] (232.94,176.40) --
	(230.77,176.40) --
	(228.60,176.40);

\path[draw=drawColor,line width= 0.4pt,line join=round,line cap=round] (221.99,175.48) -- (221.99,176.08);

\path[draw=drawColor,line width= 0.4pt,line join=round,line cap=round] (219.83,175.48) --
	(221.99,175.48) --
	(224.16,175.48);

\path[draw=drawColor,line width= 0.4pt,line join=round,line cap=round] (224.16,176.08) --
	(221.99,176.08) --
	(219.83,176.08);

\path[draw=drawColor,line width= 0.4pt,line join=round,line cap=round] (213.22,174.71) -- (213.22,175.60);

\path[draw=drawColor,line width= 0.4pt,line join=round,line cap=round] (211.05,174.71) --
	(213.22,174.71) --
	(215.39,174.71);

\path[draw=drawColor,line width= 0.4pt,line join=round,line cap=round] (215.39,175.60) --
	(213.22,175.60) --
	(211.05,175.60);

\path[draw=drawColor,line width= 0.4pt,line join=round,line cap=round] (204.44,173.45) -- (204.44,174.71);

\path[draw=drawColor,line width= 0.4pt,line join=round,line cap=round] (202.28,173.45) --
	(204.44,173.45) --
	(206.61,173.45);

\path[draw=drawColor,line width= 0.4pt,line join=round,line cap=round] (206.61,174.71) --
	(204.44,174.71) --
	(202.28,174.71);

\path[draw=drawColor,line width= 0.4pt,line join=round,line cap=round] (195.67,171.59) -- (195.67,173.45);

\path[draw=drawColor,line width= 0.4pt,line join=round,line cap=round] (193.50,171.59) --
	(195.67,171.59) --
	(197.84,171.59);

\path[draw=drawColor,line width= 0.4pt,line join=round,line cap=round] (197.84,173.45) --
	(195.67,173.45) --
	(193.50,173.45);

\path[draw=drawColor,line width= 0.4pt,line join=round,line cap=round] (186.89,169.31) -- (186.89,171.64);

\path[draw=drawColor,line width= 0.4pt,line join=round,line cap=round] (184.73,169.31) --
	(186.89,169.31) --
	(189.06,169.31);

\path[draw=drawColor,line width= 0.4pt,line join=round,line cap=round] (189.06,171.64) --
	(186.89,171.64) --
	(184.73,171.64);

\path[draw=drawColor,line width= 0.4pt,line join=round,line cap=round] (178.12,166.44) -- (178.12,169.11);

\path[draw=drawColor,line width= 0.4pt,line join=round,line cap=round] (175.95,166.44) --
	(178.12,166.44) --
	(180.29,166.44);

\path[draw=drawColor,line width= 0.4pt,line join=round,line cap=round] (180.29,169.11) --
	(178.12,169.11) --
	(175.95,169.11);

\path[draw=drawColor,line width= 0.4pt,line join=round,line cap=round] (169.35,162.76) -- (169.35,166.69);

\path[draw=drawColor,line width= 0.4pt,line join=round,line cap=round] (167.18,162.76) --
	(169.35,162.76) --
	(171.51,162.76);

\path[draw=drawColor,line width= 0.4pt,line join=round,line cap=round] (171.51,166.69) --
	(169.35,166.69) --
	(167.18,166.69);

\path[draw=drawColor,line width= 0.4pt,line join=round,line cap=round] (160.57,159.31) -- (160.57,164.16);

\path[draw=drawColor,line width= 0.4pt,line join=round,line cap=round] (158.40,159.31) --
	(160.57,159.31) --
	(162.74,159.31);

\path[draw=drawColor,line width= 0.4pt,line join=round,line cap=round] (162.74,164.16) --
	(160.57,164.16) --
	(158.40,164.16);

\path[draw=drawColor,line width= 0.4pt,line join=round,line cap=round] (151.79,155.20) -- (151.79,161.61);

\path[draw=drawColor,line width= 0.4pt,line join=round,line cap=round] (149.63,155.20) --
	(151.79,155.20) --
	(153.96,155.20);

\path[draw=drawColor,line width= 0.4pt,line join=round,line cap=round] (153.96,161.61) --
	(151.79,161.61) --
	(149.63,161.61);

\path[draw=drawColor,line width= 0.4pt,line join=round,line cap=round] (143.02,150.30) -- (143.02,158.40);

\path[draw=drawColor,line width= 0.4pt,line join=round,line cap=round] (140.85,150.30) --
	(143.02,150.30) --
	(145.19,150.30);

\path[draw=drawColor,line width= 0.4pt,line join=round,line cap=round] (145.19,158.40) --
	(143.02,158.40) --
	(140.85,158.40);

\path[draw=drawColor,line width= 0.4pt,line join=round,line cap=round] (134.25,146.31) -- (134.25,154.99);

\path[draw=drawColor,line width= 0.4pt,line join=round,line cap=round] (132.08,146.31) --
	(134.25,146.31) --
	(136.41,146.31);

\path[draw=drawColor,line width= 0.4pt,line join=round,line cap=round] (136.41,154.99) --
	(134.25,154.99) --
	(132.08,154.99);

\path[draw=drawColor,line width= 0.4pt,line join=round,line cap=round] (125.47,139.91) -- (125.47,152.14);

\path[draw=drawColor,line width= 0.4pt,line join=round,line cap=round] (123.30,139.91) --
	(125.47,139.91) --
	(127.64,139.91);

\path[draw=drawColor,line width= 0.4pt,line join=round,line cap=round] (127.64,152.14) --
	(125.47,152.14) --
	(123.30,152.14);

\path[draw=drawColor,line width= 0.4pt,line join=round,line cap=round] (116.70,135.94) -- (116.70,150.26);

\path[draw=drawColor,line width= 0.4pt,line join=round,line cap=round] (114.53,135.94) --
	(116.70,135.94) --
	(118.87,135.94);

\path[draw=drawColor,line width= 0.4pt,line join=round,line cap=round] (118.87,150.26) --
	(116.70,150.26) --
	(114.53,150.26);

\path[draw=drawColor,line width= 0.4pt,line join=round,line cap=round] (107.92,124.93) -- (107.92,148.24);

\path[draw=drawColor,line width= 0.4pt,line join=round,line cap=round] (105.75,124.93) --
	(107.92,124.93) --
	(110.09,124.93);

\path[draw=drawColor,line width= 0.4pt,line join=round,line cap=round] (110.09,148.24) --
	(107.92,148.24) --
	(105.75,148.24);

\path[draw=drawColor,line width= 0.4pt,line join=round,line cap=round] ( 99.14,100.09) -- ( 99.14,147.90);

\path[draw=drawColor,line width= 0.4pt,line join=round,line cap=round] ( 96.98,100.09) --
	( 99.14,100.09) --
	(101.31,100.09);

\path[draw=drawColor,line width= 0.4pt,line join=round,line cap=round] (101.31,147.90) --
	( 99.14,147.90) --
	( 96.98,147.90);

\path[draw=drawColor,line width= 0.4pt,line join=round,line cap=round] ( 90.37,110.04) -- ( 90.37,143.63);

\path[draw=drawColor,line width= 0.4pt,line join=round,line cap=round] ( 88.20,110.04) --
	( 90.37,110.04) --
	( 92.54,110.04);

\path[draw=drawColor,line width= 0.4pt,line join=round,line cap=round] ( 92.54,143.63) --
	( 90.37,143.63) --
	( 88.20,143.63);

\path[draw=drawColor,line width= 0.4pt,line join=round,line cap=round] ( 81.60,  0.00) -- ( 81.60,139.50);

\path[draw=drawColor,line width= 0.4pt,line join=round,line cap=round] ( 83.76,139.50) --
	( 81.60,139.50) --
	( 79.43,139.50);

\path[draw=drawColor,line width= 0.4pt,line join=round,line cap=round] ( 72.82,  0.00) -- ( 72.82,142.96);

\path[draw=drawColor,line width= 0.4pt,line join=round,line cap=round] ( 74.99,142.96) --
	( 72.82,142.96) --
	( 70.65,142.96);

\path[draw=drawColor,line width= 0.4pt,line join=round,line cap=round] ( 64.05,  0.00) -- ( 64.05,143.02);

\path[draw=drawColor,line width= 0.4pt,line join=round,line cap=round] ( 66.22,143.02) --
	( 64.05,143.02) --
	( 61.88,143.02);

\path[draw=drawColor,line width= 0.4pt,line join=round,line cap=round] ( 55.27,  0.00) -- ( 55.27,142.24);

\path[draw=drawColor,line width= 0.4pt,line join=round,line cap=round] ( 57.44,142.24) --
	( 55.27,142.24) --
	( 53.10,142.24);

\path[draw=drawColor,line width= 0.4pt,line join=round,line cap=round] ( 46.50,  0.00) -- ( 46.50,139.89);

\path[draw=drawColor,line width= 0.4pt,line join=round,line cap=round] ( 48.66,139.89) --
	( 46.50,139.89) --
	( 44.33,139.89);

\path[draw=drawColor,line width= 0.4pt,line join=round,line cap=round] ( 37.72,  0.00) -- ( 37.72,141.08);

\path[draw=drawColor,line width= 0.4pt,line join=round,line cap=round] ( 39.89,141.08) --
	( 37.72,141.08) --
	( 35.55,141.08);
\definecolor{drawColor}{RGB}{241,96,93}

\path[draw=drawColor,line width= 0.4pt,line join=round,line cap=round] (230.77,176.09) --
	(221.99,175.58) --
	(213.22,174.77) --
	(204.44,173.21) --
	(195.67,171.54) --
	(186.89,168.74) --
	(178.12,165.26) --
	(169.35,161.46) --
	(160.57,157.02) --
	(151.79,152.80) --
	(143.02,147.14) --
	(134.25,142.86) --
	(125.47,137.55) --
	(116.70,131.65) --
	(107.92,125.65) --
	( 99.14,118.86) --
	( 90.37,113.63) --
	( 81.60,107.45) --
	( 72.82,102.78) --
	( 64.05, 97.29) --
	( 55.27, 90.57) --
	( 46.50, 84.69) --
	( 37.72, 82.33);

\path[draw=drawColor,line width= 0.4pt,line join=round,line cap=round] (230.77,175.80) -- (230.77,176.36);

\path[draw=drawColor,line width= 0.4pt,line join=round,line cap=round] (228.60,175.80) --
	(230.77,175.80) --
	(232.94,175.80);

\path[draw=drawColor,line width= 0.4pt,line join=round,line cap=round] (232.94,176.36) --
	(230.77,176.36) --
	(228.60,176.36);

\path[draw=drawColor,line width= 0.4pt,line join=round,line cap=round] (221.99,175.19) -- (221.99,175.97);

\path[draw=drawColor,line width= 0.4pt,line join=round,line cap=round] (219.83,175.19) --
	(221.99,175.19) --
	(224.16,175.19);

\path[draw=drawColor,line width= 0.4pt,line join=round,line cap=round] (224.16,175.97) --
	(221.99,175.97) --
	(219.83,175.97);

\path[draw=drawColor,line width= 0.4pt,line join=round,line cap=round] (213.22,174.23) -- (213.22,175.29);

\path[draw=drawColor,line width= 0.4pt,line join=round,line cap=round] (211.05,174.23) --
	(213.22,174.23) --
	(215.39,174.23);

\path[draw=drawColor,line width= 0.4pt,line join=round,line cap=round] (215.39,175.29) --
	(213.22,175.29) --
	(211.05,175.29);

\path[draw=drawColor,line width= 0.4pt,line join=round,line cap=round] (204.44,172.31) -- (204.44,174.06);

\path[draw=drawColor,line width= 0.4pt,line join=round,line cap=round] (202.28,172.31) --
	(204.44,172.31) --
	(206.61,172.31);

\path[draw=drawColor,line width= 0.4pt,line join=round,line cap=round] (206.61,174.06) --
	(204.44,174.06) --
	(202.28,174.06);

\path[draw=drawColor,line width= 0.4pt,line join=round,line cap=round] (195.67,170.39) -- (195.67,172.60);

\path[draw=drawColor,line width= 0.4pt,line join=round,line cap=round] (193.50,170.39) --
	(195.67,170.39) --
	(197.84,170.39);

\path[draw=drawColor,line width= 0.4pt,line join=round,line cap=round] (197.84,172.60) --
	(195.67,172.60) --
	(193.50,172.60);

\path[draw=drawColor,line width= 0.4pt,line join=round,line cap=round] (186.89,167.21) -- (186.89,170.13);

\path[draw=drawColor,line width= 0.4pt,line join=round,line cap=round] (184.73,167.21) --
	(186.89,167.21) --
	(189.06,167.21);

\path[draw=drawColor,line width= 0.4pt,line join=round,line cap=round] (189.06,170.13) --
	(186.89,170.13) --
	(184.73,170.13);

\path[draw=drawColor,line width= 0.4pt,line join=round,line cap=round] (178.12,163.60) -- (178.12,166.76);

\path[draw=drawColor,line width= 0.4pt,line join=round,line cap=round] (175.95,163.60) --
	(178.12,163.60) --
	(180.29,163.60);

\path[draw=drawColor,line width= 0.4pt,line join=round,line cap=round] (180.29,166.76) --
	(178.12,166.76) --
	(175.95,166.76);

\path[draw=drawColor,line width= 0.4pt,line join=round,line cap=round] (169.35,159.40) -- (169.35,163.28);

\path[draw=drawColor,line width= 0.4pt,line join=round,line cap=round] (167.18,159.40) --
	(169.35,159.40) --
	(171.51,159.40);

\path[draw=drawColor,line width= 0.4pt,line join=round,line cap=round] (171.51,163.28) --
	(169.35,163.28) --
	(167.18,163.28);

\path[draw=drawColor,line width= 0.4pt,line join=round,line cap=round] (160.57,154.75) -- (160.57,159.00);

\path[draw=drawColor,line width= 0.4pt,line join=round,line cap=round] (158.40,154.75) --
	(160.57,154.75) --
	(162.74,154.75);

\path[draw=drawColor,line width= 0.4pt,line join=round,line cap=round] (162.74,159.00) --
	(160.57,159.00) --
	(158.40,159.00);

\path[draw=drawColor,line width= 0.4pt,line join=round,line cap=round] (151.79,150.16) -- (151.79,155.05);

\path[draw=drawColor,line width= 0.4pt,line join=round,line cap=round] (149.63,150.16) --
	(151.79,150.16) --
	(153.96,150.16);

\path[draw=drawColor,line width= 0.4pt,line join=round,line cap=round] (153.96,155.05) --
	(151.79,155.05) --
	(149.63,155.05);

\path[draw=drawColor,line width= 0.4pt,line join=round,line cap=round] (143.02,144.12) -- (143.02,149.67);

\path[draw=drawColor,line width= 0.4pt,line join=round,line cap=round] (140.85,144.12) --
	(143.02,144.12) --
	(145.19,144.12);

\path[draw=drawColor,line width= 0.4pt,line join=round,line cap=round] (145.19,149.67) --
	(143.02,149.67) --
	(140.85,149.67);

\path[draw=drawColor,line width= 0.4pt,line join=round,line cap=round] (134.25,138.90) -- (134.25,145.99);

\path[draw=drawColor,line width= 0.4pt,line join=round,line cap=round] (132.08,138.90) --
	(134.25,138.90) --
	(136.41,138.90);

\path[draw=drawColor,line width= 0.4pt,line join=round,line cap=round] (136.41,145.99) --
	(134.25,145.99) --
	(132.08,145.99);

\path[draw=drawColor,line width= 0.4pt,line join=round,line cap=round] (125.47,133.39) -- (125.47,140.82);

\path[draw=drawColor,line width= 0.4pt,line join=round,line cap=round] (123.30,133.39) --
	(125.47,133.39) --
	(127.64,133.39);

\path[draw=drawColor,line width= 0.4pt,line join=round,line cap=round] (127.64,140.82) --
	(125.47,140.82) --
	(123.30,140.82);

\path[draw=drawColor,line width= 0.4pt,line join=round,line cap=round] (116.70,122.70) -- (116.70,137.23);

\path[draw=drawColor,line width= 0.4pt,line join=round,line cap=round] (114.53,122.70) --
	(116.70,122.70) --
	(118.87,122.70);

\path[draw=drawColor,line width= 0.4pt,line join=round,line cap=round] (118.87,137.23) --
	(116.70,137.23) --
	(114.53,137.23);

\path[draw=drawColor,line width= 0.4pt,line join=round,line cap=round] (107.92,120.26) -- (107.92,129.61);

\path[draw=drawColor,line width= 0.4pt,line join=round,line cap=round] (105.75,120.26) --
	(107.92,120.26) --
	(110.09,120.26);

\path[draw=drawColor,line width= 0.4pt,line join=round,line cap=round] (110.09,129.61) --
	(107.92,129.61) --
	(105.75,129.61);

\path[draw=drawColor,line width= 0.4pt,line join=round,line cap=round] ( 99.14,114.84) -- ( 99.14,122.04);

\path[draw=drawColor,line width= 0.4pt,line join=round,line cap=round] ( 96.98,114.84) --
	( 99.14,114.84) --
	(101.31,114.84);

\path[draw=drawColor,line width= 0.4pt,line join=round,line cap=round] (101.31,122.04) --
	( 99.14,122.04) --
	( 96.98,122.04);

\path[draw=drawColor,line width= 0.4pt,line join=round,line cap=round] ( 90.37,105.50) -- ( 90.37,118.89);

\path[draw=drawColor,line width= 0.4pt,line join=round,line cap=round] ( 88.20,105.50) --
	( 90.37,105.50) --
	( 92.54,105.50);

\path[draw=drawColor,line width= 0.4pt,line join=round,line cap=round] ( 92.54,118.89) --
	( 90.37,118.89) --
	( 88.20,118.89);

\path[draw=drawColor,line width= 0.4pt,line join=round,line cap=round] ( 81.60,102.15) -- ( 81.60,111.36);

\path[draw=drawColor,line width= 0.4pt,line join=round,line cap=round] ( 79.43,102.15) --
	( 81.60,102.15) --
	( 83.76,102.15);

\path[draw=drawColor,line width= 0.4pt,line join=round,line cap=round] ( 83.76,111.36) --
	( 81.60,111.36) --
	( 79.43,111.36);

\path[draw=drawColor,line width= 0.4pt,line join=round,line cap=round] ( 72.82, 95.85) -- ( 72.82,107.52);

\path[draw=drawColor,line width= 0.4pt,line join=round,line cap=round] ( 70.65, 95.85) --
	( 72.82, 95.85) --
	( 74.99, 95.85);

\path[draw=drawColor,line width= 0.4pt,line join=round,line cap=round] ( 74.99,107.52) --
	( 72.82,107.52) --
	( 70.65,107.52);

\path[draw=drawColor,line width= 0.4pt,line join=round,line cap=round] ( 64.05, 66.34) -- ( 64.05,106.77);

\path[draw=drawColor,line width= 0.4pt,line join=round,line cap=round] ( 61.88, 66.34) --
	( 64.05, 66.34) --
	( 66.22, 66.34);

\path[draw=drawColor,line width= 0.4pt,line join=round,line cap=round] ( 66.22,106.77) --
	( 64.05,106.77) --
	( 61.88,106.77);

\path[draw=drawColor,line width= 0.4pt,line join=round,line cap=round] ( 55.27, 81.91) -- ( 55.27, 96.04);

\path[draw=drawColor,line width= 0.4pt,line join=round,line cap=round] ( 53.10, 81.91) --
	( 55.27, 81.91) --
	( 57.44, 81.91);

\path[draw=drawColor,line width= 0.4pt,line join=round,line cap=round] ( 57.44, 96.04) --
	( 55.27, 96.04) --
	( 53.10, 96.04);

\path[draw=drawColor,line width= 0.4pt,line join=round,line cap=round] ( 46.50, 79.15) -- ( 46.50, 88.73);

\path[draw=drawColor,line width= 0.4pt,line join=round,line cap=round] ( 44.33, 79.15) --
	( 46.50, 79.15) --
	( 48.66, 79.15);

\path[draw=drawColor,line width= 0.4pt,line join=round,line cap=round] ( 48.66, 88.73) --
	( 46.50, 88.73) --
	( 44.33, 88.73);

\path[draw=drawColor,line width= 0.4pt,line join=round,line cap=round] ( 37.72, 72.94) -- ( 37.72, 88.09);

\path[draw=drawColor,line width= 0.4pt,line join=round,line cap=round] ( 35.55, 72.94) --
	( 37.72, 72.94) --
	( 39.89, 72.94);

\path[draw=drawColor,line width= 0.4pt,line join=round,line cap=round] ( 39.89, 88.09) --
	( 37.72, 88.09) --
	( 35.55, 88.09);
\definecolor{drawColor}{RGB}{182,54,121}

\path[draw=drawColor,line width= 0.4pt,line join=round,line cap=round] (230.77,175.84) --
	(221.99,175.11) --
	(213.22,173.93) --
	(204.44,172.05) --
	(195.67,169.15) --
	(186.89,165.58) --
	(178.12,161.20) --
	(169.35,156.58) --
	(160.57,151.70) --
	(151.79,145.81) --
	(143.02,140.23) --
	(134.25,134.17) --
	(125.47,128.60) --
	(116.70,122.49) --
	(107.92,115.70) --
	( 99.14,109.68) --
	( 90.37,103.89) --
	( 81.60, 97.37) --
	( 72.82, 91.10) --
	( 64.05, 84.81) --
	( 55.27, 79.20) --
	( 46.50, 73.73) --
	( 37.72, 70.72);

\path[draw=drawColor,line width= 0.4pt,line join=round,line cap=round] (230.77,175.51) -- (230.77,176.17);

\path[draw=drawColor,line width= 0.4pt,line join=round,line cap=round] (228.60,175.51) --
	(230.77,175.51) --
	(232.94,175.51);

\path[draw=drawColor,line width= 0.4pt,line join=round,line cap=round] (232.94,176.17) --
	(230.77,176.17) --
	(228.60,176.17);

\path[draw=drawColor,line width= 0.4pt,line join=round,line cap=round] (221.99,174.62) -- (221.99,175.58);

\path[draw=drawColor,line width= 0.4pt,line join=round,line cap=round] (219.83,174.62) --
	(221.99,174.62) --
	(224.16,174.62);

\path[draw=drawColor,line width= 0.4pt,line join=round,line cap=round] (224.16,175.58) --
	(221.99,175.58) --
	(219.83,175.58);

\path[draw=drawColor,line width= 0.4pt,line join=round,line cap=round] (213.22,173.12) -- (213.22,174.69);

\path[draw=drawColor,line width= 0.4pt,line join=round,line cap=round] (211.05,173.12) --
	(213.22,173.12) --
	(215.39,173.12);

\path[draw=drawColor,line width= 0.4pt,line join=round,line cap=round] (215.39,174.69) --
	(213.22,174.69) --
	(211.05,174.69);

\path[draw=drawColor,line width= 0.4pt,line join=round,line cap=round] (204.44,171.12) -- (204.44,172.94);

\path[draw=drawColor,line width= 0.4pt,line join=round,line cap=round] (202.28,171.12) --
	(204.44,171.12) --
	(206.61,171.12);

\path[draw=drawColor,line width= 0.4pt,line join=round,line cap=round] (206.61,172.94) --
	(204.44,172.94) --
	(202.28,172.94);

\path[draw=drawColor,line width= 0.4pt,line join=round,line cap=round] (195.67,167.94) -- (195.67,170.27);

\path[draw=drawColor,line width= 0.4pt,line join=round,line cap=round] (193.50,167.94) --
	(195.67,167.94) --
	(197.84,167.94);

\path[draw=drawColor,line width= 0.4pt,line join=round,line cap=round] (197.84,170.27) --
	(195.67,170.27) --
	(193.50,170.27);

\path[draw=drawColor,line width= 0.4pt,line join=round,line cap=round] (186.89,164.13) -- (186.89,166.89);

\path[draw=drawColor,line width= 0.4pt,line join=round,line cap=round] (184.73,164.13) --
	(186.89,164.13) --
	(189.06,164.13);

\path[draw=drawColor,line width= 0.4pt,line join=round,line cap=round] (189.06,166.89) --
	(186.89,166.89) --
	(184.73,166.89);

\path[draw=drawColor,line width= 0.4pt,line join=round,line cap=round] (178.12,159.54) -- (178.12,162.69);

\path[draw=drawColor,line width= 0.4pt,line join=round,line cap=round] (175.95,159.54) --
	(178.12,159.54) --
	(180.29,159.54);

\path[draw=drawColor,line width= 0.4pt,line join=round,line cap=round] (180.29,162.69) --
	(178.12,162.69) --
	(175.95,162.69);

\path[draw=drawColor,line width= 0.4pt,line join=round,line cap=round] (169.35,154.78) -- (169.35,158.19);

\path[draw=drawColor,line width= 0.4pt,line join=round,line cap=round] (167.18,154.78) --
	(169.35,154.78) --
	(171.51,154.78);

\path[draw=drawColor,line width= 0.4pt,line join=round,line cap=round] (171.51,158.19) --
	(169.35,158.19) --
	(167.18,158.19);

\path[draw=drawColor,line width= 0.4pt,line join=round,line cap=round] (160.57,149.59) -- (160.57,153.54);

\path[draw=drawColor,line width= 0.4pt,line join=round,line cap=round] (158.40,149.59) --
	(160.57,149.59) --
	(162.74,149.59);

\path[draw=drawColor,line width= 0.4pt,line join=round,line cap=round] (162.74,153.54) --
	(160.57,153.54) --
	(158.40,153.54);

\path[draw=drawColor,line width= 0.4pt,line join=round,line cap=round] (151.79,143.84) -- (151.79,147.55);

\path[draw=drawColor,line width= 0.4pt,line join=round,line cap=round] (149.63,143.84) --
	(151.79,143.84) --
	(153.96,143.84);

\path[draw=drawColor,line width= 0.4pt,line join=round,line cap=round] (153.96,147.55) --
	(151.79,147.55) --
	(149.63,147.55);

\path[draw=drawColor,line width= 0.4pt,line join=round,line cap=round] (143.02,137.77) -- (143.02,142.35);

\path[draw=drawColor,line width= 0.4pt,line join=round,line cap=round] (140.85,137.77) --
	(143.02,137.77) --
	(145.19,137.77);

\path[draw=drawColor,line width= 0.4pt,line join=round,line cap=round] (145.19,142.35) --
	(143.02,142.35) --
	(140.85,142.35);

\path[draw=drawColor,line width= 0.4pt,line join=round,line cap=round] (134.25,131.93) -- (134.25,136.12);

\path[draw=drawColor,line width= 0.4pt,line join=round,line cap=round] (132.08,131.93) --
	(134.25,131.93) --
	(136.41,131.93);

\path[draw=drawColor,line width= 0.4pt,line join=round,line cap=round] (136.41,136.12) --
	(134.25,136.12) --
	(132.08,136.12);

\path[draw=drawColor,line width= 0.4pt,line join=round,line cap=round] (125.47,126.06) -- (125.47,130.77);

\path[draw=drawColor,line width= 0.4pt,line join=round,line cap=round] (123.30,126.06) --
	(125.47,126.06) --
	(127.64,126.06);

\path[draw=drawColor,line width= 0.4pt,line join=round,line cap=round] (127.64,130.77) --
	(125.47,130.77) --
	(123.30,130.77);

\path[draw=drawColor,line width= 0.4pt,line join=round,line cap=round] (116.70,119.73) -- (116.70,124.83);

\path[draw=drawColor,line width= 0.4pt,line join=round,line cap=round] (114.53,119.73) --
	(116.70,119.73) --
	(118.87,119.73);

\path[draw=drawColor,line width= 0.4pt,line join=round,line cap=round] (118.87,124.83) --
	(116.70,124.83) --
	(114.53,124.83);

\path[draw=drawColor,line width= 0.4pt,line join=round,line cap=round] (107.92,112.54) -- (107.92,118.32);

\path[draw=drawColor,line width= 0.4pt,line join=round,line cap=round] (105.75,112.54) --
	(107.92,112.54) --
	(110.09,112.54);

\path[draw=drawColor,line width= 0.4pt,line join=round,line cap=round] (110.09,118.32) --
	(107.92,118.32) --
	(105.75,118.32);

\path[draw=drawColor,line width= 0.4pt,line join=round,line cap=round] ( 99.14,107.37) -- ( 99.14,111.69);

\path[draw=drawColor,line width= 0.4pt,line join=round,line cap=round] ( 96.98,107.37) --
	( 99.14,107.37) --
	(101.31,107.37);

\path[draw=drawColor,line width= 0.4pt,line join=round,line cap=round] (101.31,111.69) --
	( 99.14,111.69) --
	( 96.98,111.69);

\path[draw=drawColor,line width= 0.4pt,line join=round,line cap=round] ( 90.37,101.09) -- ( 90.37,106.24);

\path[draw=drawColor,line width= 0.4pt,line join=round,line cap=round] ( 88.20,101.09) --
	( 90.37,101.09) --
	( 92.54,101.09);

\path[draw=drawColor,line width= 0.4pt,line join=round,line cap=round] ( 92.54,106.24) --
	( 90.37,106.24) --
	( 88.20,106.24);

\path[draw=drawColor,line width= 0.4pt,line join=round,line cap=round] ( 81.60, 94.61) -- ( 81.60, 99.71);

\path[draw=drawColor,line width= 0.4pt,line join=round,line cap=round] ( 79.43, 94.61) --
	( 81.60, 94.61) --
	( 83.76, 94.61);

\path[draw=drawColor,line width= 0.4pt,line join=round,line cap=round] ( 83.76, 99.71) --
	( 81.60, 99.71) --
	( 79.43, 99.71);

\path[draw=drawColor,line width= 0.4pt,line join=round,line cap=round] ( 72.82, 88.70) -- ( 72.82, 93.18);

\path[draw=drawColor,line width= 0.4pt,line join=round,line cap=round] ( 70.65, 88.70) --
	( 72.82, 88.70) --
	( 74.99, 88.70);

\path[draw=drawColor,line width= 0.4pt,line join=round,line cap=round] ( 74.99, 93.18) --
	( 72.82, 93.18) --
	( 70.65, 93.18);

\path[draw=drawColor,line width= 0.4pt,line join=round,line cap=round] ( 64.05, 81.75) -- ( 64.05, 87.35);

\path[draw=drawColor,line width= 0.4pt,line join=round,line cap=round] ( 61.88, 81.75) --
	( 64.05, 81.75) --
	( 66.22, 81.75);

\path[draw=drawColor,line width= 0.4pt,line join=round,line cap=round] ( 66.22, 87.35) --
	( 64.05, 87.35) --
	( 61.88, 87.35);

\path[draw=drawColor,line width= 0.4pt,line join=round,line cap=round] ( 55.27, 76.75) -- ( 55.27, 81.31);

\path[draw=drawColor,line width= 0.4pt,line join=round,line cap=round] ( 53.10, 76.75) --
	( 55.27, 76.75) --
	( 57.44, 76.75);

\path[draw=drawColor,line width= 0.4pt,line join=round,line cap=round] ( 57.44, 81.31) --
	( 55.27, 81.31) --
	( 53.10, 81.31);

\path[draw=drawColor,line width= 0.4pt,line join=round,line cap=round] ( 46.50, 70.50) -- ( 46.50, 76.39);

\path[draw=drawColor,line width= 0.4pt,line join=round,line cap=round] ( 44.33, 70.50) --
	( 46.50, 70.50) --
	( 48.66, 70.50);

\path[draw=drawColor,line width= 0.4pt,line join=round,line cap=round] ( 48.66, 76.39) --
	( 46.50, 76.39) --
	( 44.33, 76.39);

\path[draw=drawColor,line width= 0.4pt,line join=round,line cap=round] ( 37.72, 66.98) -- ( 37.72, 73.73);

\path[draw=drawColor,line width= 0.4pt,line join=round,line cap=round] ( 35.55, 66.98) --
	( 37.72, 66.98) --
	( 39.89, 66.98);

\path[draw=drawColor,line width= 0.4pt,line join=round,line cap=round] ( 39.89, 73.73) --
	( 37.72, 73.73) --
	( 35.55, 73.73);
\definecolor{drawColor}{RGB}{114,31,129}

\path[draw=drawColor,line width= 0.4pt,line join=round,line cap=round] (230.77,174.48) --
	(221.99,172.60) --
	(213.22,169.36) --
	(204.44,164.98) --
	(195.67,160.31) --
	(186.89,155.50) --
	(178.12,149.62) --
	(169.35,143.64) --
	(160.57,137.66) --
	(151.79,131.68) --
	(143.02,125.31) --
	(134.25,119.33) --
	(125.47,112.84) --
	(116.70,106.61) --
	(107.92,100.22) --
	( 99.14, 94.09) --
	( 90.37, 87.67) --
	( 81.60, 81.10) --
	( 72.82, 74.89) --
	( 64.05, 68.77) --
	( 55.27, 62.94) --
	( 46.50, 57.21) --
	( 37.72, 52.62);

\path[draw=drawColor,line width= 0.4pt,line join=round,line cap=round] (230.77,173.86) -- (230.77,175.08);

\path[draw=drawColor,line width= 0.4pt,line join=round,line cap=round] (228.60,173.86) --
	(230.77,173.86) --
	(232.94,173.86);

\path[draw=drawColor,line width= 0.4pt,line join=round,line cap=round] (232.94,175.08) --
	(230.77,175.08) --
	(228.60,175.08);

\path[draw=drawColor,line width= 0.4pt,line join=round,line cap=round] (221.99,171.80) -- (221.99,173.36);

\path[draw=drawColor,line width= 0.4pt,line join=round,line cap=round] (219.83,171.80) --
	(221.99,171.80) --
	(224.16,171.80);

\path[draw=drawColor,line width= 0.4pt,line join=round,line cap=round] (224.16,173.36) --
	(221.99,173.36) --
	(219.83,173.36);

\path[draw=drawColor,line width= 0.4pt,line join=round,line cap=round] (213.22,168.28) -- (213.22,170.37);

\path[draw=drawColor,line width= 0.4pt,line join=round,line cap=round] (211.05,168.28) --
	(213.22,168.28) --
	(215.39,168.28);

\path[draw=drawColor,line width= 0.4pt,line join=round,line cap=round] (215.39,170.37) --
	(213.22,170.37) --
	(211.05,170.37);

\path[draw=drawColor,line width= 0.4pt,line join=round,line cap=round] (204.44,163.53) -- (204.44,166.30);

\path[draw=drawColor,line width= 0.4pt,line join=round,line cap=round] (202.28,163.53) --
	(204.44,163.53) --
	(206.61,163.53);

\path[draw=drawColor,line width= 0.4pt,line join=round,line cap=round] (206.61,166.30) --
	(204.44,166.30) --
	(202.28,166.30);

\path[draw=drawColor,line width= 0.4pt,line join=round,line cap=round] (195.67,159.07) -- (195.67,161.46);

\path[draw=drawColor,line width= 0.4pt,line join=round,line cap=round] (193.50,159.07) --
	(195.67,159.07) --
	(197.84,159.07);

\path[draw=drawColor,line width= 0.4pt,line join=round,line cap=round] (197.84,161.46) --
	(195.67,161.46) --
	(193.50,161.46);

\path[draw=drawColor,line width= 0.4pt,line join=round,line cap=round] (186.89,154.18) -- (186.89,156.72);

\path[draw=drawColor,line width= 0.4pt,line join=round,line cap=round] (184.73,154.18) --
	(186.89,154.18) --
	(189.06,154.18);

\path[draw=drawColor,line width= 0.4pt,line join=round,line cap=round] (189.06,156.72) --
	(186.89,156.72) --
	(184.73,156.72);

\path[draw=drawColor,line width= 0.4pt,line join=round,line cap=round] (178.12,148.05) -- (178.12,151.04);

\path[draw=drawColor,line width= 0.4pt,line join=round,line cap=round] (175.95,148.05) --
	(178.12,148.05) --
	(180.29,148.05);

\path[draw=drawColor,line width= 0.4pt,line join=round,line cap=round] (180.29,151.04) --
	(178.12,151.04) --
	(175.95,151.04);

\path[draw=drawColor,line width= 0.4pt,line join=round,line cap=round] (169.35,142.22) -- (169.35,144.93);

\path[draw=drawColor,line width= 0.4pt,line join=round,line cap=round] (167.18,142.22) --
	(169.35,142.22) --
	(171.51,142.22);

\path[draw=drawColor,line width= 0.4pt,line join=round,line cap=round] (171.51,144.93) --
	(169.35,144.93) --
	(167.18,144.93);

\path[draw=drawColor,line width= 0.4pt,line join=round,line cap=round] (160.57,136.22) -- (160.57,138.98);

\path[draw=drawColor,line width= 0.4pt,line join=round,line cap=round] (158.40,136.22) --
	(160.57,136.22) --
	(162.74,136.22);

\path[draw=drawColor,line width= 0.4pt,line join=round,line cap=round] (162.74,138.98) --
	(160.57,138.98) --
	(158.40,138.98);

\path[draw=drawColor,line width= 0.4pt,line join=round,line cap=round] (151.79,130.14) -- (151.79,133.08);

\path[draw=drawColor,line width= 0.4pt,line join=round,line cap=round] (149.63,130.14) --
	(151.79,130.14) --
	(153.96,130.14);

\path[draw=drawColor,line width= 0.4pt,line join=round,line cap=round] (153.96,133.08) --
	(151.79,133.08) --
	(149.63,133.08);

\path[draw=drawColor,line width= 0.4pt,line join=round,line cap=round] (143.02,123.84) -- (143.02,126.65);

\path[draw=drawColor,line width= 0.4pt,line join=round,line cap=round] (140.85,123.84) --
	(143.02,123.84) --
	(145.19,123.84);

\path[draw=drawColor,line width= 0.4pt,line join=round,line cap=round] (145.19,126.65) --
	(143.02,126.65) --
	(140.85,126.65);

\path[draw=drawColor,line width= 0.4pt,line join=round,line cap=round] (134.25,117.55) -- (134.25,120.92);

\path[draw=drawColor,line width= 0.4pt,line join=round,line cap=round] (132.08,117.55) --
	(134.25,117.55) --
	(136.41,117.55);

\path[draw=drawColor,line width= 0.4pt,line join=round,line cap=round] (136.41,120.92) --
	(134.25,120.92) --
	(132.08,120.92);

\path[draw=drawColor,line width= 0.4pt,line join=round,line cap=round] (125.47,111.45) -- (125.47,114.12);

\path[draw=drawColor,line width= 0.4pt,line join=round,line cap=round] (123.30,111.45) --
	(125.47,111.45) --
	(127.64,111.45);

\path[draw=drawColor,line width= 0.4pt,line join=round,line cap=round] (127.64,114.12) --
	(125.47,114.12) --
	(123.30,114.12);

\path[draw=drawColor,line width= 0.4pt,line join=round,line cap=round] (116.70,105.15) -- (116.70,107.94);

\path[draw=drawColor,line width= 0.4pt,line join=round,line cap=round] (114.53,105.15) --
	(116.70,105.15) --
	(118.87,105.15);

\path[draw=drawColor,line width= 0.4pt,line join=round,line cap=round] (118.87,107.94) --
	(116.70,107.94) --
	(114.53,107.94);

\path[draw=drawColor,line width= 0.4pt,line join=round,line cap=round] (107.92, 98.53) -- (107.92,101.75);

\path[draw=drawColor,line width= 0.4pt,line join=round,line cap=round] (105.75, 98.53) --
	(107.92, 98.53) --
	(110.09, 98.53);

\path[draw=drawColor,line width= 0.4pt,line join=round,line cap=round] (110.09,101.75) --
	(107.92,101.75) --
	(105.75,101.75);

\path[draw=drawColor,line width= 0.4pt,line join=round,line cap=round] ( 99.14, 92.69) -- ( 99.14, 95.38);

\path[draw=drawColor,line width= 0.4pt,line join=round,line cap=round] ( 96.98, 92.69) --
	( 99.14, 92.69) --
	(101.31, 92.69);

\path[draw=drawColor,line width= 0.4pt,line join=round,line cap=round] (101.31, 95.38) --
	( 99.14, 95.38) --
	( 96.98, 95.38);

\path[draw=drawColor,line width= 0.4pt,line join=round,line cap=round] ( 90.37, 86.14) -- ( 90.37, 89.06);

\path[draw=drawColor,line width= 0.4pt,line join=round,line cap=round] ( 88.20, 86.14) --
	( 90.37, 86.14) --
	( 92.54, 86.14);

\path[draw=drawColor,line width= 0.4pt,line join=round,line cap=round] ( 92.54, 89.06) --
	( 90.37, 89.06) --
	( 88.20, 89.06);

\path[draw=drawColor,line width= 0.4pt,line join=round,line cap=round] ( 81.60, 79.72) -- ( 81.60, 82.37);

\path[draw=drawColor,line width= 0.4pt,line join=round,line cap=round] ( 79.43, 79.72) --
	( 81.60, 79.72) --
	( 83.76, 79.72);

\path[draw=drawColor,line width= 0.4pt,line join=round,line cap=round] ( 83.76, 82.37) --
	( 81.60, 82.37) --
	( 79.43, 82.37);

\path[draw=drawColor,line width= 0.4pt,line join=round,line cap=round] ( 72.82, 73.32) -- ( 72.82, 76.31);

\path[draw=drawColor,line width= 0.4pt,line join=round,line cap=round] ( 70.65, 73.32) --
	( 72.82, 73.32) --
	( 74.99, 73.32);

\path[draw=drawColor,line width= 0.4pt,line join=round,line cap=round] ( 74.99, 76.31) --
	( 72.82, 76.31) --
	( 70.65, 76.31);

\path[draw=drawColor,line width= 0.4pt,line join=round,line cap=round] ( 64.05, 67.38) -- ( 64.05, 70.05);

\path[draw=drawColor,line width= 0.4pt,line join=round,line cap=round] ( 61.88, 67.38) --
	( 64.05, 67.38) --
	( 66.22, 67.38);

\path[draw=drawColor,line width= 0.4pt,line join=round,line cap=round] ( 66.22, 70.05) --
	( 64.05, 70.05) --
	( 61.88, 70.05);

\path[draw=drawColor,line width= 0.4pt,line join=round,line cap=round] ( 55.27, 61.60) -- ( 55.27, 64.17);

\path[draw=drawColor,line width= 0.4pt,line join=round,line cap=round] ( 53.10, 61.60) --
	( 55.27, 61.60) --
	( 57.44, 61.60);

\path[draw=drawColor,line width= 0.4pt,line join=round,line cap=round] ( 57.44, 64.17) --
	( 55.27, 64.17) --
	( 53.10, 64.17);

\path[draw=drawColor,line width= 0.4pt,line join=round,line cap=round] ( 46.50, 55.53) -- ( 46.50, 58.73);

\path[draw=drawColor,line width= 0.4pt,line join=round,line cap=round] ( 44.33, 55.53) --
	( 46.50, 55.53) --
	( 48.66, 55.53);

\path[draw=drawColor,line width= 0.4pt,line join=round,line cap=round] ( 48.66, 58.73) --
	( 46.50, 58.73) --
	( 44.33, 58.73);

\path[draw=drawColor,line width= 0.4pt,line join=round,line cap=round] ( 37.72, 50.05) -- ( 37.72, 54.81);

\path[draw=drawColor,line width= 0.4pt,line join=round,line cap=round] ( 35.55, 50.05) --
	( 37.72, 50.05) --
	( 39.89, 50.05);

\path[draw=drawColor,line width= 0.4pt,line join=round,line cap=round] ( 39.89, 54.81) --
	( 37.72, 54.81) --
	( 35.55, 54.81);
\definecolor{drawColor}{RGB}{45,17,96}

\path[draw=drawColor,line width= 0.4pt,line join=round,line cap=round] (230.77,164.75) --
	(221.99,159.46) --
	(213.22,153.64) --
	(204.44,147.93) --
	(195.67,141.80) --
	(186.89,135.37) --
	(178.12,129.30) --
	(169.35,123.10) --
	(160.57,116.65) --
	(151.79,110.53) --
	(143.02,104.35) --
	(134.25, 97.81) --
	(125.47, 91.66) --
	(116.70, 85.15) --
	(107.92, 78.78) --
	( 99.14, 72.61) --
	( 90.37, 66.22) --
	( 81.60, 59.99) --
	( 72.82, 53.90) --
	( 64.05, 47.65) --
	( 55.27, 48.53) --
	( 46.50, 38.48) --
	( 37.72, 36.54);

\path[draw=drawColor,line width= 0.4pt,line join=round,line cap=round] (230.77,163.56) -- (230.77,165.84);

\path[draw=drawColor,line width= 0.4pt,line join=round,line cap=round] (228.60,163.56) --
	(230.77,163.56) --
	(232.94,163.56);

\path[draw=drawColor,line width= 0.4pt,line join=round,line cap=round] (232.94,165.84) --
	(230.77,165.84) --
	(228.60,165.84);

\path[draw=drawColor,line width= 0.4pt,line join=round,line cap=round] (221.99,158.20) -- (221.99,160.62);

\path[draw=drawColor,line width= 0.4pt,line join=round,line cap=round] (219.83,158.20) --
	(221.99,158.20) --
	(224.16,158.20);

\path[draw=drawColor,line width= 0.4pt,line join=round,line cap=round] (224.16,160.62) --
	(221.99,160.62) --
	(219.83,160.62);

\path[draw=drawColor,line width= 0.4pt,line join=round,line cap=round] (213.22,152.43) -- (213.22,154.75);

\path[draw=drawColor,line width= 0.4pt,line join=round,line cap=round] (211.05,152.43) --
	(213.22,152.43) --
	(215.39,152.43);

\path[draw=drawColor,line width= 0.4pt,line join=round,line cap=round] (215.39,154.75) --
	(213.22,154.75) --
	(211.05,154.75);

\path[draw=drawColor,line width= 0.4pt,line join=round,line cap=round] (204.44,146.74) -- (204.44,149.04);

\path[draw=drawColor,line width= 0.4pt,line join=round,line cap=round] (202.28,146.74) --
	(204.44,146.74) --
	(206.61,146.74);

\path[draw=drawColor,line width= 0.4pt,line join=round,line cap=round] (206.61,149.04) --
	(204.44,149.04) --
	(202.28,149.04);

\path[draw=drawColor,line width= 0.4pt,line join=round,line cap=round] (195.67,140.60) -- (195.67,142.91);

\path[draw=drawColor,line width= 0.4pt,line join=round,line cap=round] (193.50,140.60) --
	(195.67,140.60) --
	(197.84,140.60);

\path[draw=drawColor,line width= 0.4pt,line join=round,line cap=round] (197.84,142.91) --
	(195.67,142.91) --
	(193.50,142.91);

\path[draw=drawColor,line width= 0.4pt,line join=round,line cap=round] (186.89,134.22) -- (186.89,136.44);

\path[draw=drawColor,line width= 0.4pt,line join=round,line cap=round] (184.73,134.22) --
	(186.89,134.22) --
	(189.06,134.22);

\path[draw=drawColor,line width= 0.4pt,line join=round,line cap=round] (189.06,136.44) --
	(186.89,136.44) --
	(184.73,136.44);

\path[draw=drawColor,line width= 0.4pt,line join=round,line cap=round] (178.12,128.19) -- (178.12,130.33);

\path[draw=drawColor,line width= 0.4pt,line join=round,line cap=round] (175.95,128.19) --
	(178.12,128.19) --
	(180.29,128.19);

\path[draw=drawColor,line width= 0.4pt,line join=round,line cap=round] (180.29,130.33) --
	(178.12,130.33) --
	(175.95,130.33);

\path[draw=drawColor,line width= 0.4pt,line join=round,line cap=round] (169.35,121.85) -- (169.35,124.26);

\path[draw=drawColor,line width= 0.4pt,line join=round,line cap=round] (167.18,121.85) --
	(169.35,121.85) --
	(171.51,121.85);

\path[draw=drawColor,line width= 0.4pt,line join=round,line cap=round] (171.51,124.26) --
	(169.35,124.26) --
	(167.18,124.26);

\path[draw=drawColor,line width= 0.4pt,line join=round,line cap=round] (160.57,115.54) -- (160.57,117.69);

\path[draw=drawColor,line width= 0.4pt,line join=round,line cap=round] (158.40,115.54) --
	(160.57,115.54) --
	(162.74,115.54);

\path[draw=drawColor,line width= 0.4pt,line join=round,line cap=round] (162.74,117.69) --
	(160.57,117.69) --
	(158.40,117.69);

\path[draw=drawColor,line width= 0.4pt,line join=round,line cap=round] (151.79,109.34) -- (151.79,111.63);

\path[draw=drawColor,line width= 0.4pt,line join=round,line cap=round] (149.63,109.34) --
	(151.79,109.34) --
	(153.96,109.34);

\path[draw=drawColor,line width= 0.4pt,line join=round,line cap=round] (153.96,111.63) --
	(151.79,111.63) --
	(149.63,111.63);

\path[draw=drawColor,line width= 0.4pt,line join=round,line cap=round] (143.02,103.28) -- (143.02,105.35);

\path[draw=drawColor,line width= 0.4pt,line join=round,line cap=round] (140.85,103.28) --
	(143.02,103.28) --
	(145.19,103.28);

\path[draw=drawColor,line width= 0.4pt,line join=round,line cap=round] (145.19,105.35) --
	(143.02,105.35) --
	(140.85,105.35);

\path[draw=drawColor,line width= 0.4pt,line join=round,line cap=round] (134.25, 96.69) -- (134.25, 98.86);

\path[draw=drawColor,line width= 0.4pt,line join=round,line cap=round] (132.08, 96.69) --
	(134.25, 96.69) --
	(136.41, 96.69);

\path[draw=drawColor,line width= 0.4pt,line join=round,line cap=round] (136.41, 98.86) --
	(134.25, 98.86) --
	(132.08, 98.86);

\path[draw=drawColor,line width= 0.4pt,line join=round,line cap=round] (125.47, 90.60) -- (125.47, 92.65);

\path[draw=drawColor,line width= 0.4pt,line join=round,line cap=round] (123.30, 90.60) --
	(125.47, 90.60) --
	(127.64, 90.60);

\path[draw=drawColor,line width= 0.4pt,line join=round,line cap=round] (127.64, 92.65) --
	(125.47, 92.65) --
	(123.30, 92.65);

\path[draw=drawColor,line width= 0.4pt,line join=round,line cap=round] (116.70, 84.02) -- (116.70, 86.19);

\path[draw=drawColor,line width= 0.4pt,line join=round,line cap=round] (114.53, 84.02) --
	(116.70, 84.02) --
	(118.87, 84.02);

\path[draw=drawColor,line width= 0.4pt,line join=round,line cap=round] (118.87, 86.19) --
	(116.70, 86.19) --
	(114.53, 86.19);

\path[draw=drawColor,line width= 0.4pt,line join=round,line cap=round] (107.92, 77.57) -- (107.92, 79.91);

\path[draw=drawColor,line width= 0.4pt,line join=round,line cap=round] (105.75, 77.57) --
	(107.92, 77.57) --
	(110.09, 77.57);

\path[draw=drawColor,line width= 0.4pt,line join=round,line cap=round] (110.09, 79.91) --
	(107.92, 79.91) --
	(105.75, 79.91);

\path[draw=drawColor,line width= 0.4pt,line join=round,line cap=round] ( 99.14, 71.36) -- ( 99.14, 73.77);

\path[draw=drawColor,line width= 0.4pt,line join=round,line cap=round] ( 96.98, 71.36) --
	( 99.14, 71.36) --
	(101.31, 71.36);

\path[draw=drawColor,line width= 0.4pt,line join=round,line cap=round] (101.31, 73.77) --
	( 99.14, 73.77) --
	( 96.98, 73.77);

\path[draw=drawColor,line width= 0.4pt,line join=round,line cap=round] ( 90.37, 65.12) -- ( 90.37, 67.24);

\path[draw=drawColor,line width= 0.4pt,line join=round,line cap=round] ( 88.20, 65.12) --
	( 90.37, 65.12) --
	( 92.54, 65.12);

\path[draw=drawColor,line width= 0.4pt,line join=round,line cap=round] ( 92.54, 67.24) --
	( 90.37, 67.24) --
	( 88.20, 67.24);

\path[draw=drawColor,line width= 0.4pt,line join=round,line cap=round] ( 81.60, 58.78) -- ( 81.60, 61.11);

\path[draw=drawColor,line width= 0.4pt,line join=round,line cap=round] ( 79.43, 58.78) --
	( 81.60, 58.78) --
	( 83.76, 58.78);

\path[draw=drawColor,line width= 0.4pt,line join=round,line cap=round] ( 83.76, 61.11) --
	( 81.60, 61.11) --
	( 79.43, 61.11);

\path[draw=drawColor,line width= 0.4pt,line join=round,line cap=round] ( 72.82, 52.80) -- ( 72.82, 54.93);

\path[draw=drawColor,line width= 0.4pt,line join=round,line cap=round] ( 70.65, 52.80) --
	( 72.82, 52.80) --
	( 74.99, 52.80);

\path[draw=drawColor,line width= 0.4pt,line join=round,line cap=round] ( 74.99, 54.93) --
	( 72.82, 54.93) --
	( 70.65, 54.93);

\path[draw=drawColor,line width= 0.4pt,line join=round,line cap=round] ( 64.05, 46.56) -- ( 64.05, 48.66);

\path[draw=drawColor,line width= 0.4pt,line join=round,line cap=round] ( 61.88, 46.56) --
	( 64.05, 46.56) --
	( 66.22, 46.56);

\path[draw=drawColor,line width= 0.4pt,line join=round,line cap=round] ( 66.22, 48.66) --
	( 64.05, 48.66) --
	( 61.88, 48.66);

\path[draw=drawColor,line width= 0.4pt,line join=round,line cap=round] ( 55.27,  0.00) -- ( 55.27, 70.98);

\path[draw=drawColor,line width= 0.4pt,line join=round,line cap=round] ( 57.44, 70.98) --
	( 55.27, 70.98) --
	( 53.10, 70.98);

\path[draw=drawColor,line width= 0.4pt,line join=round,line cap=round] ( 46.50, 34.86) -- ( 46.50, 41.39);

\path[draw=drawColor,line width= 0.4pt,line join=round,line cap=round] ( 44.33, 34.86) --
	( 46.50, 34.86) --
	( 48.66, 34.86);

\path[draw=drawColor,line width= 0.4pt,line join=round,line cap=round] ( 48.66, 41.39) --
	( 46.50, 41.39) --
	( 44.33, 41.39);

\path[draw=drawColor,line width= 0.4pt,line join=round,line cap=round] ( 37.72, 32.18) -- ( 37.72, 39.93);

\path[draw=drawColor,line width= 0.4pt,line join=round,line cap=round] ( 35.55, 32.18) --
	( 37.72, 32.18) --
	( 39.89, 32.18);

\path[draw=drawColor,line width= 0.4pt,line join=round,line cap=round] ( 39.89, 39.93) --
	( 37.72, 39.93) --
	( 35.55, 39.93);
\definecolor{drawColor}{RGB}{254,175,119}

\path[draw=drawColor,line width= 0.4pt,line join=round,line cap=round] (167.38, 77.70) -- (185.38, 77.70);
\definecolor{drawColor}{RGB}{241,96,93}

\path[draw=drawColor,line width= 0.4pt,line join=round,line cap=round] (167.38, 69.30) -- (185.38, 69.30);
\definecolor{drawColor}{RGB}{182,54,121}

\path[draw=drawColor,line width= 0.4pt,line join=round,line cap=round] (167.38, 60.90) -- (185.38, 60.90);
\definecolor{drawColor}{RGB}{114,31,129}

\path[draw=drawColor,line width= 0.4pt,line join=round,line cap=round] (167.38, 52.50) -- (185.38, 52.50);
\definecolor{drawColor}{RGB}{45,17,96}

\path[draw=drawColor,line width= 0.4pt,line join=round,line cap=round] (167.38, 44.10) -- (185.38, 44.10);
\definecolor{drawColor}{RGB}{0,0,0}

\node[text=drawColor,anchor=base west,inner sep=0pt, outer sep=0pt, scale=  1.00] at (192.58, 74.25) {{\footnotesize $m = 240$}};

\node[text=drawColor,anchor=base west,inner sep=0pt, outer sep=0pt, scale=  1.00] at (192.58, 65.85) {\footnotesize $m = 280$};

\node[text=drawColor,anchor=base west,inner sep=0pt, outer sep=0pt, scale=  1.00] at (192.58, 57.45) {\footnotesize  $m=340$};

\node[text=drawColor,anchor=base west,inner sep=0pt, outer sep=0pt, scale=  1.00] at (192.58, 49.05) {\footnotesize  $m=600$};

\node[text=drawColor,anchor=base west,inner sep=0pt, outer sep=0pt, scale=  1.00] at (192.58, 40.65) {\footnotesize  $m=1800$};
\end{scope}
\begin{scope}
\path[clip] (  0.00,  0.00) rectangle (238.49,180.67);
\definecolor{drawColor}{RGB}{0,0,0}

\node[text=drawColor,anchor=base,inner sep=0pt, outer sep=0pt, scale=  1.00] at (134.25,  2.40) {$\log_{10}(\eta)$};

\node[text=drawColor,rotate= 90.00,anchor=base,inner sep=0pt, outer sep=0pt, scale=  1.00] at (  9.60,105.34) {$\log_{10}(\recerror)$};
\end{scope}
\end{tikzpicture}

%% file: Deltaeta_3qubits_Haar_SDPT3.fig.tex
\begin{tikzpicture}[x=1pt,y=1pt]
\definecolor{fillColor}{RGB}{255,255,255}
\path[use as bounding box,fill=fillColor,fill opacity=0.00] (0,0) rectangle (238.49,180.67);
\begin{scope}
\path[clip] ( 30.00, 30.00) rectangle (238.49,180.67);
\definecolor{drawColor}{RGB}{254,175,119}

\path[draw=drawColor,line width= 0.4pt,line join=round,line cap=round] (230.77,176.18) --
	(221.99,175.84) --
	(213.22,175.04) --
	(204.44,174.15) --
	(195.67,172.69) --
	(186.89,170.87) --
	(178.12,167.70) --
	(169.35,165.46) --
	(160.57,161.68) --
	(151.79,158.29) --
	(143.02,155.31) --
	(134.25,151.10) --
	(125.47,146.76) --
	(116.70,144.21) --
	(107.92,140.50) --
	( 99.14,135.77) --
	( 90.37,134.35) --
	( 81.60,135.09) --
	( 72.82,127.35) --
	( 64.05,126.21) --
	( 55.27,130.50) --
	( 46.50,124.44) --
	( 37.72,132.73);
\end{scope}
\begin{scope}
\path[clip] (  0.00,  0.00) rectangle (238.49,180.67);
\definecolor{drawColor}{RGB}{0,0,0}

\path[draw=drawColor,line width= 0.4pt,line join=round,line cap=round] ( 37.72, 30.00) -- (230.77, 30.00);

\path[draw=drawColor,line width= 0.4pt,line join=round,line cap=round] ( 37.72, 30.00) -- ( 37.72, 26.40);

\path[draw=drawColor,line width= 0.4pt,line join=round,line cap=round] ( 85.98, 30.00) -- ( 85.98, 26.40);

\path[draw=drawColor,line width= 0.4pt,line join=round,line cap=round] (134.25, 30.00) -- (134.25, 26.40);

\path[draw=drawColor,line width= 0.4pt,line join=round,line cap=round] (182.51, 30.00) -- (182.51, 26.40);

\path[draw=drawColor,line width= 0.4pt,line join=round,line cap=round] (230.77, 30.00) -- (230.77, 26.40);

\node[text=drawColor,anchor=base,inner sep=0pt, outer sep=0pt, scale=  1.00] at ( 37.72, 14.40) {-4};

\node[text=drawColor,anchor=base,inner sep=0pt, outer sep=0pt, scale=  1.00] at ( 85.98, 14.40) {-3};

\node[text=drawColor,anchor=base,inner sep=0pt, outer sep=0pt, scale=  1.00] at (134.25, 14.40) {-2};

\node[text=drawColor,anchor=base,inner sep=0pt, outer sep=0pt, scale=  1.00] at (182.51, 14.40) {-1};

\node[text=drawColor,anchor=base,inner sep=0pt, outer sep=0pt, scale=  1.00] at (230.77, 14.40) {0};

\path[draw=drawColor,line width= 0.4pt,line join=round,line cap=round] ( 30.00, 35.58) -- ( 30.00,175.09);

\path[draw=drawColor,line width= 0.4pt,line join=round,line cap=round] ( 30.00, 35.58) -- ( 26.40, 35.58);

\path[draw=drawColor,line width= 0.4pt,line join=round,line cap=round] ( 30.00, 70.46) -- ( 26.40, 70.46);

\path[draw=drawColor,line width= 0.4pt,line join=round,line cap=round] ( 30.00,105.34) -- ( 26.40,105.34);

\path[draw=drawColor,line width= 0.4pt,line join=round,line cap=round] ( 30.00,140.22) -- ( 26.40,140.22);

\path[draw=drawColor,line width= 0.4pt,line join=round,line cap=round] ( 30.00,175.09) -- ( 26.40,175.09);

\node[text=drawColor,rotate= 90.00,anchor=base,inner sep=0pt, outer sep=0pt, scale=  1.00] at ( 21.60, 35.58) {-4};

\node[text=drawColor,rotate= 90.00,anchor=base,inner sep=0pt, outer sep=0pt, scale=  1.00] at ( 21.60, 70.46) {-3};

\node[text=drawColor,rotate= 90.00,anchor=base,inner sep=0pt, outer sep=0pt, scale=  1.00] at ( 21.60,105.34) {-2};

\node[text=drawColor,rotate= 90.00,anchor=base,inner sep=0pt, outer sep=0pt, scale=  1.00] at ( 21.60,140.22) {-1};

\node[text=drawColor,rotate= 90.00,anchor=base,inner sep=0pt, outer sep=0pt, scale=  1.00] at ( 21.60,175.09) {0};

\path[draw=drawColor,line width= 0.4pt,line join=round,line cap=round] ( 30.00, 30.00) --
	(238.49, 30.00) --
	(238.49,180.67) --
	( 30.00,180.67) --
	( 30.00, 30.00);
\end{scope}
\begin{scope}
\path[clip] ( 30.00, 30.00) rectangle (238.49,180.67);
\definecolor{drawColor}{RGB}{254,175,119}

\path[draw=drawColor,line width= 0.4pt,line join=round,line cap=round] (230.77,175.90) -- (230.77,176.45);

\path[draw=drawColor,line width= 0.4pt,line join=round,line cap=round] (228.60,175.90) --
	(230.77,175.90) --
	(232.94,175.90);

\path[draw=drawColor,line width= 0.4pt,line join=round,line cap=round] (232.94,176.45) --
	(230.77,176.45) --
	(228.60,176.45);

\path[draw=drawColor,line width= 0.4pt,line join=round,line cap=round] (221.99,175.54) -- (221.99,176.13);

\path[draw=drawColor,line width= 0.4pt,line join=round,line cap=round] (219.83,175.54) --
	(221.99,175.54) --
	(224.16,175.54);

\path[draw=drawColor,line width= 0.4pt,line join=round,line cap=round] (224.16,176.13) --
	(221.99,176.13) --
	(219.83,176.13);

\path[draw=drawColor,line width= 0.4pt,line join=round,line cap=round] (213.22,174.53) -- (213.22,175.54);

\path[draw=drawColor,line width= 0.4pt,line join=round,line cap=round] (211.05,174.53) --
	(213.22,174.53) --
	(215.39,174.53);

\path[draw=drawColor,line width= 0.4pt,line join=round,line cap=round] (215.39,175.54) --
	(213.22,175.54) --
	(211.05,175.54);

\path[draw=drawColor,line width= 0.4pt,line join=round,line cap=round] (204.44,173.52) -- (204.44,174.75);

\path[draw=drawColor,line width= 0.4pt,line join=round,line cap=round] (202.28,173.52) --
	(204.44,173.52) --
	(206.61,173.52);

\path[draw=drawColor,line width= 0.4pt,line join=round,line cap=round] (206.61,174.75) --
	(204.44,174.75) --
	(202.28,174.75);

\path[draw=drawColor,line width= 0.4pt,line join=round,line cap=round] (195.67,171.78) -- (195.67,173.56);

\path[draw=drawColor,line width= 0.4pt,line join=round,line cap=round] (193.50,171.78) --
	(195.67,171.78) --
	(197.84,171.78);

\path[draw=drawColor,line width= 0.4pt,line join=round,line cap=round] (197.84,173.56) --
	(195.67,173.56) --
	(193.50,173.56);

\path[draw=drawColor,line width= 0.4pt,line join=round,line cap=round] (186.89,169.57) -- (186.89,172.07);

\path[draw=drawColor,line width= 0.4pt,line join=round,line cap=round] (184.73,169.57) --
	(186.89,169.57) --
	(189.06,169.57);

\path[draw=drawColor,line width= 0.4pt,line join=round,line cap=round] (189.06,172.07) --
	(186.89,172.07) --
	(184.73,172.07);

\path[draw=drawColor,line width= 0.4pt,line join=round,line cap=round] (178.12,166.11) -- (178.12,169.13);

\path[draw=drawColor,line width= 0.4pt,line join=round,line cap=round] (175.95,166.11) --
	(178.12,166.11) --
	(180.29,166.11);

\path[draw=drawColor,line width= 0.4pt,line join=round,line cap=round] (180.29,169.13) --
	(178.12,169.13) --
	(175.95,169.13);

\path[draw=drawColor,line width= 0.4pt,line join=round,line cap=round] (169.35,163.45) -- (169.35,167.23);

\path[draw=drawColor,line width= 0.4pt,line join=round,line cap=round] (167.18,163.45) --
	(169.35,163.45) --
	(171.51,163.45);

\path[draw=drawColor,line width= 0.4pt,line join=round,line cap=round] (171.51,167.23) --
	(169.35,167.23) --
	(167.18,167.23);

\path[draw=drawColor,line width= 0.4pt,line join=round,line cap=round] (160.57,158.98) -- (160.57,163.96);

\path[draw=drawColor,line width= 0.4pt,line join=round,line cap=round] (158.40,158.98) --
	(160.57,158.98) --
	(162.74,158.98);

\path[draw=drawColor,line width= 0.4pt,line join=round,line cap=round] (162.74,163.96) --
	(160.57,163.96) --
	(158.40,163.96);

\path[draw=drawColor,line width= 0.4pt,line join=round,line cap=round] (151.79,155.26) -- (151.79,160.82);

\path[draw=drawColor,line width= 0.4pt,line join=round,line cap=round] (149.63,155.26) --
	(151.79,155.26) --
	(153.96,155.26);

\path[draw=drawColor,line width= 0.4pt,line join=round,line cap=round] (153.96,160.82) --
	(151.79,160.82) --
	(149.63,160.82);

\path[draw=drawColor,line width= 0.4pt,line join=round,line cap=round] (143.02,150.71) -- (143.02,158.84);

\path[draw=drawColor,line width= 0.4pt,line join=round,line cap=round] (140.85,150.71) --
	(143.02,150.71) --
	(145.19,150.71);

\path[draw=drawColor,line width= 0.4pt,line join=round,line cap=round] (145.19,158.84) --
	(143.02,158.84) --
	(140.85,158.84);

\path[draw=drawColor,line width= 0.4pt,line join=round,line cap=round] (134.25,146.42) -- (134.25,154.68);

\path[draw=drawColor,line width= 0.4pt,line join=round,line cap=round] (132.08,146.42) --
	(134.25,146.42) --
	(136.41,146.42);

\path[draw=drawColor,line width= 0.4pt,line join=round,line cap=round] (136.41,154.68) --
	(134.25,154.68) --
	(132.08,154.68);

\path[draw=drawColor,line width= 0.4pt,line join=round,line cap=round] (125.47,139.89) -- (125.47,151.46);

\path[draw=drawColor,line width= 0.4pt,line join=round,line cap=round] (123.30,139.89) --
	(125.47,139.89) --
	(127.64,139.89);

\path[draw=drawColor,line width= 0.4pt,line join=round,line cap=round] (127.64,151.46) --
	(125.47,151.46) --
	(123.30,151.46);

\path[draw=drawColor,line width= 0.4pt,line join=round,line cap=round] (116.70,136.91) -- (116.70,149.12);

\path[draw=drawColor,line width= 0.4pt,line join=round,line cap=round] (114.53,136.91) --
	(116.70,136.91) --
	(118.87,136.91);

\path[draw=drawColor,line width= 0.4pt,line join=round,line cap=round] (118.87,149.12) --
	(116.70,149.12) --
	(114.53,149.12);

\path[draw=drawColor,line width= 0.4pt,line join=round,line cap=round] (107.92,127.83) -- (107.92,147.30);

\path[draw=drawColor,line width= 0.4pt,line join=round,line cap=round] (105.75,127.83) --
	(107.92,127.83) --
	(110.09,127.83);

\path[draw=drawColor,line width= 0.4pt,line join=round,line cap=round] (110.09,147.30) --
	(107.92,147.30) --
	(105.75,147.30);

\path[draw=drawColor,line width= 0.4pt,line join=round,line cap=round] ( 99.14,121.36) -- ( 99.14,143.01);

\path[draw=drawColor,line width= 0.4pt,line join=round,line cap=round] ( 96.98,121.36) --
	( 99.14,121.36) --
	(101.31,121.36);

\path[draw=drawColor,line width= 0.4pt,line join=round,line cap=round] (101.31,143.01) --
	( 99.14,143.01) --
	( 96.98,143.01);

\path[draw=drawColor,line width= 0.4pt,line join=round,line cap=round] ( 90.37, 94.56) -- ( 90.37,144.29);

\path[draw=drawColor,line width= 0.4pt,line join=round,line cap=round] ( 88.20, 94.56) --
	( 90.37, 94.56) --
	( 92.54, 94.56);

\path[draw=drawColor,line width= 0.4pt,line join=round,line cap=round] ( 92.54,144.29) --
	( 90.37,144.29) --
	( 88.20,144.29);

\path[draw=drawColor,line width= 0.4pt,line join=round,line cap=round] ( 81.60,  0.00) -- ( 81.60,147.05);

\path[draw=drawColor,line width= 0.4pt,line join=round,line cap=round] ( 83.76,147.05) --
	( 81.60,147.05) --
	( 79.43,147.05);

\path[draw=drawColor,line width= 0.4pt,line join=round,line cap=round] ( 72.82,  0.00) -- ( 72.82,139.70);

\path[draw=drawColor,line width= 0.4pt,line join=round,line cap=round] ( 74.99,139.70) --
	( 72.82,139.70) --
	( 70.65,139.70);

\path[draw=drawColor,line width= 0.4pt,line join=round,line cap=round] ( 64.05,  0.00) -- ( 64.05,139.82);

\path[draw=drawColor,line width= 0.4pt,line join=round,line cap=round] ( 66.22,139.82) --
	( 64.05,139.82) --
	( 61.88,139.82);

\path[draw=drawColor,line width= 0.4pt,line join=round,line cap=round] ( 55.27,  0.00) -- ( 55.27,144.96);

\path[draw=drawColor,line width= 0.4pt,line join=round,line cap=round] ( 57.44,144.96) --
	( 55.27,144.96) --
	( 53.10,144.96);

\path[draw=drawColor,line width= 0.4pt,line join=round,line cap=round] ( 46.50,  0.00) -- ( 46.50,137.34);

\path[draw=drawColor,line width= 0.4pt,line join=round,line cap=round] ( 48.66,137.34) --
	( 46.50,137.34) --
	( 44.33,137.34);

\path[draw=drawColor,line width= 0.4pt,line join=round,line cap=round] ( 37.72,  0.00) -- ( 37.72,145.64);

\path[draw=drawColor,line width= 0.4pt,line join=round,line cap=round] ( 39.89,145.64) --
	( 37.72,145.64) --
	( 35.55,145.64);
\definecolor{drawColor}{RGB}{241,96,93}

\path[draw=drawColor,line width= 0.4pt,line join=round,line cap=round] (230.77,176.13) --
	(221.99,175.61) --
	(213.22,174.79) --
	(204.44,173.12) --
	(195.67,171.60) --
	(186.89,168.43) --
	(178.12,165.37) --
	(169.35,161.22) --
	(160.57,157.40) --
	(151.79,153.31) --
	(143.02,147.55) --
	(134.25,142.49) --
	(125.47,137.62) --
	(116.70,131.85) --
	(107.92,126.10) --
	( 99.14,119.86) --
	( 90.37,113.96) --
	( 81.60,106.81) --
	( 72.82,101.26) --
	( 64.05, 96.49) --
	( 55.27, 89.83) --
	( 46.50, 84.52) --
	( 37.72, 82.75);

\path[draw=drawColor,line width= 0.4pt,line join=round,line cap=round] (230.77,175.89) -- (230.77,176.37);

\path[draw=drawColor,line width= 0.4pt,line join=round,line cap=round] (228.60,175.89) --
	(230.77,175.89) --
	(232.94,175.89);

\path[draw=drawColor,line width= 0.4pt,line join=round,line cap=round] (232.94,176.37) --
	(230.77,176.37) --
	(228.60,176.37);

\path[draw=drawColor,line width= 0.4pt,line join=round,line cap=round] (221.99,175.24) -- (221.99,175.97);

\path[draw=drawColor,line width= 0.4pt,line join=round,line cap=round] (219.83,175.24) --
	(221.99,175.24) --
	(224.16,175.24);

\path[draw=drawColor,line width= 0.4pt,line join=round,line cap=round] (224.16,175.97) --
	(221.99,175.97) --
	(219.83,175.97);

\path[draw=drawColor,line width= 0.4pt,line join=round,line cap=round] (213.22,174.19) -- (213.22,175.37);

\path[draw=drawColor,line width= 0.4pt,line join=round,line cap=round] (211.05,174.19) --
	(213.22,174.19) --
	(215.39,174.19);

\path[draw=drawColor,line width= 0.4pt,line join=round,line cap=round] (215.39,175.37) --
	(213.22,175.37) --
	(211.05,175.37);

\path[draw=drawColor,line width= 0.4pt,line join=round,line cap=round] (204.44,172.28) -- (204.44,173.92);

\path[draw=drawColor,line width= 0.4pt,line join=round,line cap=round] (202.28,172.28) --
	(204.44,172.28) --
	(206.61,172.28);

\path[draw=drawColor,line width= 0.4pt,line join=round,line cap=round] (206.61,173.92) --
	(204.44,173.92) --
	(202.28,173.92);

\path[draw=drawColor,line width= 0.4pt,line join=round,line cap=round] (195.67,170.38) -- (195.67,172.72);

\path[draw=drawColor,line width= 0.4pt,line join=round,line cap=round] (193.50,170.38) --
	(195.67,170.38) --
	(197.84,170.38);

\path[draw=drawColor,line width= 0.4pt,line join=round,line cap=round] (197.84,172.72) --
	(195.67,172.72) --
	(193.50,172.72);

\path[draw=drawColor,line width= 0.4pt,line join=round,line cap=round] (186.89,167.08) -- (186.89,169.67);

\path[draw=drawColor,line width= 0.4pt,line join=round,line cap=round] (184.73,167.08) --
	(186.89,167.08) --
	(189.06,167.08);

\path[draw=drawColor,line width= 0.4pt,line join=round,line cap=round] (189.06,169.67) --
	(186.89,169.67) --
	(184.73,169.67);

\path[draw=drawColor,line width= 0.4pt,line join=round,line cap=round] (178.12,163.61) -- (178.12,166.95);

\path[draw=drawColor,line width= 0.4pt,line join=round,line cap=round] (175.95,163.61) --
	(178.12,163.61) --
	(180.29,163.61);

\path[draw=drawColor,line width= 0.4pt,line join=round,line cap=round] (180.29,166.95) --
	(178.12,166.95) --
	(175.95,166.95);

\path[draw=drawColor,line width= 0.4pt,line join=round,line cap=round] (169.35,159.20) -- (169.35,163.01);

\path[draw=drawColor,line width= 0.4pt,line join=round,line cap=round] (167.18,159.20) --
	(169.35,159.20) --
	(171.51,159.20);

\path[draw=drawColor,line width= 0.4pt,line join=round,line cap=round] (171.51,163.01) --
	(169.35,163.01) --
	(167.18,163.01);

\path[draw=drawColor,line width= 0.4pt,line join=round,line cap=round] (160.57,155.04) -- (160.57,159.44);

\path[draw=drawColor,line width= 0.4pt,line join=round,line cap=round] (158.40,155.04) --
	(160.57,155.04) --
	(162.74,155.04);

\path[draw=drawColor,line width= 0.4pt,line join=round,line cap=round] (162.74,159.44) --
	(160.57,159.44) --
	(158.40,159.44);

\path[draw=drawColor,line width= 0.4pt,line join=round,line cap=round] (151.79,150.33) -- (151.79,155.80);

\path[draw=drawColor,line width= 0.4pt,line join=round,line cap=round] (149.63,150.33) --
	(151.79,150.33) --
	(153.96,150.33);

\path[draw=drawColor,line width= 0.4pt,line join=round,line cap=round] (153.96,155.80) --
	(151.79,155.80) --
	(149.63,155.80);

\path[draw=drawColor,line width= 0.4pt,line join=round,line cap=round] (143.02,144.28) -- (143.02,150.24);

\path[draw=drawColor,line width= 0.4pt,line join=round,line cap=round] (140.85,144.28) --
	(143.02,144.28) --
	(145.19,144.28);

\path[draw=drawColor,line width= 0.4pt,line join=round,line cap=round] (145.19,150.24) --
	(143.02,150.24) --
	(140.85,150.24);

\path[draw=drawColor,line width= 0.4pt,line join=round,line cap=round] (134.25,139.69) -- (134.25,144.84);

\path[draw=drawColor,line width= 0.4pt,line join=round,line cap=round] (132.08,139.69) --
	(134.25,139.69) --
	(136.41,139.69);

\path[draw=drawColor,line width= 0.4pt,line join=round,line cap=round] (136.41,144.84) --
	(134.25,144.84) --
	(132.08,144.84);

\path[draw=drawColor,line width= 0.4pt,line join=round,line cap=round] (125.47,131.76) -- (125.47,141.84);

\path[draw=drawColor,line width= 0.4pt,line join=round,line cap=round] (123.30,131.76) --
	(125.47,131.76) --
	(127.64,131.76);

\path[draw=drawColor,line width= 0.4pt,line join=round,line cap=round] (127.64,141.84) --
	(125.47,141.84) --
	(123.30,141.84);

\path[draw=drawColor,line width= 0.4pt,line join=round,line cap=round] (116.70,127.81) -- (116.70,135.03);

\path[draw=drawColor,line width= 0.4pt,line join=round,line cap=round] (114.53,127.81) --
	(116.70,127.81) --
	(118.87,127.81);

\path[draw=drawColor,line width= 0.4pt,line join=round,line cap=round] (118.87,135.03) --
	(116.70,135.03) --
	(114.53,135.03);

\path[draw=drawColor,line width= 0.4pt,line join=round,line cap=round] (107.92,121.55) -- (107.92,129.59);

\path[draw=drawColor,line width= 0.4pt,line join=round,line cap=round] (105.75,121.55) --
	(107.92,121.55) --
	(110.09,121.55);

\path[draw=drawColor,line width= 0.4pt,line join=round,line cap=round] (110.09,129.59) --
	(107.92,129.59) --
	(105.75,129.59);

\path[draw=drawColor,line width= 0.4pt,line join=round,line cap=round] ( 99.14,114.93) -- ( 99.14,123.57);

\path[draw=drawColor,line width= 0.4pt,line join=round,line cap=round] ( 96.98,114.93) --
	( 99.14,114.93) --
	(101.31,114.93);

\path[draw=drawColor,line width= 0.4pt,line join=round,line cap=round] (101.31,123.57) --
	( 99.14,123.57) --
	( 96.98,123.57);

\path[draw=drawColor,line width= 0.4pt,line join=round,line cap=round] ( 90.37,109.29) -- ( 90.37,117.53);

\path[draw=drawColor,line width= 0.4pt,line join=round,line cap=round] ( 88.20,109.29) --
	( 90.37,109.29) --
	( 92.54,109.29);

\path[draw=drawColor,line width= 0.4pt,line join=round,line cap=round] ( 92.54,117.53) --
	( 90.37,117.53) --
	( 88.20,117.53);

\path[draw=drawColor,line width= 0.4pt,line join=round,line cap=round] ( 81.60,101.77) -- ( 81.60,110.58);

\path[draw=drawColor,line width= 0.4pt,line join=round,line cap=round] ( 79.43,101.77) --
	( 81.60,101.77) --
	( 83.76,101.77);

\path[draw=drawColor,line width= 0.4pt,line join=round,line cap=round] ( 83.76,110.58) --
	( 81.60,110.58) --
	( 79.43,110.58);

\path[draw=drawColor,line width= 0.4pt,line join=round,line cap=round] ( 72.82, 95.98) -- ( 72.82,105.16);

\path[draw=drawColor,line width= 0.4pt,line join=round,line cap=round] ( 70.65, 95.98) --
	( 72.82, 95.98) --
	( 74.99, 95.98);

\path[draw=drawColor,line width= 0.4pt,line join=round,line cap=round] ( 74.99,105.16) --
	( 72.82,105.16) --
	( 70.65,105.16);

\path[draw=drawColor,line width= 0.4pt,line join=round,line cap=round] ( 64.05, 91.09) -- ( 64.05,100.46);

\path[draw=drawColor,line width= 0.4pt,line join=round,line cap=round] ( 61.88, 91.09) --
	( 64.05, 91.09) --
	( 66.22, 91.09);

\path[draw=drawColor,line width= 0.4pt,line join=round,line cap=round] ( 66.22,100.46) --
	( 64.05,100.46) --
	( 61.88,100.46);

\path[draw=drawColor,line width= 0.4pt,line join=round,line cap=round] ( 55.27, 84.34) -- ( 55.27, 93.85);

\path[draw=drawColor,line width= 0.4pt,line join=round,line cap=round] ( 53.10, 84.34) --
	( 55.27, 84.34) --
	( 57.44, 84.34);

\path[draw=drawColor,line width= 0.4pt,line join=round,line cap=round] ( 57.44, 93.85) --
	( 55.27, 93.85) --
	( 53.10, 93.85);

\path[draw=drawColor,line width= 0.4pt,line join=round,line cap=round] ( 46.50, 79.71) -- ( 46.50, 88.16);

\path[draw=drawColor,line width= 0.4pt,line join=round,line cap=round] ( 44.33, 79.71) --
	( 46.50, 79.71) --
	( 48.66, 79.71);

\path[draw=drawColor,line width= 0.4pt,line join=round,line cap=round] ( 48.66, 88.16) --
	( 46.50, 88.16) --
	( 44.33, 88.16);

\path[draw=drawColor,line width= 0.4pt,line join=round,line cap=round] ( 37.72, 69.49) -- ( 37.72, 89.71);

\path[draw=drawColor,line width= 0.4pt,line join=round,line cap=round] ( 35.55, 69.49) --
	( 37.72, 69.49) --
	( 39.89, 69.49);

\path[draw=drawColor,line width= 0.4pt,line join=round,line cap=round] ( 39.89, 89.71) --
	( 37.72, 89.71) --
	( 35.55, 89.71);
\definecolor{drawColor}{RGB}{182,54,121}

\path[draw=drawColor,line width= 0.4pt,line join=round,line cap=round] (230.77,175.92) --
	(221.99,175.20) --
	(213.22,173.89) --
	(204.44,171.94) --
	(195.67,169.27) --
	(186.89,165.39) --
	(178.12,161.18) --
	(169.35,156.72) --
	(160.57,151.37) --
	(151.79,146.21) --
	(143.02,140.65) --
	(134.25,135.16) --
	(125.47,128.19) --
	(116.70,122.58) --
	(107.92,116.24) --
	( 99.14,110.24) --
	( 90.37,103.33) --
	( 81.60, 97.66) --
	( 72.82, 91.27) --
	( 64.05, 85.18) --
	( 55.27, 78.64) --
	( 46.50, 73.90) --
	( 37.72, 70.95);

\path[draw=drawColor,line width= 0.4pt,line join=round,line cap=round] (230.77,175.62) -- (230.77,176.22);

\path[draw=drawColor,line width= 0.4pt,line join=round,line cap=round] (228.60,175.62) --
	(230.77,175.62) --
	(232.94,175.62);

\path[draw=drawColor,line width= 0.4pt,line join=round,line cap=round] (232.94,176.22) --
	(230.77,176.22) --
	(228.60,176.22);

\path[draw=drawColor,line width= 0.4pt,line join=round,line cap=round] (221.99,174.72) -- (221.99,175.67);

\path[draw=drawColor,line width= 0.4pt,line join=round,line cap=round] (219.83,174.72) --
	(221.99,174.72) --
	(224.16,174.72);

\path[draw=drawColor,line width= 0.4pt,line join=round,line cap=round] (224.16,175.67) --
	(221.99,175.67) --
	(219.83,175.67);

\path[draw=drawColor,line width= 0.4pt,line join=round,line cap=round] (213.22,172.99) -- (213.22,174.74);

\path[draw=drawColor,line width= 0.4pt,line join=round,line cap=round] (211.05,172.99) --
	(213.22,172.99) --
	(215.39,172.99);

\path[draw=drawColor,line width= 0.4pt,line join=round,line cap=round] (215.39,174.74) --
	(213.22,174.74) --
	(211.05,174.74);

\path[draw=drawColor,line width= 0.4pt,line join=round,line cap=round] (204.44,170.87) -- (204.44,172.94);

\path[draw=drawColor,line width= 0.4pt,line join=round,line cap=round] (202.28,170.87) --
	(204.44,170.87) --
	(206.61,170.87);

\path[draw=drawColor,line width= 0.4pt,line join=round,line cap=round] (206.61,172.94) --
	(204.44,172.94) --
	(202.28,172.94);

\path[draw=drawColor,line width= 0.4pt,line join=round,line cap=round] (195.67,168.11) -- (195.67,170.34);

\path[draw=drawColor,line width= 0.4pt,line join=round,line cap=round] (193.50,168.11) --
	(195.67,168.11) --
	(197.84,168.11);

\path[draw=drawColor,line width= 0.4pt,line join=round,line cap=round] (197.84,170.34) --
	(195.67,170.34) --
	(193.50,170.34);

\path[draw=drawColor,line width= 0.4pt,line join=round,line cap=round] (186.89,163.88) -- (186.89,166.76);

\path[draw=drawColor,line width= 0.4pt,line join=round,line cap=round] (184.73,163.88) --
	(186.89,163.88) --
	(189.06,163.88);

\path[draw=drawColor,line width= 0.4pt,line join=round,line cap=round] (189.06,166.76) --
	(186.89,166.76) --
	(184.73,166.76);

\path[draw=drawColor,line width= 0.4pt,line join=round,line cap=round] (178.12,159.60) -- (178.12,162.61);

\path[draw=drawColor,line width= 0.4pt,line join=round,line cap=round] (175.95,159.60) --
	(178.12,159.60) --
	(180.29,159.60);

\path[draw=drawColor,line width= 0.4pt,line join=round,line cap=round] (180.29,162.61) --
	(178.12,162.61) --
	(175.95,162.61);

\path[draw=drawColor,line width= 0.4pt,line join=round,line cap=round] (169.35,154.73) -- (169.35,158.48);

\path[draw=drawColor,line width= 0.4pt,line join=round,line cap=round] (167.18,154.73) --
	(169.35,154.73) --
	(171.51,154.73);

\path[draw=drawColor,line width= 0.4pt,line join=round,line cap=round] (171.51,158.48) --
	(169.35,158.48) --
	(167.18,158.48);

\path[draw=drawColor,line width= 0.4pt,line join=round,line cap=round] (160.57,149.52) -- (160.57,153.02);

\path[draw=drawColor,line width= 0.4pt,line join=round,line cap=round] (158.40,149.52) --
	(160.57,149.52) --
	(162.74,149.52);

\path[draw=drawColor,line width= 0.4pt,line join=round,line cap=round] (162.74,153.02) --
	(160.57,153.02) --
	(158.40,153.02);

\path[draw=drawColor,line width= 0.4pt,line join=round,line cap=round] (151.79,143.97) -- (151.79,148.16);

\path[draw=drawColor,line width= 0.4pt,line join=round,line cap=round] (149.63,143.97) --
	(151.79,143.97) --
	(153.96,143.97);

\path[draw=drawColor,line width= 0.4pt,line join=round,line cap=round] (153.96,148.16) --
	(151.79,148.16) --
	(149.63,148.16);

\path[draw=drawColor,line width= 0.4pt,line join=round,line cap=round] (143.02,138.27) -- (143.02,142.70);

\path[draw=drawColor,line width= 0.4pt,line join=round,line cap=round] (140.85,138.27) --
	(143.02,138.27) --
	(145.19,138.27);

\path[draw=drawColor,line width= 0.4pt,line join=round,line cap=round] (145.19,142.70) --
	(143.02,142.70) --
	(140.85,142.70);

\path[draw=drawColor,line width= 0.4pt,line join=round,line cap=round] (134.25,132.87) -- (134.25,137.15);

\path[draw=drawColor,line width= 0.4pt,line join=round,line cap=round] (132.08,132.87) --
	(134.25,132.87) --
	(136.41,132.87);

\path[draw=drawColor,line width= 0.4pt,line join=round,line cap=round] (136.41,137.15) --
	(134.25,137.15) --
	(132.08,137.15);

\path[draw=drawColor,line width= 0.4pt,line join=round,line cap=round] (125.47,125.68) -- (125.47,130.34);

\path[draw=drawColor,line width= 0.4pt,line join=round,line cap=round] (123.30,125.68) --
	(125.47,125.68) --
	(127.64,125.68);

\path[draw=drawColor,line width= 0.4pt,line join=round,line cap=round] (127.64,130.34) --
	(125.47,130.34) --
	(123.30,130.34);

\path[draw=drawColor,line width= 0.4pt,line join=round,line cap=round] (116.70,120.38) -- (116.70,124.51);

\path[draw=drawColor,line width= 0.4pt,line join=round,line cap=round] (114.53,120.38) --
	(116.70,120.38) --
	(118.87,120.38);

\path[draw=drawColor,line width= 0.4pt,line join=round,line cap=round] (118.87,124.51) --
	(116.70,124.51) --
	(114.53,124.51);

\path[draw=drawColor,line width= 0.4pt,line join=round,line cap=round] (107.92,113.50) -- (107.92,118.57);

\path[draw=drawColor,line width= 0.4pt,line join=round,line cap=round] (105.75,113.50) --
	(107.92,113.50) --
	(110.09,113.50);

\path[draw=drawColor,line width= 0.4pt,line join=round,line cap=round] (110.09,118.57) --
	(107.92,118.57) --
	(105.75,118.57);

\path[draw=drawColor,line width= 0.4pt,line join=round,line cap=round] ( 99.14,107.95) -- ( 99.14,112.23);

\path[draw=drawColor,line width= 0.4pt,line join=round,line cap=round] ( 96.98,107.95) --
	( 99.14,107.95) --
	(101.31,107.95);

\path[draw=drawColor,line width= 0.4pt,line join=round,line cap=round] (101.31,112.23) --
	( 99.14,112.23) --
	( 96.98,112.23);

\path[draw=drawColor,line width= 0.4pt,line join=round,line cap=round] ( 90.37,100.82) -- ( 90.37,105.47);

\path[draw=drawColor,line width= 0.4pt,line join=round,line cap=round] ( 88.20,100.82) --
	( 90.37,100.82) --
	( 92.54,100.82);

\path[draw=drawColor,line width= 0.4pt,line join=round,line cap=round] ( 92.54,105.47) --
	( 90.37,105.47) --
	( 88.20,105.47);

\path[draw=drawColor,line width= 0.4pt,line join=round,line cap=round] ( 81.60, 95.05) -- ( 81.60, 99.89);

\path[draw=drawColor,line width= 0.4pt,line join=round,line cap=round] ( 79.43, 95.05) --
	( 81.60, 95.05) --
	( 83.76, 95.05);

\path[draw=drawColor,line width= 0.4pt,line join=round,line cap=round] ( 83.76, 99.89) --
	( 81.60, 99.89) --
	( 79.43, 99.89);

\path[draw=drawColor,line width= 0.4pt,line join=round,line cap=round] ( 72.82, 88.52) -- ( 72.82, 93.60);

\path[draw=drawColor,line width= 0.4pt,line join=round,line cap=round] ( 70.65, 88.52) --
	( 72.82, 88.52) --
	( 74.99, 88.52);

\path[draw=drawColor,line width= 0.4pt,line join=round,line cap=round] ( 74.99, 93.60) --
	( 72.82, 93.60) --
	( 70.65, 93.60);

\path[draw=drawColor,line width= 0.4pt,line join=round,line cap=round] ( 64.05, 82.60) -- ( 64.05, 87.38);

\path[draw=drawColor,line width= 0.4pt,line join=round,line cap=round] ( 61.88, 82.60) --
	( 64.05, 82.60) --
	( 66.22, 82.60);

\path[draw=drawColor,line width= 0.4pt,line join=round,line cap=round] ( 66.22, 87.38) --
	( 64.05, 87.38) --
	( 61.88, 87.38);

\path[draw=drawColor,line width= 0.4pt,line join=round,line cap=round] ( 55.27, 76.41) -- ( 55.27, 80.58);

\path[draw=drawColor,line width= 0.4pt,line join=round,line cap=round] ( 53.10, 76.41) --
	( 55.27, 76.41) --
	( 57.44, 76.41);

\path[draw=drawColor,line width= 0.4pt,line join=round,line cap=round] ( 57.44, 80.58) --
	( 55.27, 80.58) --
	( 53.10, 80.58);

\path[draw=drawColor,line width= 0.4pt,line join=round,line cap=round] ( 46.50, 71.42) -- ( 46.50, 76.03);

\path[draw=drawColor,line width= 0.4pt,line join=round,line cap=round] ( 44.33, 71.42) --
	( 46.50, 71.42) --
	( 48.66, 71.42);

\path[draw=drawColor,line width= 0.4pt,line join=round,line cap=round] ( 48.66, 76.03) --
	( 46.50, 76.03) --
	( 44.33, 76.03);

\path[draw=drawColor,line width= 0.4pt,line join=round,line cap=round] ( 37.72, 67.85) -- ( 37.72, 73.52);

\path[draw=drawColor,line width= 0.4pt,line join=round,line cap=round] ( 35.55, 67.85) --
	( 37.72, 67.85) --
	( 39.89, 67.85);

\path[draw=drawColor,line width= 0.4pt,line join=round,line cap=round] ( 39.89, 73.52) --
	( 37.72, 73.52) --
	( 35.55, 73.52);
\definecolor{drawColor}{RGB}{114,31,129}

\path[draw=drawColor,line width= 0.4pt,line join=round,line cap=round] (230.77,174.60) --
	(221.99,172.40) --
	(213.22,169.45) --
	(204.44,165.34) --
	(195.67,160.55) --
	(186.89,155.19) --
	(178.12,149.77) --
	(169.35,143.52) --
	(160.57,137.71) --
	(151.79,131.35) --
	(143.02,125.45) --
	(134.25,119.28) --
	(125.47,112.93) --
	(116.70,106.60) --
	(107.92,100.45) --
	( 99.14, 94.20) --
	( 90.37, 87.55) --
	( 81.60, 81.34) --
	( 72.82, 75.12) --
	( 64.05, 68.91) --
	( 55.27, 63.09) --
	( 46.50, 57.38) --
	( 37.72, 52.79);

\path[draw=drawColor,line width= 0.4pt,line join=round,line cap=round] (230.77,174.00) -- (230.77,175.17);

\path[draw=drawColor,line width= 0.4pt,line join=round,line cap=round] (228.60,174.00) --
	(230.77,174.00) --
	(232.94,174.00);

\path[draw=drawColor,line width= 0.4pt,line join=round,line cap=round] (232.94,175.17) --
	(230.77,175.17) --
	(228.60,175.17);

\path[draw=drawColor,line width= 0.4pt,line join=round,line cap=round] (221.99,171.36) -- (221.99,173.36);

\path[draw=drawColor,line width= 0.4pt,line join=round,line cap=round] (219.83,171.36) --
	(221.99,171.36) --
	(224.16,171.36);

\path[draw=drawColor,line width= 0.4pt,line join=round,line cap=round] (224.16,173.36) --
	(221.99,173.36) --
	(219.83,173.36);

\path[draw=drawColor,line width= 0.4pt,line join=round,line cap=round] (213.22,168.31) -- (213.22,170.51);

\path[draw=drawColor,line width= 0.4pt,line join=round,line cap=round] (211.05,168.31) --
	(213.22,168.31) --
	(215.39,168.31);

\path[draw=drawColor,line width= 0.4pt,line join=round,line cap=round] (215.39,170.51) --
	(213.22,170.51) --
	(211.05,170.51);

\path[draw=drawColor,line width= 0.4pt,line join=round,line cap=round] (204.44,163.94) -- (204.44,166.62);

\path[draw=drawColor,line width= 0.4pt,line join=round,line cap=round] (202.28,163.94) --
	(204.44,163.94) --
	(206.61,163.94);

\path[draw=drawColor,line width= 0.4pt,line join=round,line cap=round] (206.61,166.62) --
	(204.44,166.62) --
	(202.28,166.62);

\path[draw=drawColor,line width= 0.4pt,line join=round,line cap=round] (195.67,159.09) -- (195.67,161.88);

\path[draw=drawColor,line width= 0.4pt,line join=round,line cap=round] (193.50,159.09) --
	(195.67,159.09) --
	(197.84,159.09);

\path[draw=drawColor,line width= 0.4pt,line join=round,line cap=round] (197.84,161.88) --
	(195.67,161.88) --
	(193.50,161.88);

\path[draw=drawColor,line width= 0.4pt,line join=round,line cap=round] (186.89,153.70) -- (186.89,156.54);

\path[draw=drawColor,line width= 0.4pt,line join=round,line cap=round] (184.73,153.70) --
	(186.89,153.70) --
	(189.06,153.70);

\path[draw=drawColor,line width= 0.4pt,line join=round,line cap=round] (189.06,156.54) --
	(186.89,156.54) --
	(184.73,156.54);

\path[draw=drawColor,line width= 0.4pt,line join=round,line cap=round] (178.12,148.23) -- (178.12,151.17);

\path[draw=drawColor,line width= 0.4pt,line join=round,line cap=round] (175.95,148.23) --
	(178.12,148.23) --
	(180.29,148.23);

\path[draw=drawColor,line width= 0.4pt,line join=round,line cap=round] (180.29,151.17) --
	(178.12,151.17) --
	(175.95,151.17);

\path[draw=drawColor,line width= 0.4pt,line join=round,line cap=round] (169.35,142.08) -- (169.35,144.84);

\path[draw=drawColor,line width= 0.4pt,line join=round,line cap=round] (167.18,142.08) --
	(169.35,142.08) --
	(171.51,142.08);

\path[draw=drawColor,line width= 0.4pt,line join=round,line cap=round] (171.51,144.84) --
	(169.35,144.84) --
	(167.18,144.84);

\path[draw=drawColor,line width= 0.4pt,line join=round,line cap=round] (160.57,136.24) -- (160.57,139.04);

\path[draw=drawColor,line width= 0.4pt,line join=round,line cap=round] (158.40,136.24) --
	(160.57,136.24) --
	(162.74,136.24);

\path[draw=drawColor,line width= 0.4pt,line join=round,line cap=round] (162.74,139.04) --
	(160.57,139.04) --
	(158.40,139.04);

\path[draw=drawColor,line width= 0.4pt,line join=round,line cap=round] (151.79,129.96) -- (151.79,132.63);

\path[draw=drawColor,line width= 0.4pt,line join=round,line cap=round] (149.63,129.96) --
	(151.79,129.96) --
	(153.96,129.96);

\path[draw=drawColor,line width= 0.4pt,line join=round,line cap=round] (153.96,132.63) --
	(151.79,132.63) --
	(149.63,132.63);

\path[draw=drawColor,line width= 0.4pt,line join=round,line cap=round] (143.02,124.07) -- (143.02,126.72);

\path[draw=drawColor,line width= 0.4pt,line join=round,line cap=round] (140.85,124.07) --
	(143.02,124.07) --
	(145.19,124.07);

\path[draw=drawColor,line width= 0.4pt,line join=round,line cap=round] (145.19,126.72) --
	(143.02,126.72) --
	(140.85,126.72);

\path[draw=drawColor,line width= 0.4pt,line join=round,line cap=round] (134.25,117.99) -- (134.25,120.46);

\path[draw=drawColor,line width= 0.4pt,line join=round,line cap=round] (132.08,117.99) --
	(134.25,117.99) --
	(136.41,117.99);

\path[draw=drawColor,line width= 0.4pt,line join=round,line cap=round] (136.41,120.46) --
	(134.25,120.46) --
	(132.08,120.46);

\path[draw=drawColor,line width= 0.4pt,line join=round,line cap=round] (125.47,111.03) -- (125.47,114.61);

\path[draw=drawColor,line width= 0.4pt,line join=round,line cap=round] (123.30,111.03) --
	(125.47,111.03) --
	(127.64,111.03);

\path[draw=drawColor,line width= 0.4pt,line join=round,line cap=round] (127.64,114.61) --
	(125.47,114.61) --
	(123.30,114.61);

\path[draw=drawColor,line width= 0.4pt,line join=round,line cap=round] (116.70,105.23) -- (116.70,107.86);

\path[draw=drawColor,line width= 0.4pt,line join=round,line cap=round] (114.53,105.23) --
	(116.70,105.23) --
	(118.87,105.23);

\path[draw=drawColor,line width= 0.4pt,line join=round,line cap=round] (118.87,107.86) --
	(116.70,107.86) --
	(114.53,107.86);

\path[draw=drawColor,line width= 0.4pt,line join=round,line cap=round] (107.92, 99.05) -- (107.92,101.73);

\path[draw=drawColor,line width= 0.4pt,line join=round,line cap=round] (105.75, 99.05) --
	(107.92, 99.05) --
	(110.09, 99.05);

\path[draw=drawColor,line width= 0.4pt,line join=round,line cap=round] (110.09,101.73) --
	(107.92,101.73) --
	(105.75,101.73);

\path[draw=drawColor,line width= 0.4pt,line join=round,line cap=round] ( 99.14, 92.68) -- ( 99.14, 95.58);

\path[draw=drawColor,line width= 0.4pt,line join=round,line cap=round] ( 96.98, 92.68) --
	( 99.14, 92.68) --
	(101.31, 92.68);

\path[draw=drawColor,line width= 0.4pt,line join=round,line cap=round] (101.31, 95.58) --
	( 99.14, 95.58) --
	( 96.98, 95.58);

\path[draw=drawColor,line width= 0.4pt,line join=round,line cap=round] ( 90.37, 85.92) -- ( 90.37, 89.01);

\path[draw=drawColor,line width= 0.4pt,line join=round,line cap=round] ( 88.20, 85.92) --
	( 90.37, 85.92) --
	( 92.54, 85.92);

\path[draw=drawColor,line width= 0.4pt,line join=round,line cap=round] ( 92.54, 89.01) --
	( 90.37, 89.01) --
	( 88.20, 89.01);

\path[draw=drawColor,line width= 0.4pt,line join=round,line cap=round] ( 81.60, 79.87) -- ( 81.60, 82.69);

\path[draw=drawColor,line width= 0.4pt,line join=round,line cap=round] ( 79.43, 79.87) --
	( 81.60, 79.87) --
	( 83.76, 79.87);

\path[draw=drawColor,line width= 0.4pt,line join=round,line cap=round] ( 83.76, 82.69) --
	( 81.60, 82.69) --
	( 79.43, 82.69);

\path[draw=drawColor,line width= 0.4pt,line join=round,line cap=round] ( 72.82, 73.45) -- ( 72.82, 76.62);

\path[draw=drawColor,line width= 0.4pt,line join=round,line cap=round] ( 70.65, 73.45) --
	( 72.82, 73.45) --
	( 74.99, 73.45);

\path[draw=drawColor,line width= 0.4pt,line join=round,line cap=round] ( 74.99, 76.62) --
	( 72.82, 76.62) --
	( 70.65, 76.62);

\path[draw=drawColor,line width= 0.4pt,line join=round,line cap=round] ( 64.05, 67.29) -- ( 64.05, 70.37);

\path[draw=drawColor,line width= 0.4pt,line join=round,line cap=round] ( 61.88, 67.29) --
	( 64.05, 67.29) --
	( 66.22, 67.29);

\path[draw=drawColor,line width= 0.4pt,line join=round,line cap=round] ( 66.22, 70.37) --
	( 64.05, 70.37) --
	( 61.88, 70.37);

\path[draw=drawColor,line width= 0.4pt,line join=round,line cap=round] ( 55.27, 61.41) -- ( 55.27, 64.60);

\path[draw=drawColor,line width= 0.4pt,line join=round,line cap=round] ( 53.10, 61.41) --
	( 55.27, 61.41) --
	( 57.44, 61.41);

\path[draw=drawColor,line width= 0.4pt,line join=round,line cap=round] ( 57.44, 64.60) --
	( 55.27, 64.60) --
	( 53.10, 64.60);

\path[draw=drawColor,line width= 0.4pt,line join=round,line cap=round] ( 46.50, 55.72) -- ( 46.50, 58.88);

\path[draw=drawColor,line width= 0.4pt,line join=round,line cap=round] ( 44.33, 55.72) --
	( 46.50, 55.72) --
	( 48.66, 55.72);

\path[draw=drawColor,line width= 0.4pt,line join=round,line cap=round] ( 48.66, 58.88) --
	( 46.50, 58.88) --
	( 44.33, 58.88);

\path[draw=drawColor,line width= 0.4pt,line join=round,line cap=round] ( 37.72, 50.51) -- ( 37.72, 54.77);

\path[draw=drawColor,line width= 0.4pt,line join=round,line cap=round] ( 35.55, 50.51) --
	( 37.72, 50.51) --
	( 39.89, 50.51);

\path[draw=drawColor,line width= 0.4pt,line join=round,line cap=round] ( 39.89, 54.77) --
	( 37.72, 54.77) --
	( 35.55, 54.77);
\definecolor{drawColor}{RGB}{45,17,96}

\path[draw=drawColor,line width= 0.4pt,line join=round,line cap=round] (230.77,164.60) --
	(221.99,159.19) --
	(213.22,153.88) --
	(204.44,147.78) --
	(195.67,141.83) --
	(186.89,135.58) --
	(178.12,129.33) --
	(169.35,123.31) --
	(160.57,116.85) --
	(151.79,110.48) --
	(143.02,104.47) --
	(134.25, 97.84) --
	(125.47, 91.65) --
	(116.70, 85.21) --
	(107.92, 78.79) --
	( 99.14, 72.70) --
	( 90.37, 66.41) --
	( 81.60, 60.10) --
	( 72.82, 53.90) --
	( 64.05, 47.66) --
	( 55.27, 42.16) --
	( 46.50, 37.32) --
	( 37.72, 36.61);

\path[draw=drawColor,line width= 0.4pt,line join=round,line cap=round] (230.77,163.64) -- (230.77,165.51);

\path[draw=drawColor,line width= 0.4pt,line join=round,line cap=round] (228.60,163.64) --
	(230.77,163.64) --
	(232.94,163.64);

\path[draw=drawColor,line width= 0.4pt,line join=round,line cap=round] (232.94,165.51) --
	(230.77,165.51) --
	(228.60,165.51);

\path[draw=drawColor,line width= 0.4pt,line join=round,line cap=round] (221.99,158.04) -- (221.99,160.27);

\path[draw=drawColor,line width= 0.4pt,line join=round,line cap=round] (219.83,158.04) --
	(221.99,158.04) --
	(224.16,158.04);

\path[draw=drawColor,line width= 0.4pt,line join=round,line cap=round] (224.16,160.27) --
	(221.99,160.27) --
	(219.83,160.27);

\path[draw=drawColor,line width= 0.4pt,line join=round,line cap=round] (213.22,152.96) -- (213.22,154.76);

\path[draw=drawColor,line width= 0.4pt,line join=round,line cap=round] (211.05,152.96) --
	(213.22,152.96) --
	(215.39,152.96);

\path[draw=drawColor,line width= 0.4pt,line join=round,line cap=round] (215.39,154.76) --
	(213.22,154.76) --
	(211.05,154.76);

\path[draw=drawColor,line width= 0.4pt,line join=round,line cap=round] (204.44,146.53) -- (204.44,148.94);

\path[draw=drawColor,line width= 0.4pt,line join=round,line cap=round] (202.28,146.53) --
	(204.44,146.53) --
	(206.61,146.53);

\path[draw=drawColor,line width= 0.4pt,line join=round,line cap=round] (206.61,148.94) --
	(204.44,148.94) --
	(202.28,148.94);

\path[draw=drawColor,line width= 0.4pt,line join=round,line cap=round] (195.67,140.70) -- (195.67,142.88);

\path[draw=drawColor,line width= 0.4pt,line join=round,line cap=round] (193.50,140.70) --
	(195.67,140.70) --
	(197.84,140.70);

\path[draw=drawColor,line width= 0.4pt,line join=round,line cap=round] (197.84,142.88) --
	(195.67,142.88) --
	(193.50,142.88);

\path[draw=drawColor,line width= 0.4pt,line join=round,line cap=round] (186.89,134.58) -- (186.89,136.53);

\path[draw=drawColor,line width= 0.4pt,line join=round,line cap=round] (184.73,134.58) --
	(186.89,134.58) --
	(189.06,134.58);

\path[draw=drawColor,line width= 0.4pt,line join=round,line cap=round] (189.06,136.53) --
	(186.89,136.53) --
	(184.73,136.53);

\path[draw=drawColor,line width= 0.4pt,line join=round,line cap=round] (178.12,128.18) -- (178.12,130.40);

\path[draw=drawColor,line width= 0.4pt,line join=round,line cap=round] (175.95,128.18) --
	(178.12,128.18) --
	(180.29,128.18);

\path[draw=drawColor,line width= 0.4pt,line join=round,line cap=round] (180.29,130.40) --
	(178.12,130.40) --
	(175.95,130.40);

\path[draw=drawColor,line width= 0.4pt,line join=round,line cap=round] (169.35,122.30) -- (169.35,124.26);

\path[draw=drawColor,line width= 0.4pt,line join=round,line cap=round] (167.18,122.30) --
	(169.35,122.30) --
	(171.51,122.30);

\path[draw=drawColor,line width= 0.4pt,line join=round,line cap=round] (171.51,124.26) --
	(169.35,124.26) --
	(167.18,124.26);

\path[draw=drawColor,line width= 0.4pt,line join=round,line cap=round] (160.57,115.75) -- (160.57,117.88);

\path[draw=drawColor,line width= 0.4pt,line join=round,line cap=round] (158.40,115.75) --
	(160.57,115.75) --
	(162.74,115.75);

\path[draw=drawColor,line width= 0.4pt,line join=round,line cap=round] (162.74,117.88) --
	(160.57,117.88) --
	(158.40,117.88);

\path[draw=drawColor,line width= 0.4pt,line join=round,line cap=round] (151.79,109.44) -- (151.79,111.46);

\path[draw=drawColor,line width= 0.4pt,line join=round,line cap=round] (149.63,109.44) --
	(151.79,109.44) --
	(153.96,109.44);

\path[draw=drawColor,line width= 0.4pt,line join=round,line cap=round] (153.96,111.46) --
	(151.79,111.46) --
	(149.63,111.46);

\path[draw=drawColor,line width= 0.4pt,line join=round,line cap=round] (143.02,103.57) -- (143.02,105.31);

\path[draw=drawColor,line width= 0.4pt,line join=round,line cap=round] (140.85,103.57) --
	(143.02,103.57) --
	(145.19,103.57);

\path[draw=drawColor,line width= 0.4pt,line join=round,line cap=round] (145.19,105.31) --
	(143.02,105.31) --
	(140.85,105.31);

\path[draw=drawColor,line width= 0.4pt,line join=round,line cap=round] (134.25, 96.81) -- (134.25, 98.80);

\path[draw=drawColor,line width= 0.4pt,line join=round,line cap=round] (132.08, 96.81) --
	(134.25, 96.81) --
	(136.41, 96.81);

\path[draw=drawColor,line width= 0.4pt,line join=round,line cap=round] (136.41, 98.80) --
	(134.25, 98.80) --
	(132.08, 98.80);

\path[draw=drawColor,line width= 0.4pt,line join=round,line cap=round] (125.47, 90.65) -- (125.47, 92.59);

\path[draw=drawColor,line width= 0.4pt,line join=round,line cap=round] (123.30, 90.65) --
	(125.47, 90.65) --
	(127.64, 90.65);

\path[draw=drawColor,line width= 0.4pt,line join=round,line cap=round] (127.64, 92.59) --
	(125.47, 92.59) --
	(123.30, 92.59);

\path[draw=drawColor,line width= 0.4pt,line join=round,line cap=round] (116.70, 84.09) -- (116.70, 86.26);

\path[draw=drawColor,line width= 0.4pt,line join=round,line cap=round] (114.53, 84.09) --
	(116.70, 84.09) --
	(118.87, 84.09);

\path[draw=drawColor,line width= 0.4pt,line join=round,line cap=round] (118.87, 86.26) --
	(116.70, 86.26) --
	(114.53, 86.26);

\path[draw=drawColor,line width= 0.4pt,line join=round,line cap=round] (107.92, 77.74) -- (107.92, 79.78);

\path[draw=drawColor,line width= 0.4pt,line join=round,line cap=round] (105.75, 77.74) --
	(107.92, 77.74) --
	(110.09, 77.74);

\path[draw=drawColor,line width= 0.4pt,line join=round,line cap=round] (110.09, 79.78) --
	(107.92, 79.78) --
	(105.75, 79.78);

\path[draw=drawColor,line width= 0.4pt,line join=round,line cap=round] ( 99.14, 71.69) -- ( 99.14, 73.64);

\path[draw=drawColor,line width= 0.4pt,line join=round,line cap=round] ( 96.98, 71.69) --
	( 99.14, 71.69) --
	(101.31, 71.69);

\path[draw=drawColor,line width= 0.4pt,line join=round,line cap=round] (101.31, 73.64) --
	( 99.14, 73.64) --
	( 96.98, 73.64);

\path[draw=drawColor,line width= 0.4pt,line join=round,line cap=round] ( 90.37, 65.28) -- ( 90.37, 67.47);

\path[draw=drawColor,line width= 0.4pt,line join=round,line cap=round] ( 88.20, 65.28) --
	( 90.37, 65.28) --
	( 92.54, 65.28);

\path[draw=drawColor,line width= 0.4pt,line join=round,line cap=round] ( 92.54, 67.47) --
	( 90.37, 67.47) --
	( 88.20, 67.47);

\path[draw=drawColor,line width= 0.4pt,line join=round,line cap=round] ( 81.60, 59.07) -- ( 81.60, 61.06);

\path[draw=drawColor,line width= 0.4pt,line join=round,line cap=round] ( 79.43, 59.07) --
	( 81.60, 59.07) --
	( 83.76, 59.07);

\path[draw=drawColor,line width= 0.4pt,line join=round,line cap=round] ( 83.76, 61.06) --
	( 81.60, 61.06) --
	( 79.43, 61.06);

\path[draw=drawColor,line width= 0.4pt,line join=round,line cap=round] ( 72.82, 52.82) -- ( 72.82, 54.91);

\path[draw=drawColor,line width= 0.4pt,line join=round,line cap=round] ( 70.65, 52.82) --
	( 72.82, 52.82) --
	( 74.99, 52.82);

\path[draw=drawColor,line width= 0.4pt,line join=round,line cap=round] ( 74.99, 54.91) --
	( 72.82, 54.91) --
	( 70.65, 54.91);

\path[draw=drawColor,line width= 0.4pt,line join=round,line cap=round] ( 64.05, 46.39) -- ( 64.05, 48.82);

\path[draw=drawColor,line width= 0.4pt,line join=round,line cap=round] ( 61.88, 46.39) --
	( 64.05, 46.39) --
	( 66.22, 46.39);

\path[draw=drawColor,line width= 0.4pt,line join=round,line cap=round] ( 66.22, 48.82) --
	( 64.05, 48.82) --
	( 61.88, 48.82);

\path[draw=drawColor,line width= 0.4pt,line join=round,line cap=round] ( 55.27, 40.42) -- ( 55.27, 43.72);

\path[draw=drawColor,line width= 0.4pt,line join=round,line cap=round] ( 53.10, 40.42) --
	( 55.27, 40.42) --
	( 57.44, 40.42);

\path[draw=drawColor,line width= 0.4pt,line join=round,line cap=round] ( 57.44, 43.72) --
	( 55.27, 43.72) --
	( 53.10, 43.72);

\path[draw=drawColor,line width= 0.4pt,line join=round,line cap=round] ( 46.50, 33.93) -- ( 46.50, 40.08);

\path[draw=drawColor,line width= 0.4pt,line join=round,line cap=round] ( 44.33, 33.93) --
	( 46.50, 33.93) --
	( 48.66, 33.93);

\path[draw=drawColor,line width= 0.4pt,line join=round,line cap=round] ( 48.66, 40.08) --
	( 46.50, 40.08) --
	( 44.33, 40.08);

\path[draw=drawColor,line width= 0.4pt,line join=round,line cap=round] ( 37.72, 32.48) -- ( 37.72, 39.84);

\path[draw=drawColor,line width= 0.4pt,line join=round,line cap=round] ( 35.55, 32.48) --
	( 37.72, 32.48) --
	( 39.89, 32.48);

\path[draw=drawColor,line width= 0.4pt,line join=round,line cap=round] ( 39.89, 39.84) --
	( 37.72, 39.84) --
	( 35.55, 39.84);
\definecolor{drawColor}{RGB}{254,175,119}

\path[draw=drawColor,line width= 0.4pt,line join=round,line cap=round] (167.38, 77.70) -- (185.38, 77.70);
\definecolor{drawColor}{RGB}{241,96,93}

\path[draw=drawColor,line width= 0.4pt,line join=round,line cap=round] (167.38, 69.30) -- (185.38, 69.30);
\definecolor{drawColor}{RGB}{182,54,121}

\path[draw=drawColor,line width= 0.4pt,line join=round,line cap=round] (167.38, 60.90) -- (185.38, 60.90);
\definecolor{drawColor}{RGB}{114,31,129}

\path[draw=drawColor,line width= 0.4pt,line join=round,line cap=round] (167.38, 52.50) -- (185.38, 52.50);
\definecolor{drawColor}{RGB}{45,17,96}

\path[draw=drawColor,line width= 0.4pt,line join=round,line cap=round] (167.38, 44.10) -- (185.38, 44.10);
\definecolor{drawColor}{RGB}{0,0,0}

\node[text=drawColor,anchor=base west,inner sep=0pt, outer sep=0pt, scale=  1.00] at (192.58, 74.25) {{\footnotesize $m = 240$}};

\node[text=drawColor,anchor=base west,inner sep=0pt, outer sep=0pt, scale=  1.00] at (192.58, 65.85) {\footnotesize $m = 280$};

\node[text=drawColor,anchor=base west,inner sep=0pt, outer sep=0pt, scale=  1.00] at (192.58, 57.45) {\footnotesize  $m=340$};

\node[text=drawColor,anchor=base west,inner sep=0pt, outer sep=0pt, scale=  1.00] at (192.58, 49.05) {\footnotesize  $m=600$};

\node[text=drawColor,anchor=base west,inner sep=0pt, outer sep=0pt, scale=  1.00] at (192.58, 40.65) {\footnotesize  $m=1800$};
\end{scope}
\begin{scope}
\path[clip] (  0.00,  0.00) rectangle (238.49,180.67);
\definecolor{drawColor}{RGB}{0,0,0}

\node[text=drawColor,anchor=base,inner sep=0pt, outer sep=0pt, scale=  1.00] at (134.25,  2.40) {$\log_{10}(\eta)$};

\node[text=drawColor,rotate= 90.00,anchor=base,inner sep=0pt, outer sep=0pt, scale=  1.00] at (  9.60,105.34) {$\log_{10}(\recerror)$};
\end{scope}
\end{tikzpicture}

%% file: Recovering_quantum_gates_precompiled.bbl
%